\documentclass{article}

\usepackage[numbers,sort&compress]{natbib}
\usepackage[colorlinks=true, allcolors=blue]{hyperref}
\usepackage[utf8]{inputenc}
\usepackage[T1]{fontenc}
\usepackage{algorithmicx}
\usepackage{algcompatible}
\usepackage{url}
\usepackage{booktabs}
\usepackage{array}
\newcolumntype{L}[1]{>{\raggedright\arraybackslash}p{#1}}
\usepackage{fullpage}
\usepackage{amsfonts}
\usepackage{nicefrac}
\usepackage{microtype}
\usepackage{xcolor}
\usepackage{algorithm}
\usepackage{algpseudocode}
\usepackage{float}
\usepackage{tabularx}

\DeclareMathAlphabet{\bdxcal}{U}{BOONDOX-cal}{m}{n}

\usepackage{amsmath}
\usepackage{mathtools}
\usepackage{amssymb}
\usepackage{amsthm}
\usepackage{bbm}
\usepackage{graphicx}
\usepackage{mathrsfs}
\usepackage{braket}

\usepackage{doi}
\usepackage{tikz}
\usetikzlibrary{decorations.pathreplacing,calc}
\usepackage{tcolorbox}
\xdefinecolor{bluebox}{RGB}{0,101,189}
\newtcolorbox{mybox}{colback=bluebox!5!white,colframe=bluebox!75!black}
\usepackage{subcaption}

\usepackage{theoremref}

\newtheorem{thm}{\protect\theoremname}[section]
\theoremstyle{plain}
\newtheorem{lem}[thm]{\protect\lemmaname}
\theoremstyle{plain}
\newtheorem{rem}[thm]{\protect\remarkname}
\theoremstyle{plain}

\theoremstyle{plain}
\newtheorem*{lem*}{\protect\lemmaname}
\theoremstyle{plain}
\newtheorem*{rem*}{\protect\remarkname}
\theoremstyle{plain}
\newtheorem*{thm*}{\protect\theoremname}
\theoremstyle{plain}
\newtheorem{prop}[thm]{\protect\propositionname}
\theoremstyle{plain}
\newtheorem*{prop*}{\protect\propositionname}
\theoremstyle{plain}
\newtheorem*{alg*}{Algorithm}
\theoremstyle{plain}
\newtheorem{cor}[thm]{\protect\corollaryname}

\newtheorem{defn}[thm]{Definition}

\usepackage[USenglish]{babel}
  \providecommand{\corollaryname}{Corollary}
  \providecommand{\lemmaname}{Lemma}
  \providecommand{\propositionname}{Proposition}
  \providecommand{\remarkname}{Remark}
\providecommand{\theoremname}{Theorem}

\usepackage[normalem]{ulem}

\newcommand{\Tr}{\mathrm{Tr}}

\newcommand{\dd}{\mathrm{d}}

\DeclareMathOperator{\Var}{Var}

\title{Improved Adaptive Estimation of Quantum Partition Functions with Heisenberg Scaling}

\author{%
Yufei Wang\textsuperscript{1}\thanks{\href{mailto:yuw@math.ku.dk}{\texttt{yuw@math.ku.dk}}}
\qquad Daniel Stilck Fran\c{c}a\textsuperscript{1}\thanks{\href{mailto:dsfranca@math.ku.dk}{\texttt{dsfranca@math.ku.dk}}}
\qquad Samuel Slezak\textsuperscript{2}\thanks{\href{mailto:samuel.slezak@ens-lyon.fr}{\texttt{samuel.slezak@ens-lyon.fr}}}\\[0.5em]
\small \textsuperscript{1}Department of Mathematical Sciences, University of Copenhagen\\
\small Universitetsparken 5, 2100 Denmark\\[0.25em]
\small \textsuperscript{2}Univ Lyon, ENS Lyon, UCBL, CNRS, Inria, LIP,\\
\small F-69342, Lyon Cedex 07, France
}

\mathtoolsset{showonlyrefs}

\begin{document}
\maketitle

\begin{abstract}
We give quantum algorithms that estimate the log partition function of an
\(n\)-qubit quantum Hamiltonian to additive error
\(\epsilon\) with Heisenberg scaling, while
minimizing calls to thermal-state purification preparation. The input is a block encoding
of \(H\), unitaries preparing thermofield-double states at chosen inverse
temperatures, with their inverses and controlled versions, and an adaptive slowly-varying
cooling schedule.
The system-size bounds below assume
\(0\preceq H\preceq hI\) with \(h=O(n)\) and
block-encoding normalization \(\alpha=\Theta(h)\).
For constant inverse temperature,
we show that such a schedule of length \(O(\sqrt n)\) always exists and can be generated from
overlap estimates with \(\widetilde O(\sqrt n)\) further preparations.

Given such a schedule and classically chosen temperatures, we estimate the
log partition function with \(\widetilde O(n^{5/4}/\epsilon)\) expected calls to
state preparation and to the block encoding, improving by \(n^{1/4}\) on non-adaptive strategies. To achieve our result, we use a recursive-doubling
identity that expresses each schedule increment as a weighted sum of logarithms of
thermofield-double overlaps and of short imaginary-time steps implemented by
quantum singular value transformation combined with previously developed variance reduction techniques. If thermofield doubles can be prepared
controlled on a temperature register, a quantum inverse-binomial log estimator
combined with quantum mean estimation reduces both query counts to
\(\widetilde O(n/\epsilon)\). We show that the latter scaling is optimal in block-encoding queries up to polylogarithmic factors and discuss end-to-end complexities for one-dimensional partition-function estimation.
\end{abstract}

\section{Introduction} \label{sec:intro}

Computing log partition functions is a fundamental problem in statistical
physics and quantum computing.  For an \(n\)-qubit Hamiltonian \(H\) at inverse
temperature \(\beta\), the quantities
\[
    Z(\beta)=\Tr(e^{-\beta H}),
    \qquad f(\beta)=\ln Z(\beta)
\]
determine the free energy and encode equilibrium thermodynamics.  Their
estimation also connects thermal simulation to counting and optimization,
and to free-energy calculations in quantum chemistry
\cite{Stefankovic2009,PoulinWocjan2009,Huang2025Fullqubit}.
Our aim is to achieve Heisenberg scaling
\(\widetilde O(1/\epsilon)\) for general quantum Hamiltonians, i.e. to obtain an estimate of \(f(\beta)\) with additive error \(0<\epsilon\le1\),
equivalently an estimate of \(Z(\beta)\) within a multiplicative factor
\(e^{\pm\epsilon}\), while minimizing dependence on system size and on the resources related to thermal state access.

\medskip

We give quantum algorithms for this problem in two access models. In the first model, a temperature is specified classically, and the
algorithm can call the corresponding thermofield-double preparation
unitary, its inverse, and the required controlled operations.  We estimate
the local terms separately, using nearly unbiased quantum estimators so
that errors accumulate through their variance, building upon the approach of
\cite{Cornelissen2023}.  Optimizing the precision allocated to each term
gives, for a supplied schedule of length
\(\ell_{\mathrm{SV}}=\widetilde O(\sqrt n)\),
\[
    \widetilde O\!\left(\frac{n^{5/4}}{\epsilon}\right)
\]
expected calls to purified-state preparation and to a Hamiltonian block
encoding with normalization \(\alpha=\Theta(h)\), at constant failure
probability (Theorem~\ref{thm:single-temperature-main}).  More generally,
the bound depends on the distribution of interval lengths, rather than
only their number.  This gives a Heisenberg-precision reconstruction for
arbitrary real spectra using coherent access at separately chosen
temperatures.

\medskip

The second model allows temperatures and reconstruction terms to be
selected coherently.  We use the classical inverse-binomial idea of
estimating a log probability from a harmonic function of the waiting time
to the first success~\cite{VanOpheusden2020} to build a quantum version of
this estimator.  Coherent truncation and reversible evaluation give a
logarithmic sample with controlled bias and bounded second moment.
Preparing a superposition over the reconstruction terms then lets us
estimate their weighted sum with one quantum mean-estimation procedure.
Under this stronger access model, for the same supplied schedule of length
\(\ell_{\mathrm{SV}}=\widetilde O(\sqrt n)\), the expected reconstruction
query counts improve to
\[
    \widetilde O\!\left(\frac{n}{\epsilon}\right)
\]
for both coherent-temperature preparation and Hamiltonian block-encoding
access (Theorem~\ref{thm:coherent-indexed-main}). Coherent access to
temperature-labeled preparations can carry substantial implementation
overhead: a direct selector has a gate cost that grows with the number of
required temperatures. A common circuit whose parameters are computed
reversibly from the temperature label provides a sufficient route to a
gate-complexity advantage, when selection and parameter-loading costs are
small; see Section~\ref{sec:end-to-end}.

\medskip

Quantum algorithms for partition-function estimation have developed along several lines.
Annealing methods estimate successive partition-function ratios using
thermal states, whereas direct methods estimate traces of approximations to
the exponential operator; hybrid approaches offer further alternatives
\cite{Wocjan2009Partition,PoulinWocjan2009,ChowdhurySommaSubasi2021,WuWang2022}.
We focus on algorithms with access to purified Gibbs states.  Preparing thermal
states can itself be the dominant computational task, and its cost depends
strongly on the Hamiltonian and temperature
\cite{PoulinWocjan2009,ChenKastoryanoBrandaoGilyen2023,RouzeFrancaAlhambra2024}.
We therefore separate the number of calls to state-preparation unitaries and
Hamiltonian block encodings from their circuit implementations.  This
distinction also makes the access assumptions essential: a supply of mixed
Gibbs states, coherent access to a purification at a classically chosen
temperature, and preparation controlled by a quantum temperature register
are different resources.

\medskip

For classical systems, adaptive cooling schedules and bounded-variance
ratio estimators form the basis of efficient counting algorithms
\cite{Stefankovic2009,Huber2015,Kolmogorov2017}.
Quantum algorithms build on this framework using coherent encodings of
classical Gibbs distributions, often called quantum samples
\cite{Wocjan2009Partition,Montanaro2015,Harrow2020,Arunachalam2022,Cornelissen2023}.
These works establish quadratic improvements in precision dependence,
non-destructive estimation and state reuse, and efficient construction of
slowly-varying schedules.  Cornelissen and Hamoudi~\cite{Cornelissen2023}
further obtain sublinear dependence on the logarithm of the configuration
space size by combining nearly unbiased estimates through their variance.
These results exploit classical energy-evaluation access; several are
formulated for nonnegative integer-valued energies, as arise naturally in
counting problems.  Their algorithms and bounds therefore do not transfer directly to
an arbitrary quantum Hamiltonian supplied through a block encoding.

\medskip

General quantum partition-function estimation has a substantial history
\cite{PoulinWocjan2009,ChowdhurySommaSubasi2021,Matsumoto2022}.
The issue here is how to retain favorable precision and system-size scaling
without assuming ingredients that are naturally present classically, like an integer spectrum or efficient evaluation of individual
energies.  Indeed, consider, e.g., the approach of~\cite{Arunachalam2022}, which uses the quantum samples to put energy samples into bins as an integral part of the estimator. Resolving sharp energy bins for a quantum Hamiltonian can require
additional spectral precision, especially when eigenvalues lie close to bin
boundaries.  To overcome this issue, we instead use smooth quantities: overlaps of canonical
purified Gibbs states, or thermofield doubles, and short imaginary-time
reweighting expectations.  These quantities can be implemented without
resolving individual energy levels, using purified-state access and quantum
singular value transformation~\cite{Gilyen2019QSVT}.

\medskip

Our reconstruction applies to any known shifted energy window
\(0\preceq H\preceq hI\).  To state the system-size benchmarks,
we take an extensive bound \(h=O(n)\), conventionally
\(h=n\), and constant inverse temperature \(\beta=\Theta(1)\).  This includes
bounded-strength local Hamiltonians with \(O(n)\) interaction terms, after a
known scalar shift.  Bounds for general \(h\) and \(\beta\) are retained in
the main theorems.

\medskip

The common starting point is a slowly-varying cooling schedule, i.e. a sequence of inverse temperatures along which
neighboring thermofield doubles have constant overlap.  Such schedules have
long been central to classical and quantum annealing
\cite{Stefankovic2009,BoixoKnillSomma2009,Harrow2020,Arunachalam2022}.
In the constant-temperature regime, schedules with
\(O(\sqrt n)\) length exist for all Hamiltonians; see
Section~\ref{sec:adaptive-cooling-schedules}.  We use
overlap-based adaptive schedule generation and analyze the reconstruction
cost separately from the cost of obtaining the schedule.

\medskip

On each schedule interval, a recursive-doubling identity similar to one observed in~\cite{Matsumoto2022} expresses the
change in \(f\) as a weighted sum of logarithms of overlap and reweighting
probabilities.  Here we combine this identity with
short reweighting steps and a slowly-varying schedule so that all elementary
probabilities stay uniformly bounded away from zero.  Our contribution is the
resulting reconstruction and its variance-sensitive quantum estimation,
with two implementations distinguished by their access requirements.

\medskip

A matching lower bound explains the linear system-size dependence in the
stronger model compared to the sublinear results for classical systems of~\cite{Cornelissen2023}.  A scalar shift \(H\mapsto H+cI\) leaves every Gibbs state
unchanged but changes \(f(\beta)\) by \(-\beta c\).  When this offset is
accessible only through an \(\alpha\)-normalized Hamiltonian block encoding,
estimating the absolute log partition function requires
\(\Omega(\alpha\beta/\epsilon)\) block-encoding queries in the worst case
for \(0<\epsilon\le\beta\min\{h,\alpha\}/16\),  even with free coherent access to thermofield doubles at every temperature
(Theorem~\ref{thm:linear-lower-bound}).  Thus, for
\(\alpha=\Theta(h)=\Theta(n)\) and constant \(\beta\), the coherent-temperature
algorithm is optimal up to logarithmic factors in this query model. The
separately chosen-temperature algorithm is within a factor
\(\widetilde O(n^{1/4})\) of this benchmark.

\medskip

Finally, we translate the oracle bounds into circuit resource estimates
for one-dimensional systems with efficient preparation
\cite{bergamaschi2026fastmixingquantumspin}.
We also explain how finite-temperature estimates yield ground-state
degeneracy estimates when the ground energy and a suitable excitation-gap
bound are available.

\medskip

A direct fixed-temperature benchmark combines coherent adaptations of the
uniform-grid ratio estimator of \citet{RouzeFrancaAlhambra2024} with
nearly unbiased estimation~\cite{Cornelissen2023}.
For extensive Hamiltonians at constant inverse temperature and fixed failure
probability, our short-schedule reconstruction improves its
\(\widetilde O(n^{3/2}/\epsilon)\) query cost to
\(\widetilde O(n^{5/4}/\epsilon)\).
Coherent-temperature access gives \(\widetilde O(n/\epsilon)\), also
on a uniform grid when its temperatures are supported at comparable
per-call cost. Remark~\ref{rem:uniform-grid-comparison} derives these
benchmarks and explains their relation to the original sampling bound.

\section{Problem Setup and Main Results}

We study partition-function estimation for quantum Hamiltonians with
real-valued spectra. We first state the problem and the two thermal-state
access models, then present query-complexity upper bounds for both algorithms
and a lower bound for block-encoding queries. We next explain the common algorithmic framework, followed by
the cooling-schedule guarantees and applications.

\medskip

\subsection{Estimation problem and access models}

Let \(H=\sum_a E_a\ket{a}\!\bra{a}\) be an \(n\)-qubit Hamiltonian
with \(0\preceq H\preceq hI\), where \(h>0\) is known.
Section~\ref{sec:technical-overview} explains this normalization and
its effect on the partition function.
For a finite inverse temperature \(\beta\geq0\), define
\[
Z_H(\beta):=\Tr(e^{-\beta H}),
\qquad
f_H(\beta):=\ln Z_H(\beta),
\qquad
\rho_\beta^H:=\frac{e^{-\beta H}}{Z_H(\beta)}.
\]
We omit the subscript or superscript \(H\) when the Hamiltonian is
clear from context. Our goal is to solve the following estimation
problem.

\begin{defn}[Partition-function estimation]
Given a Hamiltonian \(H\), a finite inverse temperature
\(\beta\geq0\), a desired precision \(0<\epsilon\leq1\), and a
failure probability \(0<p<1\), output an estimate
\(\widehat f(\beta)\) satisfying
\[
\bigl|\widehat f(\beta)-f(\beta)\bigr|\leq\epsilon
\]
with probability at least \(1-p\).
Equivalently, \(\widehat Z(\beta):=e^{\widehat f(\beta)}\)
satisfies
\[
e^{-\epsilon}Z(\beta)
\leq\widehat Z(\beta)
\leq e^\epsilon Z(\beta)
\]
with the same success probability.
\end{defn}

\medskip
\medskip

The thermal states used by our algorithms are the canonical
\emph{thermofield-double} (TFD) purifications
\begin{equation}
\label{eq:results-tfd-state}
\ket{\psi_\theta}
:=
\frac{1}{\sqrt{Z(\theta)}}
\sum_a e^{-\theta E_a/2}\ket{a}\otimes\ket{a^*},
\end{equation}
where \(\ket{a^*}\) is the complex conjugate of \(\ket a\) in the
computational basis. We distinguish two ways of accessing these states.

\begin{itemize}
\item \textbf{Fixed-temperature coherent access.}
For each classically specified inverse temperature \(\theta\), we
have a preparation unitary and its inverse satisfying
\begin{equation}
\label{eq:results-tfd-prep-U}
U_\theta\ket0=\ket{\psi_\theta}.
\end{equation}
The algorithm chooses the temperature classically, while its
estimation subroutines use the preparation coherently.
The precise resource model is given in
Definition~\ref{def:partition-function-resource-model}.

\item \textbf{Coherent-temperature access.}
The inverse temperature may itself be stored in a quantum register.
We have a parameter-controlled unitary and its inverse satisfying
\begin{equation}
\label{eq:results-coherent-tfd-prep}
U_{\mathrm{cohTFD}}\ket\theta\ket0
=\ket\theta\ket{\psi_\theta}.
\end{equation}
This model also supplies coherent preparation of the distribution
over reconstruction terms and reversible loading of their parameters,
together with the inverse operations, as specified in
Definition~\ref{def:coherent-temperature-access}.
\end{itemize}

\medskip

Both models use an \(\alpha\)-block encoding \(U_H\) of the
Hamiltonian: a unitary whose block on an all-zero ancilla register
equals \(H/\alpha\).
We take \(\alpha=\Theta(h)\) and assume the inverse and controlled
operations required by the algorithms. Let \(N_{\mathrm{TFD}}\),
\(N_{\mathrm{cohTFD}}\), and \(N_{\mathrm{BE}}\) count calls to the
fixed-temperature preparation, coherent-temperature preparation,
and block encoding, respectively, including their required inverse
and controlled uses. These are query counts.

\medskip

Both models reconstruct \(f(\beta)\) using a supplied cooling
schedule \(0=\beta_0<\cdots<\beta_\ell=\beta\), where \(\ell\)
is the number of intervals. We call it \emph{\(C\)-slowly varying} if
\(\lvert\langle\psi_{\beta_i}|\psi_{\beta_{i+1}}\rangle\rvert^2
\geq1/C\) for every interval. Throughout the reconstruction results,
\(C=O(1)\). Section~\ref{sec:adaptive-cooling-schedules}
gives an adaptive construction of such a schedule.
The notation \(\widetilde O\) suppresses polylogarithmic factors,
with the detailed dependences specified in the formal theorems.

\medskip

With fixed-temperature access, we choose a temperature classically
and call its TFD preparation circuit.
Coherent-temperature access prepares the matching TFD state in each
branch of a temperature superposition, preserving coherence.
Our first algorithm estimates reconstruction terms separately and
forms their weighted sum. Our second prepares term labels in
superposition, loads their parameters coherently, and estimates the
sum directly by quantum mean estimation.
For the same supplied schedule, the latter improves our query upper
bound by a factor of \(\sqrt{\ell}\), up to polylogarithmic factors.

\medskip

\subsection{Query complexity and optimality}

We now compare the query complexity of reconstruction under the
two access models. The following informal statements summarize
Theorems~\ref{thm:single-temperature-main}
and~\ref{thm:coherent-indexed-main}.

\begin{thm}[Fixed-temperature reconstruction, informal]
\label{thm:fixed-temp-per-term-results}
Let \(0\preceq H\preceq hI\), and assume an \(\alpha\)-block
encoding with \(\alpha=\Theta(h)\) and fixed-temperature coherent
access. Given a supplied \(C\)-slowly-varying schedule of length
\(\ell\) to a finite \(\beta>0\), with \(C=O(1)\),
Algorithm~\ref{alg:finite-temperature-reconstruction} estimates
\(f(\beta)\) to additive error \(\epsilon\), with failure
probability at most \(p\), using expected query counts
\[
N_{\mathrm{TFD}},N_{\mathrm{BE}}
=\widetilde O\!\left(
\frac{\sqrt{\ell}(\ell+h\beta)\ln(2/p)}{\epsilon}
\right).
\]
\end{thm}

\medskip

Coherent-temperature access allows us to sample reconstruction
terms in superposition and estimate their weighted sum directly.

\begin{thm}[Coherent-temperature reconstruction, informal]
\label{thm:results-indexed}
Let \(0\preceq H\preceq hI\), and assume an \(\alpha\)-block
encoding with \(\alpha=\Theta(h)\) and the coherent-temperature
access of Definition~\ref{def:coherent-temperature-access}.
Given a supplied \(C\)-slowly-varying schedule of length \(\ell\)
to a finite \(\beta>0\), with \(C=O(1)\),
Algorithm~\ref{alg:sampled-dyadic-reconstruction} estimates
\(f(\beta)\) to additive error \(\epsilon\), with failure
probability at most \(p\), using expected query counts
\[
N_{\mathrm{cohTFD}},N_{\mathrm{BE}}
=\widetilde O\!\left(
\frac{(\ell+h\beta)\ln(2/p)}{\epsilon}
\right).
\]
\end{thm}

\medskip

These asymptotic query bounds also cover schedule generation
when the same TFD oracles can also prepare states at the temperatures
needed to generate the schedule; see Section~\ref{sec:adaptive-cooling-schedules}.

\medskip

The fixed-temperature reconstruction bound contains an additional
factor of \(\sqrt\ell\) from estimating the logarithmic terms
separately with coefficient-dependent depths. Coherent branch
sampling allows quantum mean estimation to estimate the weighted
sum without this factor in the query bound.
In the extensive constant-temperature regime
\(h=O(n)\), \(\beta=\Theta(1)\), and
\(\ell=\widetilde O(\sqrt n)\), the two bounds become
\[
\widetilde O\!\left(\frac{n^{5/4}\ln(2/p)}{\epsilon}\right)
\qquad\text{and}\qquad
\widetilde O\!\left(\frac{n\ln(2/p)}{\epsilon}\right).
\]

\medskip

Efficient preparation at each fixed temperature does not automatically
give efficient coherent-temperature preparation.
Suppose we have \(m\) separate preparation circuits, one for each
temperature. An index register identifies the temperature in each
branch. A direct coherent implementation applies all \(m\) circuits
in sequence, with each circuit acting only on the branch for its
temperature.
The total gate count is the sum of all controlled circuit sizes,
plus the cost of the index controls.
These controls can be implemented with \(O(m)\) additional gates
using unary iteration
\cite[Section~III~A]{Babbush2018LinearT}.
If the preparation circuits have similar sizes, this direct
construction uses roughly \(m\) times as many gates as one
fixed-temperature preparation.
This overhead can cancel the benefit of fewer queries.

\medskip

A single preparation circuit with temperature as an input can avoid this overhead by computing its parameters
coherently from the temperature input.
Section~\ref{sec:end-to-end} gives conditions under which the query
improvement carries over to gate-count bounds.
Our spin-chain construction achieves this using a common
quasi-adiabatic preparation path and temperature-controlled
subroutines
\cite{bergamaschi2026fastmixingquantumspin,chen2023efficient}.
In this application, the full gate-count upper bound decreases from
\(\widetilde O(n^{9/4}\ln(2/p)/\epsilon)\)
for fixed-temperature reconstruction to
\(\widetilde O(n^2\ln(2/p)/\epsilon)\)
for coherent-temperature reconstruction.
These bounds include auxiliary operations, so the coherent method
retains an \(n^{1/4}\) improvement in the gate-count upper bounds,
up to polylogarithmic factors.
Proposition~\ref{prop:coherent-tfd-spin-chains} gives the coherent
TFD resource bounds, and
Corollary~\ref{cor:spin-chain-circuit-comparison} gives the resulting
reconstruction costs.
The construction and auxiliary-cost analysis are provided in
Appendix~\ref{app:coherent-temperature-circuits}.

\medskip

The following informal lower-bound statement summarizes
Theorem~\ref{thm:linear-lower-bound}.

\begin{thm}[Hidden energy-offset lower bound, informal]
\label{thm:results-lower-bound}
In the hidden-offset oracle model of
Theorem~\ref{thm:linear-lower-bound}, there are Hamiltonians
\(H_c=A+cI\) satisfying \(0\preceq H_c\preceq hI\), with
\(A\) fixed and known, for which estimating \(\ln Z_{H_c}(\beta)\)
to additive error \(\epsilon\) with success probability at least
\(2/3\) requires
\[
\Omega\!\left(\frac{\alpha\beta}{\epsilon}\right)
\]
block-encoding queries in the worst case, including for algorithms
with variable stopping time, where the cost is measured in expectation.
Here \(0<\epsilon\leq\beta\min\{h,\alpha\}/16\), and the
\(\alpha\)-block encoding is the only oracle that depends on
the unknown offset \(c\). TFD preparation may be supplied for free,
provided its full unitary implementation and all other auxiliary
operations are independent of \(c\).
\end{thm}

\medskip

For \(\alpha=\Theta(h)\), \(h\beta\geq1\),
\(\ell=O(1+h\beta)\), fixed \(p\leq1/3\), and the
accuracy range above, the coherent reconstruction therefore matches
the lower bound up to polylogarithmic factors; see
Corollary~\ref{cor:coherent-block-query-optimality}.
In particular, extensive normalization \(h=\Theta(n)\) and constant
\(\beta\) give a matching \(\widetilde\Theta(n/\epsilon)\)
block-encoding query bound.

\medskip

Classical and quantum product estimators usually assume energy
values and Gibbs-sampling or reflection access for classical
configurations. Our algorithms use TFD preparation and a block encoding
of a quantum Hamiltonian. Table~\ref{tab:prior-work} in the appendix
compares their guarantees and resource units.

\subsection{Algorithmic idea and reconstruction framework}

A cooling schedule expresses the target as a telescopic sum,
\begin{equation}
\label{eq:results-telescopic-sum}
f(\beta)
=f(0)+\sum_{i=0}^{\ell-1}
\bigl(f(\beta_{i+1})-f(\beta_i)\bigr),
\end{equation}
where \(f(0)=n\ln2\) is known exactly.
For \(\theta,s\geq0\), define the overlap probability and
reweighting ratio
\begin{equation}
\label{eq:results-O-R}
O(\theta,s):=\lvert\langle\psi_\theta|\psi_{\theta+s}\rangle\rvert^2,
\qquad
R(\theta,s):=\Tr(\rho_\theta e^{-sH}),
\end{equation}
and write
\begin{equation}
\label{eq:results-A-and-B}
A(\theta,s):=-\ln O(\theta,s),
\qquad
B(\theta,s):=\ln R(\theta,s).
\end{equation}
The logarithm of \(R\) gives an increment of \(f\), while the
negative logarithm of \(O\) gives a second finite difference.

\medskip

The slowly-varying condition controls overlaps but need not keep
\(R\) bounded away from zero over a full schedule interval.
The \emph{recursive-doubling decomposition} starts with a short
reweighting step and recovers the full increment using overlap
corrections:
\begin{equation}
\label{eq:results-dyadic-reconstruction}
f(\beta_{i+1})-f(\beta_i)
=2^{k_i}B(\beta_i,\delta_i)
+\sum_{j=1}^{k_i}2^{k_i-j}A(\beta_i,2^j\delta_i).
\end{equation}
Here \(k_i\geq0\) is an integer and
\(2^{k_i}\delta_i=\beta_{i+1}-\beta_i\), with
\(h\delta_i\leq1\). The step lengths double at each stage.
Section~\ref{sec:technical-overview} proves this identity and shows
that all probabilities appearing in it are bounded below by a
positive constant. It also implements
reweighting by quantum singular value transformation (QSVT) and
specifies the known normalization correction.

\medskip

Taking the logarithm of an amplitude estimate can introduce bias.
The fixed-temperature method controls each term's bias and variance
using the nearly unbiased amplitude estimator of
Cornelissen and Hamoudi~\cite{Cornelissen2023}, then allocates
estimation depths according to the reconstruction coefficients.
The coherent-temperature method prepares branches in superposition,
with probabilities proportional to their coefficients.
In each branch, it performs a fixed number of trials, counts the
failures before the first success (or all trials if none succeeds),
and converts this count to a harmonic number.
The mean of this sample approximates the negative logarithm
of the branch success probability.
We then construct a random variable \(X\) whose mean approximates
the reconstruction sum \(f(\beta)-f(0)\).
Quantum mean estimation~\cite{Montanaro2015} estimates this mean.

\medskip

The flowchart summarizes the construction. Its first step generates
the schedule; the reconstruction bounds above assume that this
schedule is supplied and count its generation separately.

\begin{center}
\setlength{\fboxsep}{6pt}
\begin{tabular}{c}

\fbox{
\begin{minipage}{0.92\textwidth}
\textbf{Generate a slowly-varying cooling schedule}
\smallskip

Construct a schedule
\[
    0=\beta_0<\beta_1<\cdots<\beta_\ell=\beta
\]
such that the overlap between TFD states at consecutive temperatures
is bounded below by a constant.
This keeps the logarithms of the overlap probabilities well
conditioned and decomposes the target quantity as
\[
    f(\beta)
    =
    f(0)+
    \sum_{i=0}^{\ell-1}
    \bigl(f(\beta_{i+1})-f(\beta_i)\bigr).
\]
The schedule-generation procedure is given in
Algorithm~\ref{alg:quantum-cooling-schedule-generation}.
\end{minipage}
}

\\[1em]
\(\downarrow\)
\\

\fbox{
\begin{minipage}{0.92\textwidth}
\textbf{Recursive-doubling decomposition:
express the target as a weighted logarithmic sum}
\smallskip

For each schedule interval, choose \(k_i\) and \(\delta_i\) such that
\[
    2^{k_i}\delta_i=\beta_{i+1}-\beta_i,
    \qquad
    h\delta_i\leq1.
\]
The recursive-doubling identity then gives
\begin{equation}
\label{eq:results-full-dyadic-reconstruction}
f(\beta)
=
f(0)
+
\sum_{i=0}^{\ell-1}
\left[
    2^{k_i}B(\beta_i,\delta_i)
    +
    \sum_{j=1}^{k_i}
    2^{k_i-j}A(\beta_i,2^j\delta_i)
\right].
\end{equation}
Thus \(f(\beta)\) is reduced to a weighted sum of logarithms of
projection probabilities together with the known reweighting
correction.
The slowly-varying condition bounds all overlap probabilities away
from zero, while \(h\delta_i\leq1\) gives an analogous constant lower
bound for the probabilities defining \(B(\beta_i,\delta_i)\).
\end{minipage}
}

\\[1em]
\(\downarrow\)
\\

\fbox{
\begin{minipage}{0.92\textwidth}
\textbf{Estimate the weighted logarithmic sum}
\smallskip

\textit{Fixed-temperature coherent access:}
Algorithm~\ref{alg:finite-temperature-reconstruction}
treats the branches separately, calling
Algorithm~\ref{alg:local-log-estimate} to estimate each logarithmic
quantity \(A(\beta_i,2^j\delta_i)\) or \(B(\beta_i,\delta_i)\).
The reconstruction algorithm chooses the estimation depths
according to the coefficients and combines the estimates using
Eq.~\eqref{eq:results-full-dyadic-reconstruction}.
The total complexity is given in
Theorem~\ref{thm:fixed-temp-per-term-results}.

\medskip

\textit{Coherent-temperature access:}
Algorithm~\ref{alg:sampled-dyadic-reconstruction}
prepares the branch index in superposition, with probabilities
proportional to the reconstruction coefficients.
It uses the selected branch to generate a scaled logarithmic
sample whose expectation approximates the full weighted sum in
Eq.~\eqref{eq:results-full-dyadic-reconstruction}.
The algorithm prepares this sample coherently and applies
quantum mean estimation to estimate its expectation.
The total complexity is given in
Theorem~\ref{thm:results-indexed}.
\end{minipage}
}

\\[1em]
\(\downarrow\)
\\

\fbox{
\begin{minipage}{0.92\textwidth}
\textbf{Output}
\smallskip

Since \(f(0)=n\ln2\) is known exactly, output
\(\widehat f(\beta)\) satisfying
\[
    |\widehat f(\beta)-f(\beta)|\leq\epsilon
\]
with probability at least \(1-p\).
Equivalently, \(\widehat Z(\beta):=e^{\widehat f(\beta)}\) satisfies
\[
    e^{-\epsilon}Z(\beta)
    \leq
    \widehat Z(\beta)
    \leq
    e^\epsilon Z(\beta).
\]
\end{minipage}
}

\end{tabular}
\end{center}

\medskip

We next state how the schedule can be obtained, then describe two
applications of the reconstruction.

\subsection{Cooling schedules for real-valued spectra}

The overlap-based procedure of
Arunachalam et al.~\cite[Theorem~3.4]{Arunachalam2022} generates
the required schedules. Adapting existing schedule-length analyses
\cite{Stefankovic2009,Arunachalam2022} to real-valued spectra
gives the following bounds; see
Theorems~\ref{thm:adaptive_schedule_existence}
and~\ref{thm:quantum-schedule-arunachalam} and
Section~\ref{sec:adaptive-cooling-schedules}.

\begin{thm}[Adaptive cooling schedules, informal]
\label{thm:results-schedule}
Let \(0\preceq H\preceq hI\), and assume fixed-temperature TFD
access at every adaptively queried temperature.
For any finite \(\beta>0\), the overlap-based procedure constructs,
with probability at least \(1-p\), a \(15\)-slowly-varying schedule
from zero to \(\beta\) of length
\[
\ell=O\!\left(1+\sqrt{h\beta}\right).
\]
It uses \(\widetilde O(\ell\ln(2+h\beta)\ln(2/p))\) TFD queries
and no Hamiltonian block-encoding queries. The length guarantee
requires no spectral-gap assumption.

\medskip

For the low-temperature specialization, suppose \(H\) is
nonconstant. Let \(E_0\) be its ground energy, \(N_0\) its ground
degeneracy, \(\bar H:=H-E_0I\), and \(E_s>0\) the smallest
positive eigenvalue of \(\bar H\).
If \(\ln N_0+2<n\ln2\), define the finite endpoint \(\gamma\)
by \(\ln\Tr(e^{-\gamma\bar H})=\ln N_0+2\).
At this endpoint the schedule length also satisfies
\[
\ell=O\!\left(
1+\sqrt{n\ln\!\left(\frac{e\|\bar H\|_\infty}{E_s}\right)}
\right).
\]
Here \(\|\bar H\|_\infty\) is the operator norm.
If \(\ln N_0+2\geq n\ln2\), one interval already gives an
\(e^2\)-slowly-varying schedule to any finite endpoint.
\end{thm}

\medskip

The first bound gives \(\ell=O(\sqrt n)\) when
\(h=O(n)\) and \(\beta=O(1)\), without any spectral-gap assumption.
For low temperatures, Remark~\ref{thm:adaptive_schedule_tight}
gives a family with \(E_s=2^{-n}\) and
\(\|\bar H\|_\infty=\Theta(1)\).
At the specified endpoint, every \(e^2\)-slowly-varying schedule
needs \(\Omega(n)\) intervals, and a schedule with \(O(n)\)
intervals exists. The low-temperature bound is therefore tight
for this family.

\medskip

The generation algorithm is given the final inverse temperature.
The parameters \(E_0,N_0,E_s\) are used only in the length analysis
and are not required as inputs.
For \(0<\epsilon\le1\), including schedule generation does not
change the asymptotic reconstruction query bounds, provided the
same TFD oracles can also prepare states at the temperatures
needed to generate the schedule.

\subsection{Applications}
\label{subsec:ground-degeneracy}

A sufficiently large finite inverse temperature suppresses excited
states. A gap lower bound and a ground-energy estimate then allow
us to recover the ground-state degeneracy as follows.

\begin{cor}[Ground-state degeneracy from finite temperature]
\label{prop:ground-degeneracy}
Let \(E_0\) be the ground energy of \(H\), let
\(N_0:=\dim\ker(H-E_0I)\), and suppose that a known
\(\Delta_{\mathrm{lb}}>0\) satisfies
\(E-E_0\geq\Delta_{\mathrm{lb}}\) for every eigenvalue \(E>E_0\).
For \(0<\epsilon\leq1\), choose a finite inverse temperature
\begin{equation}
\label{eq:ground-degeneracy-beta}
\beta\geq\frac{n\ln2+\ln(6/\epsilon)}{\Delta_{\mathrm{lb}}}.
\end{equation}
If estimates \(\widehat f(\beta)\) and \(\widehat E_0\) satisfy
\[
|\widehat f(\beta)-f(\beta)|\leq\frac{\epsilon}{6},
\qquad
|\widehat E_0-E_0|\leq\frac{\epsilon}{6\beta},
\]
then
\begin{equation}
\label{eq:ground-degeneracy-estimator}
\widehat N_0:=\exp\!\left(\widehat f(\beta)+\beta\widehat E_0\right)
\end{equation}
satisfies \(|\widehat N_0-N_0|\leq\epsilon N_0\).
If each estimation guarantee fails with probability at most \(p/2\),
the relative-error guarantee holds with probability at least \(1-p\).
\end{cor}

\medskip

The proof and resource bounds are given in
Appendix~\ref{app:ground-state-degeneracy},
Theorem~\ref{thm:ground-state-degeneracy}
and Corollary~\ref{cor:ground-state-degeneracy-cost}.
The gap sets the required temperature, which in turn sets the
ground-energy precision. The cost of estimating the ground energy
is additional to the reconstruction cost.

\medskip
As a second application, we consider partition-function estimation
for one-dimensional spin chains.
For one-dimensional systems with uniformly bounded finite-range
interactions and open boundaries, canonical TFD preparation has
polylogarithmic depth at temperatures in a fixed bounded interval;
see \cite[Corollary~I.1]{bergamaschi2026fastmixingquantumspin}.
Under all implementation and accuracy assumptions of
Corollary~\ref{4.1}, a supplied short schedule therefore gives
expected reconstruction depth
\(\widetilde O(n^{5/4}\ln(2/p)/\epsilon)\) for fixed-temperature
access, or \(\widetilde O(n\ln(2/p)/\epsilon)\) with an efficient
coherent-temperature interface.
Proposition~\ref{prop:coherent-tfd-spin-chains} supplies the
coherent-temperature interface for this application.

\medskip

Our main technical contribution is the reconstruction framework
developed in Section~\ref{sec:technical-overview}.
It forms the foundation of both algorithms.
Sections~\ref{sec:single-temperature}
and~\ref{sec:sampled-dyadic} build on this framework to give
the two estimation methods and their query-complexity bounds.

\section{Common Reconstruction Framework}
\label{sec:technical-overview}

Starting from a given cooling schedule, we express
\(f(\beta)-f(0)\) as a weighted sum of logarithms of quantum
projection probabilities. We identify these probabilities, derive
the recursive-doubling identity, and prove that its probabilities
are bounded away from zero. This common reconstruction underlies
both estimation algorithms in Sections~\ref{sec:single-temperature}
and~\ref{sec:sampled-dyadic}.

\medskip

The assumption \(0\preceq H\preceq hI\) does not require the
exact ground energy. Given a known lower bound
\(L\leq\lambda_{\min}(H_{\mathrm{orig}})\), set
\(H:=H_{\mathrm{orig}}-LI\) and choose \(h>0\) so that
\(H\preceq hI\). The normalized Gibbs and canonical TFD states
are unchanged, while
\[
f_{H_{\mathrm{orig}}}(\beta)=f_H(\beta)-\beta L.
\]
The shifted Hamiltonian makes \(e^{-sH}\) a contraction and gives
\(e^{-hs}\leq\Tr(\rho_\beta e^{-sH})\leq1\) for \(s\geq0\).
In particular, \(hs\leq1\) bounds the reweighting ratio below by
\(e^{-1}\). A looser lower bound may increase \(h\) and hence the
resource cost.

\subsection{TFD states, overlaps, and reweighting ratios}
\label{subsec:tfd-local-quantities}

We begin with the purified Gibbs states used throughout the
reconstruction.

\begin{defn}[Thermofield-double state]
\label{def:tfd-state}
Let \(H=\sum_a E_a|a\rangle\!\langle a|\) be the spectral
decomposition of \(H\). Let \(\mathcal H^*\) be the conjugate
purifying space and let \(|a^*\rangle\in\mathcal H^*\) be the
complex conjugate of \(|a\rangle\) in the computational basis.
For inverse temperature \(\beta\geq0\), define the
thermofield-double state
\begin{equation}
\label{eq:technical-tfd}
    |\psi_\beta\rangle
    :=
    \frac{1}{\sqrt{Z(\beta)}}
    \sum_a e^{-\beta E_a/2}
    |a\rangle\otimes|a^*\rangle.
\end{equation}
Tracing out the purifying subsystem gives the Gibbs state
\(\rho_\beta\); the reduced state on the conjugate subsystem is
\(\rho_\beta^*\).
\end{defn}

\medskip

Both algorithms use TFD-overlap probabilities and reweighting
ratios to reconstruct the log-partition function.
We define these quantities below.

\begin{defn}[Overlaps and reweighting ratios]
\label{def:local-thermal-quantities}
For \(\beta,s\geq0\), define
\begin{equation}
\label{eq:technical-O-R}
    O(\beta,s)
    :=
    |\langle\psi_\beta|\psi_{\beta+s}\rangle|^2
    =
    \frac{Z(\beta+s/2)^2}{Z(\beta)Z(\beta+s)},
    \qquad
    R(\beta,s)
    :=
    \Tr(\rho_\beta e^{-sH})
    =
    \frac{Z(\beta+s)}{Z(\beta)},
\end{equation}
and their logarithms
\begin{equation}
\label{eq:technical-A-B}
    A(\beta,s):=-\ln O(\beta,s),
    \qquad
    B(\beta,s):=\ln R(\beta,s).
\end{equation}
Equivalently,
\[
    A(\beta,s)
    =
    f(\beta)+f(\beta+s)
    -2f\!\left(\beta+\frac{s}{2}\right),
    \qquad
    B(\beta,s)
    =
    f(\beta+s)-f(\beta).
\]
\end{defn}

\medskip

Thus \(B(\beta,s)\) is the increment of \(f\) over an interval of
length \(s\), whereas \(A(\beta,s)\) is a centered second finite
difference of \(f\).
Since \(f\) is convex, \(A(\beta,s)\geq0\).
The corresponding thermodynamic derivatives are recorded in
Remark~\ref{rem:thermodynamic-derivatives}.

\subsection{Projection probabilities for the two quantities}
\label{subsec:projection-probabilities}

The estimation algorithms operate on quantum measurement
probabilities. For overlaps, we use the TFD preparations
\(U_\theta|0\rangle=|\psi_\theta\rangle\).
For reweighting, we use a block encoding of
\(e^{-1}e^{-sH/2}\), constructed from the block encoding of \(H\)
by QSVT; see Appendix~\ref{app:single-temperature},
Subsection~\ref{app:qsvt-imaginary-time}.
The following lemma gives the two success probabilities and
explains why the reweighting logarithm requires adding \(2\).

\begin{lem}[Projection probabilities for overlaps and reweighting]
\label{lem:local-quantity-projection}
Let \(Q=S\otimes R\), where \(S\) is the physical register and
\(R\) is the purifying register.

\medskip

For the overlap probability \(O(\beta,s)\), define
\[
    |\phi_{\beta,s}^{\rm ov}\rangle
    :=
    U_\beta^\dagger U_{\beta+s}|0\rangle_Q,
    \qquad
    \Pi^{\rm ov}:=|0\rangle\!\langle0|_Q.
\]
Then
\begin{equation}
\label{eq:overlap-projection-probability}
    \bigl\|\Pi^{\rm ov}
    |\phi_{\beta,s}^{\rm ov}\rangle\bigr\|^2
    =
    O(\beta,s).
\end{equation}

\medskip

For the reweighting ratio \(R(\beta,s)\), let \(D_s\) be a unitary acting
on an \(a\)-qubit ancilla register \(A\) and the physical register
\(S\), whose all-zero ancilla block is an operator
\(\widetilde A_s\) satisfying, for an accuracy parameter \(\nu\geq0\),
\[
    (\langle0^a|_A\otimes I_S)
    D_s
    (|0^a\rangle_A\otimes I_S)
    =
    \widetilde A_s,
    \qquad
    \|\widetilde A_s-e^{-1}e^{-sH/2}\|\leq\nu.
\]
Define
\[
    |\phi_{\beta,s}^{R}\rangle
    :=
    (D_s\otimes I_R)
    \bigl(|0^a\rangle_A|\psi_\beta\rangle_{SR}\bigr),
    \qquad
    \Pi^{R}
    :=
    |0^a\rangle\!\langle0^a|_A\otimes I_{SR}.
\]
The corresponding projection probability is
\begin{equation}
\label{eq:R-implemented-probability}
    \widetilde R(\beta,s)
    :=
    \bigl\|\Pi^{R}|\phi_{\beta,s}^{R}\rangle\bigr\|^2
    =
    \langle\psi_\beta|
    \bigl(\widetilde A_s^\dagger\widetilde A_s\otimes I_R\bigr)
    |\psi_\beta\rangle.
\end{equation}
In the exact case \(\widetilde A_s=e^{-1}e^{-sH/2}\),
\[
    \widetilde R(\beta,s)=e^{-2}R(\beta,s).
\]
More generally, if both \(\widetilde A_s\) and
\(e^{-1}e^{-sH/2}\) are contractions, then
\begin{equation}
\label{eq:R-probability-error}
    |\widetilde R(\beta,s)-e^{-2}R(\beta,s)|
    \leq2\nu.
\end{equation}
\end{lem}

\begin{proof}
For the overlap quantity,
\[
    \langle0|_Q
    U_\beta^\dagger U_{\beta+s}
    |0\rangle_Q
    =
    \langle\psi_\beta|\psi_{\beta+s}\rangle,
\]
and squaring the amplitude of the all-zero outcome gives
Eq.~\eqref{eq:overlap-projection-probability}.

\medskip

For the second quantity, projecting the dilation ancilla onto
\(|0^a\rangle_A\) leaves
\[
    |0^a\rangle_A
    (\widetilde A_s\otimes I_R)
    |\psi_\beta\rangle_{SR}.
\]
Taking its squared norm gives
Eq.~\eqref{eq:R-implemented-probability}.
For the ideal operator,
\[
    \bigl\|
        (e^{-1}e^{-sH/2}\otimes I_R)|\psi_\beta\rangle
    \bigr\|^2
    =
    e^{-2}\Tr(\rho_\beta e^{-sH})
    =
    e^{-2}R(\beta,s).
\]
Finally,
\[
    \|\widetilde A_s^\dagger\widetilde A_s-e^{-2}e^{-sH}\|
    \leq2\nu
\]
whenever the two operators are contractions, which implies
Eq.~\eqref{eq:R-probability-error}.
\end{proof}

\medskip

Thus reweighting logarithms require the known additive correction
\(2\), with approximation error controlled by
Eq.~\eqref{eq:R-probability-error}. We next determine which
probabilities to estimate and how to combine their logarithms.

\subsection{Recursive-doubling reconstruction}
\label{subsec:recursive-doubling}

Consider a cooling schedule
\(0=\beta_0<\beta_1<\cdots<\beta_\ell=\beta\), and write
\(\Delta_i:=\beta_{i+1}-\beta_i\).
Even when \(O(\beta_i,\Delta_i)\geq1/C\), the reweighting ratio
\(R(\beta_i,\Delta_i)\) need not have a constant lower bound.
For example, \(H=hI\) gives \(O(\beta_i,\Delta_i)=1\) but
\(R(\beta_i,\Delta_i)=e^{-h\Delta_i}\).
We therefore start with a short reweighting step and recover the
full increment by doubling the interval length and adding overlap
corrections. For each interval we choose
\begin{equation}
\label{eq:dyadic-parameters}
    k_i
    :=
    \max\{0,\lceil\log_2(h\Delta_i)\rceil\},
    \qquad
    \delta_i
    :=
    \frac{\Delta_i}{2^{k_i}}.
\end{equation}
By construction, \(2^{k_i}\delta_i=\Delta_i\) and
\(h\delta_i\leq1\). Since \(0\preceq H\preceq hI\), this gives
\(R(\beta_i,\delta_i)\geq e^{-h\delta_i}\geq e^{-1}\).

\medskip

The following identity expresses each increment
\(f(\beta_{i+1})-f(\beta_i)\) exactly in terms of one \(B\)-term at
the smallest scale \(\delta_i\) and a sequence of \(A\)-terms over
intervals of lengths
\(2\delta_i,4\delta_i,\ldots,2^{k_i}\delta_i\).
Figure~\ref{fig:recursive-doubling} illustrates the successive
doubling steps and the corresponding correction terms.

\begin{lem}[Recursive-doubling reconstruction formula]
\label{lem:recursive-doubling}
For the parameters in Eq.~\eqref{eq:dyadic-parameters},
\begin{equation}
\label{eq:interval-dyadic-reconstruction}
    f(\beta_{i+1})
    =
    f(\beta_i)
    +
    2^{k_i}B(\beta_i,\delta_i)
    +
    \sum_{j=1}^{k_i}
    2^{k_i-j}A(\beta_i,2^j\delta_i).
\end{equation}
Consequently,
\begin{equation}
\label{eq:full-dyadic-reconstruction}
    f(\beta)
    =
    f(0)
    +
    \sum_{i=0}^{\ell-1}
    \left[
        2^{k_i}B(\beta_i,\delta_i)
        +
        \sum_{j=1}^{k_i}
        2^{k_i-j}A(\beta_i,2^j\delta_i)
    \right].
\end{equation}
\end{lem}

\begin{proof}
Fix \(i\).
For each \(r=1,\ldots,k_i\), the definition of \(A\) gives
\[
    A(\beta_i,2^r\delta_i)
    =
    f(\beta_i)
    +
    f(\beta_i+2^r\delta_i)
    -
    2f(\beta_i+2^{r-1}\delta_i),
\]
and hence
\[
    f(\beta_i+2^r\delta_i)-f(\beta_i)
    =
    2\bigl[
        f(\beta_i+2^{r-1}\delta_i)-f(\beta_i)
      \bigr]
    +
    A(\beta_i,2^r\delta_i).
\]
Iterating this identity over \(r=1,\ldots,k_i\) yields
\[
    f(\beta_i+2^{k_i}\delta_i)-f(\beta_i)
    =
    2^{k_i}
    \bigl[f(\beta_i+\delta_i)-f(\beta_i)\bigr]
    +
    \sum_{j=1}^{k_i}
    2^{k_i-j}A(\beta_i,2^j\delta_i).
\]
Using
\[
    f(\beta_i+\delta_i)-f(\beta_i)
    =
    B(\beta_i,\delta_i)
\]
and \(2^{k_i}\delta_i=\beta_{i+1}-\beta_i\) proves
Eq.~\eqref{eq:interval-dyadic-reconstruction}.
When \(k_i=0\), the sum over \(j\) is empty and equals zero.
Summing over the schedule intervals then telescopes to
Eq.~\eqref{eq:full-dyadic-reconstruction}.
\end{proof}

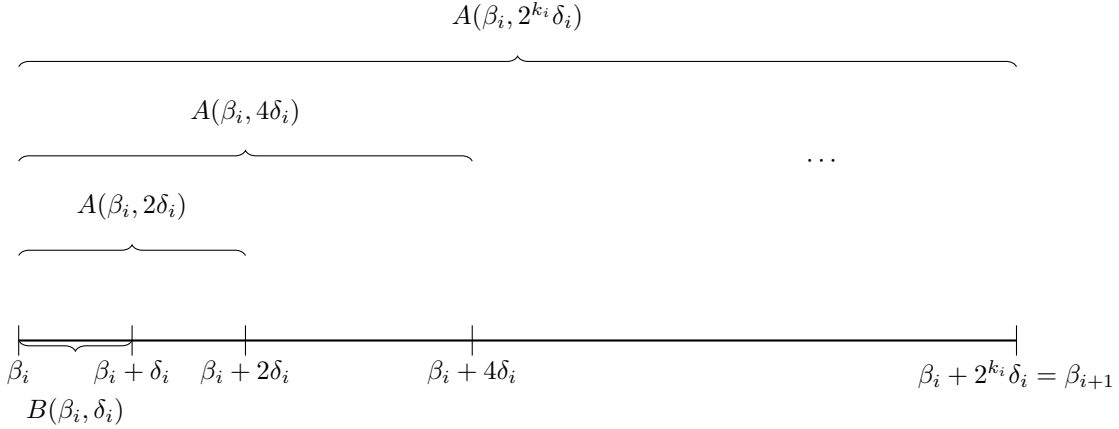
\begin{figure}[tbp]
\centering
\begin{tikzpicture}[x=3cm,y=2.5cm,>=stealth]

\coordinate (b0) at (0,0);
\coordinate (b1) at (0.5,0);
\coordinate (b2) at (1,0);
\coordinate (b4) at (2,0);
\coordinate (b8) at (4.4,0);

\draw[thick] (b0) -- (b8);

\foreach \p/\lab in {
  b0/$\beta_i$,
  b1/$\beta_i+\delta_i$,
  b2/$\beta_i+2\delta_i$,
  b4/$\beta_i+4\delta_i$,
  b8/$\beta_i+2^{k_i}\delta_i=\beta_{i+1}$}
{
  \draw (\p) ++(0,0.08) -- ++(0,-0.16);
  \node[below=4pt] at (\p) {\lab};
}

\draw[decorate,decoration={brace,mirror,amplitude=4pt}]
  (b0) -- (b1)
  node[midway,below=18pt] {$B(\beta_i,\delta_i)$};

\draw[decorate,decoration={brace,amplitude=4pt}]
  ($(b0)+(0,0.45)$) -- ($(b2)+(0,0.45)$)
  node[midway,above=10pt] {$A(\beta_i,2\delta_i)$};

\draw[decorate,decoration={brace,amplitude=4pt}]
  ($(b0)+(0,0.95)$) -- ($(b4)+(0,0.95)$)
  node[midway,above=10pt] {$A(\beta_i,4\delta_i)$};

\draw[decorate,decoration={brace,amplitude=4pt}]
  ($(b0)+(0,1.45)$) -- ($(b8)+(0,1.45)$)
  node[midway,above=10pt] {$A(\beta_i,2^{k_i}\delta_i)$};

\node at (3.55,0.95) {$\cdots$};

\end{tikzpicture}

\caption{Recursive-doubling reconstruction on the schedule interval
\([\beta_i,\beta_{i+1}]\). The base term \(B(\beta_i,\delta_i)\) gives
the increment at the finest scale, while each term
\(A(\beta_i,2^j\delta_i)\) supplies the correction required to double
the reconstructed scale from \(2^{j-1}\delta_i\) to
\(2^j\delta_i\). The horizontal spacing is schematic.}
\label{fig:recursive-doubling}
\end{figure}

For later use, define the reconstruction coefficients
\begin{equation}
\label{eq:reconstruction-coefficients}
    a_{i,0}:=2^{k_i},
    \qquad
    a_{i,j}:=2^{k_i-j},
    \quad 1\leq j\leq k_i,
\end{equation}
and the corresponding logarithmic quantities
\begin{equation}
\label{eq:local-logarithmic-terms}
    Y_{i,0}:=B(\beta_i,\delta_i),
    \qquad
    Y_{i,j}:=A(\beta_i,2^j\delta_i),
    \quad 1\leq j\leq k_i.
\end{equation}
Equation~\eqref{eq:full-dyadic-reconstruction} can then be written
compactly as
\begin{equation}
\label{eq:weighted-log-reconstruction}
    f(\beta)-f(0)
    =
    \sum_{i=0}^{\ell-1}
    \sum_{j=0}^{k_i}
    a_{i,j}Y_{i,j}.
\end{equation}
By Definition~\ref{def:local-thermal-quantities}, the exact
logarithmic terms are
\[
    Y_{i,0}
    =
    \ln R(\beta_i,\delta_i),
    \qquad
    Y_{i,j}
    =
    -\ln O(\beta_i,2^j\delta_i)
    \quad (j\geq1).
\]

\medskip

We multiply the short-step \(B\)-term by \(2^{k_i}\) and add
the weighted \(A\)-terms to obtain
\(f(\beta_{i+1})-f(\beta_i)\) exactly.
Remark~\ref{rem:recursive-doubling-interpretation} explains
the roles of these terms in more detail.
\subsection{Conditioning through a slowly-varying schedule}
\label{subsec:slowly-varying-conditioning}

The logarithms are well conditioned when the probabilities are
bounded below by a positive constant.
For overlaps, it is enough to require this bound between
consecutive schedule temperatures.

\begin{defn}[Slowly-varying schedule]
\label{def:slowly-varying-schedule}
A schedule \(0=\beta_0<\beta_1<\cdots<\beta_\ell=\beta\)
is \(C\)-slowly varying if
\begin{equation}
\label{eq:slowly-varying-condition}
    O(\beta_i,\beta_{i+1}-\beta_i)\geq\frac1C
\end{equation}
for every \(i\), where \(C=O(1)\).
\end{defn}

\medskip

The recursive-doubling formula uses shorter intervals of the form
\([\beta_i,\beta_i+2^j\delta_i]\).
By convexity of \(f\), shortening a schedule interval does not
decrease its overlap probability.

\begin{lem}[Overlap monotonicity under restriction]
\label{lem:overlap-monotonicity}
If
\[
    \beta_1\leq\beta_2\leq\beta_3\leq\beta_4,
\]
then
\[
    O(\beta_2,\beta_3-\beta_2)
    \geq
    O(\beta_1,\beta_4-\beta_1).
\]
In particular, every overlap probability used in the recursive-doubling
decomposition within a
\(C\)-slowly-varying schedule interval is at least \(1/C\).
\end{lem}

\begin{proof}
For \(x\leq y\), define
\(F(x,y):=f(x)+f(y)-2f((x+y)/2)\).
Convexity of \(f\) gives
\[
\begin{aligned}
\partial_xF(x,y)&=f'(x)-f'\!\left(\frac{x+y}{2}\right)\leq0,\\
\partial_yF(x,y)&=f'(y)-f'\!\left(\frac{x+y}{2}\right)\geq0.
\end{aligned}
\]
Moving the left endpoint to the right or the right endpoint to the
left therefore cannot increase \(F\), so
\(F(\beta_2,\beta_3)\leq F(\beta_1,\beta_4)\).
The claim follows from \(O(x,y-x)=e^{-F(x,y)}\).
\end{proof}

\medskip

Combining the slowly-varying condition with the choice of the
elementary steps \(\delta_i\) gives a uniform lower bound for every
probability entering the reconstruction.

\begin{cor}[Uniform conditioning of the local probabilities]
\label{cor:uniform-local-conditioning}
Suppose that the schedule
\(0=\beta_0<\beta_1<\cdots<\beta_\ell=\beta\) is \(C\)-slowly varying.
For every schedule interval \(0\leq i<\ell\), the scaled reweighting
probability at the elementary step \(\delta_i\) and all overlap
probabilities appearing in the recursive-doubling reconstruction
satisfy
\begin{equation}
\label{eq:uniform-conditioning}
e^{-2}R(\beta_i,\delta_i)\geq\gamma_0,
\qquad
O(\beta_i,2^j\delta_i)\geq\gamma_0,
\quad 1\leq j\leq k_i,
\qquad
\gamma_0:=\min\{e^{-3},1/C\}.
\end{equation}
Thus every ideal circuit probability is at least \(\gamma_0\),
which is \(\Omega(1)\) when \(C=O(1)\).
\end{cor}

\begin{proof}
For \(j\geq1\), the interval
\([\beta_i,\beta_i+2^j\delta_i]\) is contained in
\([\beta_i,\beta_{i+1}]\).
Lemma~\ref{lem:overlap-monotonicity} therefore gives
\[
    O(\beta_i,2^j\delta_i)\geq\frac1C.
\]
For \(j=0\), the assumptions \(0\preceq H\preceq hI\) and
\(h\delta_i\leq1\) imply
\[
    R(\beta_i,\delta_i)
    =
    \Tr(\rho_{\beta_i}e^{-\delta_iH})
    \geq
    e^{-h\delta_i}
    \geq
    e^{-1}.
\]
Multiplying the reweighting bound by \(e^{-2}\) gives
\(e^{-2}R(\beta_i,\delta_i)\geq e^{-3}\).
Taking the smaller of \(e^{-3}\) and \(1/C\) proves
Eq.~\eqref{eq:uniform-conditioning}.
\end{proof}

\medskip

Thus the deterministic part of the algorithm has reduced
partition-function estimation to a weighted sum of logarithms of
ideal projection probabilities lying in \([\gamma_0,1]\), together
with the known reweighting correction.
For \(C=O(1)\), this constant lower bound is what makes the
logarithmic estimation procedures in
Sections~\ref{sec:single-temperature} and~\ref{sec:sampled-dyadic}
stable.

\medskip

Theorem~\ref{thm:quantum-schedule-arunachalam} gives the
schedule-generation procedure, its length guarantees, and its
separate query cost. Here the schedule is supplied input.

\subsection{From the common reconstruction to the two algorithms}
\label{subsec:two-statistical-assemblies}

The reconstruction uses the coefficients \(a_{i,j}\) and exact
logarithmic terms \(Y_{i,j}\) in
Eq.~\eqref{eq:weighted-log-reconstruction}.
Once the QSVT implementation in
Lemma~\ref{lem:local-quantity-projection} is fixed, define the
actual circuit probabilities by
\begin{equation}
\label{eq:implemented-local-probabilities}
p_{i,j}
:=
\begin{cases}
\widetilde R(\beta_i,\delta_i), & j=0,\\[1mm]
O(\beta_i,2^j\delta_i), & 1\leq j\leq k_i.
\end{cases}
\end{equation}
For \(j\geq1\), this is the exact overlap probability.
For \(j=0\), it approximates \(e^{-2}R(\beta_i,\delta_i)\), so
\(\ln p_{i,0}+2\) approximates \(Y_{i,0}\).
Both algorithms include this correction in each reweighting
estimate or sample and budget the QSVT error separately.

\medskip

Section~\ref{sec:single-temperature} estimates the logarithmic
terms separately with coefficient-dependent depths.
Section~\ref{sec:sampled-dyadic} prepares a coherent sample over
the same terms and estimates its mean.
Both use the reconstruction and conditioning bounds proved here.
\section{Partition-function estimation with fixed-temperature access}
\label{sec:single-temperature}
In this section, we give the first of two algorithms for estimating the
weighted sum in Eq.~\eqref{eq:weighted-log-reconstruction}. We consider
fixed-temperature access, where each required inverse temperature is supplied
classically to the corresponding state-preparation routine. Under this access
model, the logarithmic terms are estimated independently from their
corresponding projection probabilities and then combined using the
reconstruction coefficients. The estimation accuracy is allocated across the
terms according to these coefficients, which reduces the total query cost
while controlling the error of the weighted sum. We first define the access
model and analyze the estimator for a single logarithmic term, and then
present the full reconstruction algorithm and its complexity.

\begin{defn}[Fixed-temperature oracle access and query complexity]
\label{def:partition-function-resource-model}
Let \(U_H\) be an \(\alpha\)-block encoding of \(H\), with
\(\alpha=\Theta(h)\), and assume access to \(U_H^\dagger\) and to the
controlled versions of \(U_H\) and \(U_H^\dagger\).
For each classically specified inverse temperature \(\theta\), let
\(U_\theta\) be the fixed-temperature TFD preparation unitary satisfying
\[
U_\theta\ket0=\ket{\psi_\theta},
\]
where \(\ket{\psi_\theta}\) is the TFD state of
Definition~\ref{def:tfd-state}.  We also assume access to
\(U_\theta^\dagger\).

\medskip

Define
\[
\begin{aligned}
N_{\mathrm{TFD}}
&:=
\text{the total number of fixed-temperature TFD preparation queries},
\\
N_{\mathrm{BE}}
&:=
\text{the total number of Hamiltonian block-encoding queries}.
\end{aligned}
\]
Thus calls to \(U_\theta\) and \(U_\theta^\dagger\) contribute to
\(N_{\mathrm{TFD}}\), while calls to \(U_H\), \(U_H^\dagger\), and their
controlled versions each contribute one query to \(N_{\mathrm{BE}}\).
\end{defn}

\medskip

The reconstruction takes the cooling schedule as classical input.
We count TFD preparation and Hamiltonian block-encoding calls here.

\subsection{Local logarithmic estimation}
\label{subsec:single-temperature-local}
To estimate the local logarithmic terms, we use the nondestructive unbiased
amplitude-estimation algorithm
\(\operatorname{NdUb\text{-}Amplitude}(\ket{\psi},\Pi,M,\eta)\) of Cornelissen and
Hamoudi~\cite[Algorithm~3]{Cornelissen2023}, together with its bias,
variance, and reflection-cost guarantees from
\cite[Theorem~2.4]{Cornelissen2023}. When applied
with depth parameter \(M\) and accuracy parameter \(\eta\), denote the
resulting estimate by \(\widehat p_{M,\eta}\). Its bias is at most \(\eta\),
and its variance is \(O(p/M^2+\eta)\). The full statement is recorded as
Theorem~\ref{thm:ndub-amplitude} in
Appendix~\ref{app:local-log-estimation}.

\medskip

By Lemma~\ref{lem:local-quantity-projection} and
Eq.~\eqref{eq:implemented-local-probabilities}, the circuit associated
with each reconstruction term \((i,j)\) produces the projection probability
\[
p_{i,j}=\|\Pi\ket{\psi}\|^2.
\]
This is the quantity estimated by the amplitude-estimation procedure.
For an overlap term \(j\ge1\),
\(-\ln p_{i,j}=Y_{i,j}\) exactly. For a reweighting term \(j=0\),
\(\ln p_{i,0}+2\) approximates the exact quantity
\(Y_{i,0}=\ln R(\beta_i,\delta_i)\), with the difference controlled by
Lemma~\ref{lem:single-temperature-dilation-stability}.
By Eq.~\eqref{eq:uniform-conditioning} and
Lemma~\ref{lem:single-temperature-dilation-stability}, every implemented
projection probability satisfies
\[
p_{i,j}\ge \gamma,
\qquad
\gamma:=\frac12\min\{e^{-3},1/C\}.
\]
This uniform lower bound ensures that the logarithm can be estimated stably.

\begin{algorithm}[H]
\caption{Local logarithmic estimation,
\(\mathsf{LocalLogEst}(\mathsf{type},\theta,\delta,M,\gamma,\eta)\)}
\label{alg:local-log-estimate}
\small
\begin{algorithmic}[1]

\Require A branch type
\(\mathsf{type}\in\{\mathsf{overlap},\mathsf{reweight}\}\);
parameters \(\theta,\delta\ge0\); an integer \(M\ge5\);
the uniform lower bound \(\gamma\) defined above; and
\(\eta\in(0,1)\).

\Require Fixed-temperature TFD access to \(U_\theta\) and its adjoint.
For an overlap branch, also require
\(U_{\theta+\delta}\) and its adjoint.
For a reweighting branch, require the QSVT unitary
\(D_\delta\) from Lemma~\ref{lem:qsvt-imaginary-time},
constructed at the prescribed approximation accuracy, and its adjoint.

\Ensure An estimate of \(-\ln p\) for an overlap branch, or of
\(\ln p+2\) for a reweighting branch, including the QSVT normalization
correction.

\If{\(\mathsf{type}=\mathsf{overlap}\)}
    \State Prepare
    \[
    \ket{\psi}
    \gets
    U_\theta^\dagger U_{\theta+\delta}\ket0_{SR},
    \qquad
    \Pi
    \gets
    \ket0\!\bra0_{SR}.
    \]
\Else
    \State Prepare
    \[
    \ket{\psi}
    \gets
    (D_\delta\otimes I_R)
    \bigl(\ket{0^a}_A\ket{\psi_\theta}_{SR}\bigr),
    \qquad
    \Pi
    \gets
    \ket{0^a}\!\bra{0^a}_A\otimes I_{SR}.
    \]
\EndIf

\State Run
\(\operatorname{NdUb\text{-}Amplitude}(\ket{\psi},\Pi,M,\eta)\)
and denote its estimate of
\(p=\|\Pi\ket{\psi}\|^2\) by \(\widehat p_{M,\eta}\).

\State Set
\[
\widehat p^{\,\mathrm{clip}}_{M,\eta}
\gets
\min\!\left\{
2,\max\!\left\{\frac{\gamma}{2},\widehat p_{M,\eta}\right\}
\right\}.
\]

\If{\(\mathsf{type}=\mathsf{overlap}\)}
    \State \Return
    \(-\ln\widehat p^{\,\mathrm{clip}}_{M,\eta}\).
\Else
    \State \Return
    \(\ln\widehat p^{\,\mathrm{clip}}_{M,\eta}+2\).
\EndIf

\end{algorithmic}
\end{algorithm}
\medskip

The amplitude-estimation procedure used in
Algorithm~\ref{alg:local-log-estimate} requires reflections about the
prepared state \(\ket{\psi}\) and the reflection \(I-2\Pi\), as stated in
Theorem~\ref{thm:ndub-amplitude}. If \(\ket{\psi}=V\ket{0}\), the
reflection about \(\ket{\psi}\) can be implemented as
\[
I-2\ket{\psi}\!\bra{\psi}
=
V\bigl(I-2\ket{0}\!\bra{0}\bigr)V^\dagger .
\]
For an overlap term, \(V=U_\theta^\dagger U_{\theta+\delta}\), so one
such reflection uses four fixed-temperature TFD queries. For a reweighting term, the prepared state is
\[
\ket{\psi}
=
(D_\delta\otimes I_R)
\bigl(\ket{0_a}_A\otimes\ket{\psi_\theta}_{SR}\bigr).
\]
Thus, if the preparation unitary is denoted by \(V\), then
\[
V
=
(D_\delta\otimes I_R)
(I_A\otimes U_\theta).
\]
A reflection about \(\ket{\psi}\) is implemented by applying
\(V^\dagger\), performing the phase flip about the all-zero input
state, and applying \(V\). Hence one such reflection uses one call each to
\(U_\theta\) and \(U_\theta^\dagger\), and one application each of
\(D_\delta\) and \(D_\delta^\dagger\). It therefore uses two
fixed-temperature TFD queries and \(O(q)\) Hamiltonian block-encoding
queries when \(D_\delta\) has QSVT degree \(q\).

\medskip

With the reflection costs established, we next bound the error of
Algorithm~\ref{alg:local-log-estimate}. The following lemma translates the
amplitude-estimation guarantees into bias and variance bounds for the
logarithmic estimate, and records the corresponding query cost.
Theorem~\ref{thm:single-temperature-main} combines these local bounds to
control the error and cost of the full reconstruction.

\begin{lem}[Local logarithmic estimation]
\label{lem:almost-unbiased-local}
Fix \(0<\gamma\le1\), an integer \(M\ge5\), and
\(\eta\in(0,1)\). Let \(\ket{\psi}\) and \(\Pi\) be the state and
projector constructed in Algorithm~\ref{alg:local-log-estimate}, and let
\(p:=\|\Pi\ket{\psi}\|^2\in[\gamma,1]\).
Let \(\widehat Y_{M,\eta}\) be the output of the algorithm, and let
\(c_v:=91\) be the variance constant in
Theorem~\ref{thm:ndub-amplitude}. Define
\begin{equation}
\label{eq:single-temperature-K-gamma-main}
K_\gamma
:=
\frac{1}{\gamma}
+
\frac{
\bigl(4+8\ln(2/\gamma)\bigr)(c_v+2)
}{
\gamma^2
}.
\end{equation}
Set \(Y=-\ln p\) for an overlap term and \(Y=\ln p+2\) for a
reweighting term. Then
\[
\left|\mathbb E[\widehat Y_{M,\eta}]-Y\right|
\le
K_\gamma\left(M^{-2}+\eta\right),
\qquad
\operatorname{Var}(\widehat Y_{M,\eta})
\le
K_\gamma\left(M^{-2}+\eta\right).
\]
In particular, if \(\eta\le M^{-2}\), then both quantities are at most
\(2K_\gamma/M^2\).

\medskip

In Algorithm~\ref{alg:finite-temperature-reconstruction}, we choose
\(\eta=M^{-2}\). Thus the internal accuracy parameter does not dominate
the \(M^{-2}\) error from amplitude estimation. Its contribution to the
reflection cost is only logarithmic and is absorbed into the
\(\widetilde{\mathcal O}\)-notation. Moreover, in the regime \(C=O(1)\) considered in
Theorem~\ref{thm:single-temperature-main},
Eq.~\eqref{eq:uniform-conditioning} and
Lemma~\ref{lem:single-temperature-dilation-stability}
allow us to take
\(\gamma=\gamma_0/2=\Omega(1)\).
Hence \(K_\gamma=O(1)\).

\medskip

The expected query counts are
\[
\begin{aligned}
\textnormal{overlap:}\qquad
&N_{\mathrm{TFD}}
=\widetilde{\mathcal O}(M),
&&N_{\mathrm{BE}}=0,
\\
\textnormal{reweighting:}\qquad
&N_{\mathrm{TFD}}
=\widetilde{\mathcal O}(M),
&&N_{\mathrm{BE}}
=\widetilde{\mathcal O}(Mq),
\end{aligned}
\]
where \(q\) is the QSVT degree of \(D_\delta\) from
Lemma~\ref{lem:qsvt-imaginary-time}. If
\(\alpha=\Theta(h)\), \(h\delta\le1\), and the QSVT accuracy is inverse
polynomial in the reconstruction parameters, then
\(q=\widetilde{\mathcal O}(1)\). Hence the reweighting branch also has
\(N_{\mathrm{BE}}=\widetilde{\mathcal O}(M)\).
\end{lem}

\medskip

The clipping argument, the logarithmic bias bound, and the conversion of the
reflection count in Theorem~\ref{thm:ndub-amplitude} to the oracle counts
above are proved in Appendix~\ref{app:local-log-estimation}.
\subsection{Variance-sensitive reconstruction}
\label{subsec:variance-sensitive-reconstruction}

The reconstruction in
Eq.~\eqref{eq:full-dyadic-reconstruction} is a weighted sum
of local logarithms.  By Lemma~\ref{lem:almost-unbiased-local}, a local
estimate with depth \(M_{i,j}\) has variance
\(O(M_{i,j}^{-2})\).  After multiplication by its reconstruction coefficient
\(a_{i,j}\), its contribution to the variance of the final estimator is
therefore \(O(a_{i,j}^2/M_{i,j}^2)\).  Since the coefficients
\(a_{i,j}\) vary across the reconstruction, using the same depth for every
local term is not query efficient.
We instead choose the depths to minimize the total local cost
\(\sum_{i,j}M_{i,j}\) subject to the required variance bound on the weighted
sum.  Solving this constrained optimization with a Lagrange multiplier gives
the relative allocation
\[
M_{i,j}\propto a_{i,j}^{2/3}.
\]
The common proportionality factor is then chosen so that the total variance
is \(O(\epsilon^2)\).  The estimation depths are chosen according to
Eq.~\eqref{eq:single-temperature-optimal-depth}; the full optimization
procedure is given in
Appendix~\ref{app:local-depth-allocation}.
The resulting reconstruction algorithm is given below.

\begin{algorithm}[H]
\caption{Fixed-temperature partition-function reconstruction}
\label{alg:finite-temperature-reconstruction}
\small
\begin{algorithmic}[1]
\Require An \(n\)-qubit Hamiltonian \(H\) with
\(0\preceq H\preceq hI\) and \(h>0\); an \(\alpha\)-block encoding of
\(H\) with \(\alpha=\Theta(h)\); a supplied \(C\)-slowly-varying schedule
\(0=\beta_0<\cdots<\beta_{\ell}=\beta\); the
fixed-temperature TFD access of
Definition~\ref{def:partition-function-resource-model};
\(0<\epsilon\le1\); and \(0<p<1\).
\Ensure An estimate \(\widehat f(\beta)\) satisfying
\(\Pr[|\widehat f(\beta)-f(\beta)|\le\epsilon]\ge1-p\).

\State Compute \(\Delta_i\), \(k_i\), and \(\delta_i\) from
\eqref{eq:dyadic-parameters} for every
\(i=0,\ldots,\ell-1\).

\State Set
\(\gamma_0\gets\min\{e^{-3},1/C\}\),
\(\gamma\gets\gamma_0/2\), and
\(\epsilon_{\mathrm{stat}}\gets\epsilon/2\).

\State Compute
\[
S\gets
\sum_{i=0}^{\ell-1}
\left(
2^{2k_i/3}
+
\sum_{j=1}^{k_i}2^{2(k_i-j)/3}
\right).
\]

\State
\[
\nu\gets
\frac{\gamma_0\epsilon}
{64\sum_{i=0}^{\ell-1}2^{k_i}}.
\]
For each \(i\), choose the QSVT dilation used for the reweighting term so
that its all-zero block \(\widetilde A_{\delta_i}\) satisfies
\[
\left\|
\widetilde A_{\delta_i}-e^{-1}e^{-\delta_iH/2}
\right\|
\le\nu.
\]

\For{\(i=0,\ldots,\ell-1\)}
    \State Set
    \(
    M_{i,0}\gets
    \max\!\left\{
    5,
    \left\lceil
    \frac{
    4\sqrt{K_\gamma}\,S^{1/2}2^{2k_i/3}
    }{\epsilon_{\mathrm{stat}}}
    \right\rceil
    \right\},
    \qquad
    \eta_{i,0}\gets M_{i,0}^{-2}.
    \)
    \For{\(j=1,\ldots,k_i\)}
        \State Set
        \(
        M_{i,j}\gets
        \max\!\left\{
        5,
        \left\lceil
        \frac{
        4\sqrt{K_\gamma}\,S^{1/2}2^{2(k_i-j)/3}
        }{\epsilon_{\mathrm{stat}}}
        \right\rceil
        \right\},
        \qquad
        \eta_{i,j}\gets M_{i,j}^{-2}.
        \)
    \EndFor
\EndFor

\State Set
\[
N_{\mathrm{rep}}
\gets
2\left\lceil9\ln(2/p)\right\rceil+1.
\]

\For{\(r=1,\ldots,N_{\mathrm{rep}}\)}
    \State Initialize
    \(\widehat f^{(r)}(\beta)\gets n\ln2\).

    \For{\(i=0,\ldots,\ell-1\)}
        \State Using a fresh run, set
        \[
        \widehat B_i^{(r)}
        \gets
        \mathsf{LocalLogEst}
        (\mathsf{reweight},
        \beta_i,\delta_i,M_{i,0},\gamma,\eta_{i,0}),
        \]
        using the QSVT dilation chosen above.
        \State Update
        \[
        \widehat f^{(r)}(\beta)
        \gets
        \widehat f^{(r)}(\beta)
        +
        2^{k_i}\widehat B_i^{(r)}.
        \]

        \For{\(j=1,\ldots,k_i\)}
            \State Using a fresh run, set
            \[
            \widehat A_{i,j}^{(r)}
            \gets
            \mathsf{LocalLogEst}
            (\mathsf{overlap},
            \beta_i,2^j\delta_i,M_{i,j},\gamma,\eta_{i,j}).
            \]
            \State Update
            \[
            \widehat f^{(r)}(\beta)
            \gets
            \widehat f^{(r)}(\beta)
            +
            2^{k_i-j}\widehat A_{i,j}^{(r)}.
            \]
        \EndFor
    \EndFor
\EndFor

\State \Return
\[
\operatorname{median}
\left\{
\widehat f^{(1)}(\beta),
\ldots,
\widehat f^{(N_{\mathrm{rep}})}(\beta)
\right\}.
\]
\end{algorithmic}
\end{algorithm}
\medskip

The bias and variance bounds in Lemma~\ref{lem:almost-unbiased-local},
together with the depth allocation in
Algorithm~\ref{alg:finite-temperature-reconstruction}, control the error
of the complete weighted sum. The following theorem gives the
reconstruction cost for a supplied slowly-varying schedule and, in
particular, for the extensive constant-temperature regime.
\begin{thm}[Partition-function estimation with fixed-temperature access]
\label{thm:single-temperature-main}
\label{cor:complexity_simple}
Let \(H\) be an \(n\)-qubit Hamiltonian with
\(0\preceq H\preceq hI\) for some \(h>0\). Let
\(0<\epsilon\le1\) and \(0<p<1\). Suppose that a
\(C\)-slowly-varying schedule
\[
0=\beta_0<\beta_1<\cdots<\beta_{\ell}=\beta
\]
is supplied, where \(C=O(1)\) and \(0<\beta<\infty\). Assume an
\(\alpha\)-block encoding of \(H\) with \(\alpha=\Theta(h)\) and the
fixed-temperature TFD access of
Definition~\ref{def:partition-function-resource-model}.

\medskip

\noindent
Then Algorithm~\ref{alg:finite-temperature-reconstruction} outputs
\(\widehat f(\beta)\) such that
\[
\Pr\!\left[
|\widehat f(\beta)-f(\beta)|
\le\epsilon
\right]
\ge1-p.
\]
Its expected reconstruction-stage query counts satisfy
\begin{equation}
\label{eq:single-temperature-fixed-schedule-cost}
N_{\mathrm{TFD}},N_{\mathrm{BE}}
=
\widetilde{\mathcal O}\!\left(
\frac{
\sqrt{\ell}
(\ell+h\beta)
\ln(2/p)
}{\epsilon}
\right).
\end{equation}
The hidden factors are polylogarithmic in
\(\ell\), \(h\beta\), and \(1/\epsilon\), and the constant
may depend on \(C\).

\medskip

In particular, if
\begin{equation}
\label{eq:single-temperature-extensive-regime}
h=O(n),
\qquad
\beta=\Theta(1),
\qquad
\ell=\widetilde O(\sqrt n),
\end{equation}
then
\begin{equation}
\label{eq:single-temperature-main-n-scaling}
\boxed{
N_{\mathrm{TFD}},N_{\mathrm{BE}}
=
\widetilde{\mathcal O}\!\left(
\frac{n^{5/4}\ln(2/p)}{\epsilon}
\right).}
\end{equation}
On the same event, the estimate
\(\widehat Z(\beta):=e^{\widehat f(\beta)}\) satisfies
\[
e^{-\epsilon}Z(\beta)
\le
\widehat Z(\beta)
\le
e^\epsilon Z(\beta).
\]

\end{thm}

\begin{proof}[Proof of Theorem~\ref{thm:single-temperature-main}]
Fix one repetition \(r\) of
Algorithm~\ref{alg:finite-temperature-reconstruction} before the final
median. Let \(\widetilde f(\beta)\) denote the deterministic reconstruction
obtained from Eq.~\eqref{eq:full-dyadic-reconstruction} by replacing each
\(\ln R(\beta_i,\delta_i)\) with \(\ln p_{i,0}+2\), where
\(p_{i,0}\) is the implemented probability from
Eq.~\eqref{eq:implemented-local-probabilities}. The overlap terms
are unchanged.

\medskip

\noindent
For an overlap term, Eq.~\eqref{eq:uniform-conditioning} gives
\[
p_{i,j}
=
O(\beta_i,2^j\delta_i)
\ge
\gamma_0,
\qquad
j\ge1.
\]
For a reweighting term,
Lemma~\ref{lem:single-temperature-dilation-stability}, together with the
choice of \(\nu\) in
Algorithm~\ref{alg:finite-temperature-reconstruction}, gives
\[
p_{i,0}\ge\frac{\gamma_0}{2}.
\]
Hence every projection probability estimated by the algorithm is at least
\(\gamma:=\gamma_0/2=\Omega(1)\). Since \(C=O(1)\), we also have
\(K_\gamma=O(1)\). Moreover,
Eq.~\eqref{eq:single-temperature-total-dilation-bias} gives
\[
|\widetilde f(\beta)-f(\beta)|
\le
\frac{\epsilon}{16}.
\]

\medskip

\noindent
The algorithm chooses
\(\eta_{i,j}=M_{i,j}^{-2}\).
Lemma~\ref{lem:almost-unbiased-local} therefore gives bias and variance
at most \(2K_\gamma/M_{i,j}^2\) for every local logarithmic estimate,
relative to the corresponding term in \(\widetilde f(\beta)\).
Since all local calls within one reconstruction are independent, the
triangle inequality and additivity of variance give
\begin{align*}
\left|
\mathbb E[\widehat f^{(r)}(\beta)]-\widetilde f(\beta)
\right|
&\le
2K_\gamma
\sum_{i,j}\frac{a_{i,j}}{M_{i,j}^2},\\
\operatorname{Var}(\widehat f^{(r)}(\beta))
&\le
2K_\gamma
\sum_{i,j}\frac{a_{i,j}^2}{M_{i,j}^2}.
\end{align*}

\medskip

\noindent
Let
\(S=\sum_{i,j}a_{i,j}^{2/3}\), as in
Eq.~\eqref{eq:single-temperature-depth-S}, and set
\(\epsilon_{\rm stat}=\epsilon/2\).
The depth choice derived in
Appendix~\ref{app:local-depth-allocation} satisfies
\[
M_{i,j}
\ge
\frac{
4\sqrt{K_\gamma}\,S^{1/2}a_{i,j}^{2/3}
}{
\epsilon_{\rm stat}
}.
\]
Substitution into the bias bound gives
\[
\left|
\mathbb E[\widehat f^{(r)}(\beta)]-\widetilde f(\beta)
\right|
\le
\frac{\epsilon_{\rm stat}^2}{8S}
\sum_{i,j}a_{i,j}^{-1/3}
\le
\frac{\epsilon_{\rm stat}^2}{8}
\le
\frac{\epsilon_{\rm stat}}8,
\]
where we used \(a_{i,j}\ge1\), and hence
\(a_{i,j}^{-1/3}\le a_{i,j}^{2/3}\).
The same depth choice gives
\[
\operatorname{Var}(\widehat f^{(r)}(\beta))
\le
\frac{\epsilon_{\rm stat}^2}{8S}
\sum_{i,j}a_{i,j}^{2/3}
=
\frac{\epsilon_{\rm stat}^2}{8}.
\]

\medskip

\noindent
Chebyshev's inequality now gives
\[
\Pr\!\left[
\left|
\widehat f^{(r)}(\beta)
-
\mathbb E[\widehat f^{(r)}(\beta)]
\right|
>
\frac{7\epsilon_{\rm stat}}8
\right]
\le
\frac{\epsilon_{\rm stat}^2/8}
{(7\epsilon_{\rm stat}/8)^2}
=
\frac8{49}
<
\frac13.
\]
Therefore, with probability greater than \(2/3\),
\[
\left|
\widehat f^{(r)}(\beta)-\widetilde f(\beta)
\right|
\le
\epsilon_{\rm stat}.
\]
Combining this event with
Eq.~\eqref{eq:single-temperature-total-dilation-bias} and using the
triangle inequality gives
\[
\left|
\widehat f^{(r)}(\beta)-f(\beta)
\right|
\le
\epsilon_{\rm stat}+\frac{\epsilon}{16}
=
\frac{\epsilon}{2}+\frac{\epsilon}{16}
<
\epsilon.
\]
Thus one reconstruction succeeds with probability greater than \(2/3\).

\medskip

\noindent
We next bound the query cost of one reconstruction. The rounding analysis in
Appendix~\ref{app:local-depth-allocation} gives
\[
\sum_{i,j}M_{i,j}
=
O\!\left(\frac{S^{3/2}}{\epsilon}\right).
\]
It remains to express \(S\) in terms of the supplied schedule. For each
interval \(i\), the definition of the coefficients \(a_{i,j}\) gives
\[
\sum_{j=0}^{k_i}a_{i,j}^{2/3}
=
2^{2k_i/3}
+
\sum_{j=1}^{k_i}2^{2(k_i-j)/3}
=
O(2^{2k_i/3})
=
O\!\left(
\max\{1,h\Delta_i\}^{2/3}
\right),
\]
where the middle bound is a geometric sum and the last bound follows from
the definition of \(k_i\). Hence H\"older's inequality gives
\[
S^{3/2}
=
O\!\left(
\left[
\sum_i\max\{1,h\Delta_i\}^{2/3}
\right]^{3/2}
\right)
\le
O\!\left(
\sqrt{\ell}
\sum_i\max\{1,h\Delta_i\}
\right).
\]
Since \(\max\{1,x\}\le1+x\) for \(x\ge0\) and
\(\sum_i\Delta_i=\beta\),
\[
S^{3/2}
=
O\!\left(
\sqrt{\ell}(\ell+h\beta)
\right).
\]

\medskip

\noindent
Lemma~\ref{lem:almost-unbiased-local} gives
\(\widetilde O(M_{i,j})\) fixed-temperature TFD calls for each local term.
Overlap terms use no block-encoding queries, while reweighting terms use
\(\widetilde O(M_{i,0})\) block-encoding queries. Therefore, one
repetition uses
\[
\widetilde{\mathcal O}\!\left(
\frac{
\sqrt{\ell}(\ell+h\beta)
}{\epsilon}
\right)
\]
TFD queries and the same asymptotic number of block-encoding queries.

\medskip

\noindent
It remains to amplify the success probability. Each independent repetition
fails with probability at most \(1/3\). If more than half of the repetitions
succeed, then their median also has error at most \(\epsilon\). Therefore,
the median can fail only if at least half of the \(N_{\mathrm{rep}}\)
repetitions fail. Hoeffding's inequality gives
\[
\Pr[\textnormal{median fails}]
\le
\exp\!\left[
-2N_{\mathrm{rep}}
\left(\frac12-\frac13\right)^2
\right]
=
e^{-N_{\mathrm{rep}}/18}.
\]
For the choice
\[
N_{\mathrm{rep}}
=
2\left\lceil9\ln(2/p)\right\rceil+1
\]
in Algorithm~\ref{alg:finite-temperature-reconstruction}, we have
\(N_{\mathrm{rep}}\ge18\ln(2/p)\), and hence
\[
e^{-N_{\mathrm{rep}}/18}
\le
\frac{p}{2}
<
p.
\]
Multiplying the query cost of one repetition by
\(N_{\mathrm{rep}}=O(\ln(2/p))\) proves
Eq.~\eqref{eq:single-temperature-fixed-schedule-cost}.

\medskip

\noindent
Finally, under
Eq.~\eqref{eq:single-temperature-extensive-regime},
\(\sqrt{\ell}=\widetilde O(n^{1/4})\) and
\(\ell+h\beta=\widetilde O(n)\).
This proves
Eq.~\eqref{eq:single-temperature-main-n-scaling}.
On the event
\(|\widehat f(\beta)-f(\beta)|\le\epsilon\), exponentiation gives
\[
e^{-\epsilon}Z(\beta)
\le
\widehat Z(\beta)
\le
e^\epsilon Z(\beta).
\]
\end{proof}

\section{Partition-function estimation with coherent-temperature access}
\label{sec:sampled-dyadic}

In this section the exact reconstruction formula
\eqref{eq:full-dyadic-reconstruction} is unchanged.
The difference from the fixed-temperature method is how the local
logarithmic terms are combined. Instead of estimating every local term
separately and then forming their weighted sum, we introduce a distribution
over the terms with probabilities proportional to their reconstruction
coefficients \(a_{i,j}\). This rewrites the same reconstruction as an
expectation and allows us to apply quantum mean estimation directly to the
weighted sum.

\subsection{Coherent branch construction}
\label{subsec:coherent-branch-construction}

Using the parameters in Eq.~\eqref{eq:dyadic-parameters} and the
reconstruction coefficients \(a_{i,j}\), define
\begin{equation}
\label{eq:coherent-branch-data}
\mathcal R
:=
\{(i,j):0\le i<\ell,\ 0\le j\le k_i\},
\qquad
L
:=
\sum_{(i,j)\in\mathcal R}a_{i,j},
\qquad
\pi_{i,j}
:=
\frac{a_{i,j}}{L}.
\end{equation}
Then Eq.~\eqref{eq:weighted-log-reconstruction} becomes
\begin{equation}
\label{eq:coherent-sampled-identity}
f(\beta)-f(0)
=
\sum_{(i,j)\in\mathcal R}a_{i,j}Y_{i,j}
=
L\,\mathbb E_{(I,J)\sim\pi}[Y_{I,J}].
\end{equation}
The identity above expresses the reconstruction as the mean of the exact
local logarithmic terms \(Y_{I,J}\) under the distribution \(\pi\).
The coherent algorithm does not evaluate \(Y_{i,j}\) directly. Instead, it
constructs a branch-dependent random variable whose expectation approximates
\(Y_{i,j}\), with the normalization correction included on each reweighting
branch. Quantum mean estimation is then applied to the weighted sample.

\begin{defn}[Coherent-temperature access model]
\label{def:coherent-temperature-access}
For the supplied schedule and the branch distribution in
Eq.~\eqref{eq:coherent-branch-data}, assume access to the following
unitaries and their adjoints.

\begin{enumerate}
    \item A parameter-controlled TFD preparation unitary
    \(U_{\mathrm{cohTFD}}\) satisfying
    \[
    U_{\mathrm{cohTFD}}
    \ket{\theta}\ket0_Q
    =
    \ket{\theta}\ket{\psi_\theta}_Q
    \]
    for every inverse temperature \(\theta\) used in the reconstruction.
    The same unitary acts coherently when the inverse-temperature register
    is in superposition. Define
    \[
    N_{\mathrm{cohTFD}}
    :=
    \text{the total number of calls to }
    U_{\mathrm{cohTFD}}
    \text{ and }
    U_{\mathrm{cohTFD}}^\dagger .
    \]
    Each forward or adjoint call counts as one query.

    \item A branch-preparation unitary \(U_\pi\) and a reversible
    parameter-loading unitary \(U_{\mathrm{load}}\) satisfying
    \[
    U_\pi\ket0
    =
    \sum_{(i,j)\in\mathcal R}
    \sqrt{\pi_{i,j}}\ket{i,j},
    \qquad
    U_{\mathrm{load}}\ket{i,j}\ket0
    =
    \ket{i,j}\ket{\mathrm{par}_{i,j}},
    \]
    where \(\mathrm{par}_{i,j}\) contains the inverse-temperature
    parameters required by the branch construction. The coherent
    preparation and loading of these registers are described in
    Appendix~\ref{app:sampled-dyadic},
    Subsection~\ref{subsec:coherent-branch-preparation}.

    \item The Hamiltonian block-encoding access, including adjoints and
    controlled versions, from
    Definition~\ref{def:partition-function-resource-model}. The
    corresponding query count \(N_{\mathrm{BE}}\) is defined as in that
    definition.
\end{enumerate}

\medskip

\end{defn}

\medskip

The two branch types implement the overlap and reweighting tests
associated with
\(Y_{i,j}\) in Eq.~\eqref{eq:coherent-sampled-identity}. For \(j\ge1\),
coherently prepare the TFD state at
\(\beta_i+2^j\delta_i\), apply the inverse TFD preparation at
\(\beta_i\), and test the all-zero TFD registers. The resulting projection
weight is \(O(\beta_i,2^j\delta_i)\). For \(j=0\), coherently prepare the
TFD state at \(\beta_i\), apply the branch-controlled QSVT unitary
\(D_{\delta_i}\), and test the all-zero QSVT ancilla register. In the ideal
case, the resulting weight is
\(e^{-2}\Tr(\rho_{\beta_i}e^{-\delta_iH})\); the QSVT approximation error is
handled separately.

\medskip

The register-level implementations of the two branch types, together with
their coherent-TFD and block-encoding query counts, are given in
Appendix~\ref{app:sampled-dyadic},
Subsection~\ref{subsec:coherent-branch-implementation}.

\subsection{Coherent logarithmic sampling}
\label{subsec:coherent-logarithmic-sampling}

Equation~\eqref{eq:coherent-sampled-identity} reduces the reconstruction to
the mean of the exact local logarithmic terms \(Y_{I,J}\) under
\(\pi\). The coherent circuit does not evaluate \(Y_{i,j}\) directly.
Instead, for each branch \((i,j)\), it produces the corresponding overlap
or reweighting probability. The main technical step is therefore to convert
this probability into a random variable whose expectation approximates the
required logarithm.

\medskip

By Lemma~\ref{lem:local-quantity-projection} and
Eq.~\eqref{eq:implemented-local-probabilities}, the circuit associated
with each branch \((i,j)\) produces the projection probability
\(p_{i,j}\). The coherent logarithmic construction below requires repeated
Bernoulli trials with this probability while keeping the branch registers
\(I,J\) coherent. We therefore combine the overlap and reweighting circuits
into a single branch-controlled unitary \(U_{\mathrm{branch}}\).
The explicit register-level construction of \(U_{\mathrm{branch}}\) is given
in Appendix~\ref{app:sampled-dyadic},
Subsection~\ref{subsec:coherent-branch-implementation}. Conditioned on a
branch \((i,j)\in\mathcal R\), it acts as
\begin{equation}
\label{eq:coherent-branch-flag-probability}
\begin{aligned}
&U_{\mathrm{branch}}
\ket{i,j}_{IJ}\ket{\mathrm{par}_{i,j}}_P
\ket0_Q\ket{0^a}_A\ket0_F
\\
&\qquad=
\ket{i,j}_{IJ}\ket{\mathrm{par}_{i,j}}_P
\left(
\sqrt{p_{i,j}}\,
\ket{\phi^1_{i,j}}_{QA}\ket1_F
+
\sqrt{1-p_{i,j}}\,
\ket{\phi^0_{i,j}}_{QA}\ket0_F
\right),
\end{aligned}
\end{equation}
where \(\ket{\phi^1_{i,j}}_{QA}\) and
\(\ket{\phi^0_{i,j}}_{QA}\) are normalized states on \(Q\otimes A\).
Their explicit form will not be needed. Hence
\(\Pr(F=1\mid i,j)=p_{i,j}\), while the branch and parameter registers are
unchanged. One application of \(U_{\mathrm{branch}}\) uses at most two
calls to \(U_{\mathrm{cohTFD}}\) or its adjoint, no fixed-temperature TFD call, and
\(\widetilde{\mathcal O}(1)\) block-encoding queries.

\medskip

For the overlap branches, Eq.~\eqref{eq:uniform-conditioning} gives
\(p_{i,j}\ge\gamma_0\). For the reweighting branches,
Lemma~\ref{lem:single-temperature-dilation-stability} gives
\(p_{i,0}\ge\gamma_0/2\) whenever \(2\nu\le\gamma_0/2\).
The choice of \(\nu\) in
Algorithm~\ref{alg:sampled-dyadic-reconstruction} satisfies this condition.
Hence every branch probability obeys
\[
p_{i,j}\ge\gamma:=\gamma_0/2=\Omega(1).
\]
Fix a branch and denote its success probability by
\(\mu\in[\gamma,1]\). We use the expansion
\[
-\ln\mu
=
\sum_{m=1}^{\infty}\frac{(1-\mu)^m}{m}.
\]
This suggests estimating the logarithm from repeated coherent trials with
success probability \(\mu\). The following lemma shows that a truncated
harmonic sample constructed from these trials approximates \(-\ln\mu\)
while keeping its second moment controlled.

\begin{lem}[Coherent logarithmic sampling]
\label{lem:coherent-logarithmic-sampling-summary}
Fix a branch of \(U_{\mathrm{branch}}\) with success probability
\(\mu\in[\gamma,1]\), and let \(M\ge1\). Apply
\(U_{\mathrm{branch}}\) \(M\) times to fresh workspace and flag registers,
without measurement, and let
\(b_1,\ldots,b_M\in\{0,1\}\) denote the corresponding success flags.
Let \(G_M\) be the number of initial failures before the first success,
truncated at \(M\). Define
\[
H_{G_M}
:=
\sum_{m=1}^{G_M}\frac1m
=
\sum_{m=1}^{M}
\frac1m
\prod_{t=1}^{m}(1-b_t),
\qquad
H_0:=0.
\]
Then
\begin{equation}
\label{eq:coherent-log-sample-bounds}
\begin{aligned}
\mathbb E[H_{G_M}]
&=
\sum_{m=1}^{M}\frac{(1-\mu)^m}{m},\\
0\le
-\ln\mu-\mathbb E[H_{G_M}]
&\le
\frac{e^{-\gamma(M+1)}}{\gamma(M+1)},\\
\mathbb E[H_{G_M}^{\,2}]
&\le
H_M\ln(1/\gamma),
\end{aligned}
\end{equation}
where \(H_M=\sum_{m=1}^{M}1/m\) is the \(M\)-th harmonic number.
Moreover, \(H_{G_M}\) can be computed reversibly from the \(M\) unmeasured
flags. Its preparation uses \(M\) applications of
\(U_{\mathrm{branch}}\), and its uncomputation uses \(M\) applications of
\(U_{\mathrm{branch}}^\dagger\).
\end{lem}

\medskip

Increasing \(M\) reduces the bias in the estimate of \(-\ln\mu\),
while the second-moment bound controls the cost of mean estimation.
We also need to include the error from computing the sample with finite
precision.
The proof of Lemma~\ref{lem:coherent-logarithmic-sampling-summary} and the
reversible computation of \(H_{G_M}\) are given in
Appendix~\ref{app:sampled-dyadic},
Subsection~\ref{subsec:reversible-harmonic-sample}.

\medskip

The reversible arithmetic uses a fixed-point approximation
\(\widetilde H_{G_M}\) satisfying
\(|\widetilde H_{G_M}-H_{G_M}|\le\zeta\). Hence
\begin{equation}
\label{eq:coherent-fixed-point-log-bounds}
\begin{aligned}
\left|
\mathbb E[\widetilde H_{G_M}]+\ln\mu
\right|
&\le
\frac{e^{-\gamma(M+1)}}{\gamma(M+1)}+\zeta,\\
\mathbb E[\widetilde H_{G_M}^{\,2}]
&\le
2H_M\ln(1/\gamma)+2\zeta^2.
\end{aligned}
\end{equation}

\medskip

Define the local logarithmic sample by
\[
\widehat Y_{i,j}:=
\begin{cases}
2-\widetilde H_{G_M}, & j=0,\\
\widetilde H_{G_M}, & j\ge1.
\end{cases}
\]
For \(j\ge1\), its expectation approximates
\(Y_{i,j}=-\ln O(\beta_i,2^j\delta_i)\). For \(j=0\), the constant
\(2\) compensates for the QSVT normalization, so its expectation
approximates \(Y_{i,0}=B(\beta_i,\delta_i)\), with the remaining QSVT
error controlled by Lemma~\ref{lem:single-temperature-dilation-stability}.

\medskip

We now have a logarithmic sample for each branch. To estimate the full
weighted sum by quantum mean estimation, we combine these samples into
the global random variable
\(X=L\widehat Y_{I,J}\), where \((I,J)\sim\pi\).
The following lemma constructs its coherent preparation and the inverse
needed for mean estimation.
\begin{lem}[Full coherent sample preparation]
\label{lem:full-coherent-sample-preparation}
Let \(\mathcal R\) be the branch set,
\(\pi=(\pi_{i,j})_{(i,j)\in\mathcal R}\) its sampling distribution,
and \(L>0\) the reconstruction scale, as defined in
Eq.~\eqref{eq:coherent-branch-data}.
Assume that coherent preparation of the branch distribution and
reversible loading of the branch parameters are available.
For each \((i,j)\in\mathcal R\), let \(p_{i,j}\) be the success
probability of \(U_{\mathrm{branch}}\), which leaves the branch
index and parameter registers unchanged.

\medskip

Fix an integer \(M\ge1\) and a fixed-point error tolerance \(\zeta>0\).
For a string of success flags
\(\mathbf b=(b_1,\ldots,b_M)\in\{0,1\}^M\), where \(b_m=1\)
denotes success on the \(m\)-th trial, define
\[
H(\mathbf b)
:=
\sum_{m=1}^{M}\frac1m\prod_{t=1}^{m}(1-b_t).
\]
Let \(\widetilde H(\mathbf b)\) be a reversibly computable
fixed-point approximation satisfying
\[
\bigl|\widetilde H(\mathbf b)-H(\mathbf b)\bigr|
\le\zeta
\qquad
\text{for every }\mathbf b\in\{0,1\}^M.
\]
Define the conditional flag distribution and the sample-value map by
\[
q_{i,j}(\mathbf b)
:=
\prod_{m=1}^{M}
p_{i,j}^{\,b_m}(1-p_{i,j})^{1-b_m},
\qquad
x_j(y)
:=
\begin{cases}
L(2-y), & j=0,\\
Ly, & j\ge1,
\end{cases}
\qquad y\in\mathbb R.
\]

\medskip

There is a unitary \(U_{\mathrm{sample}}\) satisfying
\[
U_{\mathrm{sample}}\ket0
=
\sum_{(i,j)\in\mathcal R}
\sum_{\mathbf b\in\{0,1\}^M}
\sqrt{\pi_{i,j}q_{i,j}(\mathbf b)}\,
\ket{\mathrm{work}_{i,j,\mathbf b}}
\ket{x_j(\widetilde H(\mathbf b))}_X.
\]
Here \(\ket0\) is the joint all-zero initial state, \(X\) labels
the sample-value register, and
\(\ket{\mathrm{work}_{i,j,\mathbf b}}\) contains all other registers.
For nonzero summands, these workspace states are normalized and
mutually orthogonal, since they retain the branch index and all flags.

\medskip

Measuring the value register produces the random variable
\[
X=L\widehat Y_{I,J},
\qquad (I,J)\sim\pi,
\]
where, conditional on \((I,J)=(i,j)\), one has
\(\mathbf b\sim q_{i,j}\) and
\(\widehat Y_{i,j}=x_j(\widetilde H(\mathbf b))/L\).
No measurement is performed during the preparation.

\medskip

The preparation uses one coherent branch preparation, one parameter
loading, and \(M\) applications of \(U_{\mathrm{branch}}\), together
with reversible arithmetic.
Its inverse reverses all these operations and uses \(M\)
applications of \(U_{\mathrm{branch}}^\dagger\).
\end{lem}

\medskip

The proof and register-level construction are given in
Appendix~\ref{app:sampled-dyadic},
on page~\pageref{proof:full-coherent-sample-preparation}.
\paragraph{Quantum mean estimation.}
The construction above provides the coherent sampler required by the
bounded-variance quantum mean-estimation procedure of
Montanaro~\cite[Algorithm~3 and Theorem~5]{Montanaro2015}.
For either branch, \(\widehat Y_{i,j}^{\,2}\le
2\widetilde H_{G_M}^{\,2}+8\), using
\((2-y)^2\le2y^2+8\) for the reweighting branch.
Thus Eq.~\eqref{eq:coherent-fixed-point-log-bounds} gives the known bound
\[
\operatorname{Var}(X)
\le
\mathbb E[X^2]
\le
L^2\!\left(4H_M\ln(1/\gamma)+4\zeta^2+8\right)
=:\sigma_X^2.
\]
The encoded value of \(X\) is available for the reversible arithmetic,
comparisons, and controlled rotations used by the procedure.
Writing \(\ket{\Psi_X}:=U_{\mathrm{sample}}\ket0\), the reflection about
the prepared sample state is implemented as
\[
2\ket{\Psi_X}\!\langle\Psi_X|-I
=
U_{\mathrm{sample}}
\bigl(2\ket0\!\langle0|-I\bigr)
U_{\mathrm{sample}}^\dagger,
\]
where \(\ket0\) denotes the joint initial state of all preparation
registers. Thus no additional sample-state-reflection oracle is required.
For \(0<\eta\le\sigma_X\), the procedure estimates
\(\mathbb E[X]\) to additive error \(\eta\), with success probability
at least \(2/3\), using
\[
\widetilde O\!\left(\frac{\sigma_X}{\eta}\right)
\]
applications of \(U_{\mathrm{sample}}\) and
\(U_{\mathrm{sample}}^\dagger\).

\medskip

The choices of \(M\) and \(\zeta\), and the application of this
mean-estimation procedure to the full reconstruction, are given in
Subsection~\ref{subsec:coherent-reconstruction},
Algorithm~\ref{alg:sampled-dyadic-reconstruction}.

\subsection{Reconstruction algorithm and complexity}
\label{subsec:coherent-reconstruction}

The preceding construction provides a coherent sample \(X\) whose expectation
approximates \(f(\beta)-f(0)\). We now combine this sample with
bounded-variance quantum mean estimation in
Algorithm~\ref{alg:sampled-dyadic-reconstruction} to obtain the full
partition-function reconstruction procedure.
Theorem~\ref{thm:coherent-indexed-main} proves the correctness of this
algorithm and gives its coherent-TFD and block-encoding query complexities.
In the extensive constant-temperature regime, both query counts scale as
\(\widetilde O(n/\epsilon)\) for constant failure probability.

\begin{thm}[Coherent-temperature partition-function estimation]
\label{thm:coherent-indexed-main}
\label{thm:sampled-dyadic-main}
Let \(H\) be an \(n\)-qubit Hamiltonian satisfying
\(0\preceq H\preceq hI\) for some \(h>0\). Let
\(0<\epsilon\le1\) and \(0<p<1\). Suppose that a
\(C\)-slowly-varying schedule
\(0=\beta_0<\cdots<\beta_{\ell}=\beta\) is supplied, where
\(C=O(1)\) and \(0<\beta<\infty\). Assume an \(\alpha\)-block encoding of
\(H\) with \(\alpha=\Theta(h)\) and the coherent-temperature access of
Definition~\ref{def:coherent-temperature-access}.

\medskip

Algorithm~\ref{alg:sampled-dyadic-reconstruction} outputs
\(\widehat f(\beta)\) such that
\[
\Pr\!\left[
|\widehat f(\beta)-f(\beta)|
\le\epsilon
\right]
\ge1-p.
\]
Its expected reconstruction-stage query counts satisfy
\begin{equation}
\label{eq:coherent-fixed-schedule-resources}
N_{\mathrm{cohTFD}},\,N_{\mathrm{BE}}
=
\widetilde{\mathcal O}\!\left(
\frac{(\ell+h\beta)\ln(2/p)}{\epsilon}
\right).
\end{equation}
In particular, if
\[
h=O(n),
\qquad
\beta=\Theta(1),
\qquad
\ell=\widetilde O(\sqrt n),
\]
then
\begin{equation}
\label{eq:sampled-main-n-scaling}
\boxed{
N_{\mathrm{cohTFD}},\,N_{\mathrm{BE}}
=
\widetilde{\mathcal O}\!\left(
\frac{n\ln(2/p)}{\epsilon}
\right).
}
\end{equation}
For constant \(p\), both nonzero query counts are
\(\widetilde{\mathcal O}(n/\epsilon)\).

\end{thm}

\begin{proof}
Let \((I,J)\sim\pi\), and define
\[
\widetilde Y_{i,j}
:=
\begin{cases}
\ln p_{i,0}+2, & j=0,\\
-\ln p_{i,j}, & j\ge1.
\end{cases}
\]
For \(j\ge1\), Eq.~\eqref{eq:implemented-local-probabilities} gives
\(p_{i,j}=O(\beta_i,2^j\delta_i)\), and hence
\(\widetilde Y_{i,j}=Y_{i,j}\). For \(j=0\),
Lemma~\ref{lem:single-temperature-dilation-stability} gives
\[
p_{i,0}\ge\frac{\gamma_0}{2}=\gamma,
\qquad
|\widetilde Y_{i,0}-Y_{i,0}|
=
|\ln p_{i,0}+2-\ln R(\beta_i,\delta_i)|
\le
\frac{4\nu}{\gamma_0},
\]
where we used the choice of \(\nu\) in
Algorithm~\ref{alg:sampled-dyadic-reconstruction}.
For the overlap branches,
Eq.~\eqref{eq:uniform-conditioning} gives
\(p_{i,j}\ge\gamma_0\). Thus every branch probability satisfies
\(p_{i,j}\ge\gamma\).

\medskip

Only the reweighting branches contribute to the QSVT approximation error.
Since \(\sum_i a_{i,0}\le L\),
\[
\sum_i a_{i,0}
|\widetilde Y_{i,0}-Y_{i,0}|
\le
\frac{4\nu}{\gamma_0}\sum_i a_{i,0}
\le
\frac{4\nu L}{\gamma_0}
=
\frac{\epsilon}{16}.
\]

\medskip

Let \(\widehat Y_{i,j}\) denote the fixed-point logarithmic output of the
coherent sampler, including the correction on reweighting branches. By
Eq.~\eqref{eq:coherent-fixed-point-log-bounds}, the inequality
\((2-y)^2\le2y^2+8\), and the choices of
\(M\) and \(\zeta\) in
Algorithm~\ref{alg:sampled-dyadic-reconstruction},
\[
\left|
\mathbb E[\widehat Y_{i,j}]-\widetilde Y_{i,j}
\right|
\le
\frac{\epsilon}{8L},
\qquad
\mathbb E[\widehat Y_{i,j}^{\,2}]
\le
4H_M\ln(1/\gamma)+4\zeta^2+8.
\]
The first bound combines the logarithmic truncation error
\(\epsilon/(16L)\) and the arithmetic error
\(\zeta=\epsilon/(16L)\). Moreover, \(\gamma=\Omega(1)\) and
\(M=O(\log(L/\epsilon))\), so
\(\mathbb E[\widehat Y_{i,j}^{\,2}]=\widetilde O(1)\) uniformly over all
branches.

\medskip

Now define the global coherent sample
\(X:=L\widehat Y_{I,J}\). Since
\(\pi_{i,j}=a_{i,j}/L\),
\[
\mathbb E[X]
=
\sum_{(i,j)\in\mathcal R}
a_{i,j}\mathbb E[\widehat Y_{i,j}],
\]
whereas Eq.~\eqref{eq:coherent-sampled-identity} gives
\[
f(\beta)-f(0)
=
\sum_{(i,j)\in\mathcal R}a_{i,j}Y_{i,j}.
\]
Therefore
\begin{align*}
\left|
\mathbb E[X]-(f(\beta)-f(0))
\right|
&\le
\sum_{(i,j)\in\mathcal R}
a_{i,j}
\left|
\mathbb E[\widehat Y_{i,j}]-\widetilde Y_{i,j}
\right|
+
\sum_i a_{i,0}
|\widetilde Y_{i,0}-Y_{i,0}|\\
&\le
L\frac{\epsilon}{8L}
+
\frac{\epsilon}{16}
=
\frac{3\epsilon}{16}
<
\frac{\epsilon}{4}.
\end{align*}
The second moment satisfies
\[
\mathbb E[X^2]
=
L^2\sum_{(i,j)\in\mathcal R}
\pi_{i,j}\mathbb E[\widehat Y_{i,j}^{\,2}]
\le
L^2\max_{i,j}\mathbb E[\widehat Y_{i,j}^{\,2}]
=
\widetilde O(L^2),
\]
and hence
\(\operatorname{Var}(X)=\widetilde O(L^2)\).

\medskip

Bounded-variance quantum mean estimation
\cite[Theorem~5]{Montanaro2015} therefore estimates
\(\mathbb E[X]\) to additive error \(\epsilon/2\) using
\(\widetilde O(L/\epsilon)\) coherent sample-preparation calls and their
inverses. Combining this with the preceding bias gives
\[
\left|
\widehat g-(f(\beta)-f(0))
\right|
\le
\frac{\epsilon}{2}+\frac{3\epsilon}{16}
<
\epsilon.
\]
Adding the known baseline \(f(0)=n\ln2\), one run produces an
\(\epsilon\)-accurate estimate of \(f(\beta)\) with constant success
probability. Using the
\(N_{\mathrm{rep}}\) independent repetitions in
Algorithm~\ref{alg:sampled-dyadic-reconstruction} and taking their median
reduces the failure probability to at most \(p\).

\medskip

It remains to count oracle queries. By
Subsection~\ref{subsec:coherent-branch-implementation}, one application of
\(U_{\mathrm{branch}}\) uses at most two coherent-TFD calls and
\(\widetilde O(1)\) block-encoding queries. One logarithmic sample uses
\(M\) such branch applications, and its inverse has the same asymptotic
cost. Thus one sample-state reflection uses \(O(M)\) coherent-TFD calls
and \(\widetilde O(M)\) block-encoding queries. Since
\(M=O(\log(L/\epsilon))\), this factor is absorbed into
\(\widetilde O\). Together with quantum mean estimation and the
\(N_{\mathrm{rep}}=O(\ln(2/p))\) repetitions,
\[
N_{\mathrm{TFD}}=0,
\qquad
N_{\mathrm{cohTFD}},\,N_{\mathrm{BE}}
=
\widetilde{\mathcal O}\!\left(
\frac{L\ln(2/p)}{\epsilon}
\right).
\]

\medskip

Finally, using \(a_{i,j}=2^{k_i-j}\) and
\(k_i=\max\{0,\lceil\log_2(h\Delta_i)\rceil\}\),
\begin{align*}
L
&=
\sum_{i=0}^{\ell-1}
\sum_{j=0}^{k_i}2^{k_i-j}
=
\sum_{i=0}^{\ell-1}
(2^{k_i+1}-1)\\
&=
O\!\left(
\sum_{i=0}^{\ell-1}
\max\{1,h\Delta_i\}
\right)
=
O(\ell+h\beta),
\end{align*}
where \(\sum_i\Delta_i=\beta\). This proves
Eq.~\eqref{eq:coherent-fixed-schedule-resources}.

\medskip

Under
\(h=O(n)\), \(\beta=\Theta(1)\), and
\(\ell=\widetilde O(\sqrt n)\),
we have \(L=\widetilde O(n)\). Therefore
\[
N_{\mathrm{TFD}}=0,
\qquad
N_{\mathrm{cohTFD}},\,N_{\mathrm{BE}}
=
\widetilde{\mathcal O}\!\left(
\frac{n\ln(2/p)}{\epsilon}
\right),
\]
which proves Eq.~\eqref{eq:sampled-main-n-scaling}.
\end{proof}
\begin{algorithm}[H]
\caption{Partition-function reconstruction with coherent-temperature access}
\label{alg:sampled-dyadic-reconstruction}
\small
\begin{algorithmic}[1]
\Require An \(n\)-qubit Hamiltonian \(H\) with
\(0\preceq H\preceq hI\) and \(h>0\); an \(\alpha\)-block encoding of
\(H\) with \(\alpha=\Theta(h)\); a supplied \(C\)-slowly-varying schedule
\(0=\beta_0<\cdots<\beta_{\ell}=\beta\); the
coherent-temperature access of
Definition~\ref{def:coherent-temperature-access};
\(0<\epsilon\le1\); and \(0<p<1\).
\Ensure An estimate \(\widehat f(\beta)\) satisfying
\(\lvert\widehat f(\beta)-\ln Z(\beta)\rvert\le\epsilon\)
with probability at least \(1-p\).

\State Compute \(\Delta_i\), \(k_i\), and \(\delta_i\) from
Eq.~\eqref{eq:dyadic-parameters}. Set \(a_{i,j}=2^{k_i-j}\) for
\(0\le j\le k_i\), and compute \(\mathcal R\), \(L\), and \(\pi_{i,j}\)
from Eq.~\eqref{eq:coherent-branch-data}.

\State Set
\[
\begin{gathered}
\gamma_0\gets\min\{e^{-3},1/C\},
\qquad
\gamma\gets\gamma_0/2,
\qquad
\nu\gets\frac{\gamma_0\epsilon}{64L},
\qquad
\zeta\gets\frac{\epsilon}{16L},
\\
M\gets
\left\lceil
\frac{1}{\gamma}
\ln\!\left(\frac{16L}{\gamma\epsilon}\right)
\right\rceil.
\end{gathered}
\]

\State Construct \(U_{\mathrm{sample}}\) and its inverse as in
Lemma~\ref{lem:full-coherent-sample-preparation}, using branch
distribution \(\pi\), scale \(L\), truncation length \(M\),
fixed-point error at most \(\zeta\), and QSVT approximation error
\(\nu\) on the reweighting branches. The encoded sample is
\(X=L\widehat Y_{I,J}\), where \((I,J)\sim\pi\).

\State Set the variance bound
\[
\sigma_X^2
\gets
L^2\!\left(4H_M\ln(1/\gamma)+4\zeta^2+8\right),
\]
where \(H_M=\sum_{m=1}^{M}1/m\).

\State Set
\[
N_{\mathrm{rep}}
\gets
2\left\lceil9\ln(2/p)\right\rceil+1.
\]

\For{\(r=1,\ldots,N_{\mathrm{rep}}\)}
    \State Independently apply the bounded-variance quantum
    mean-estimation procedure of
    Montanaro~\cite[Algorithm~3 and Theorem~5]{Montanaro2015},
    using \(U_{\mathrm{sample}}\),
    \(U_{\mathrm{sample}}^\dagger\), the variance bound
    \(\sigma_X^2\), and additive target \(\epsilon/2\).
    Denote the output by \(\widehat g_r\).
\EndFor

\State \Return
\[
\widehat f(\beta)
\gets
n\ln2+
\operatorname{median}
\{\widehat g_1,\ldots,\widehat g_{N_{\mathrm{rep}}}\}.
\]
\end{algorithmic}
\end{algorithm}
\begin{cor}[Block-encoding-query optimality]
\label{cor:coherent-block-query-optimality}
Assume the hypotheses of
Theorem~\ref{thm:coherent-indexed-main}. Suppose that the supplied schedule
satisfies \(\ell=O(1+h\beta)\), that \(h\beta\ge1\), and that
\(p\in(0,1/3]\) is fixed. Under the hidden-offset model of
Theorem~\ref{thm:linear-lower-bound}, for
\[
0<\epsilon\le
\min\!\left\{1,\frac{\beta\min\{h,\alpha\}}{16}\right\},
\]
\[
N_{\mathrm{BE}}
=
\widetilde{\Theta}\!\left(\frac{h\beta}{\epsilon}\right).
\]
In particular, if \(h=\Theta(n)\) and \(\beta=\Theta(1)\), then
\(N_{\mathrm{BE}}=\widetilde{\Theta}(n/\epsilon)\).
This equality concerns the worst-case expected block-encoding query
complexity.
\end{cor}

\begin{proof}
Since \(\ell=O(1+h\beta)\) and \(h\beta\ge1\),
Theorem~\ref{thm:coherent-indexed-main} gives
\(N_{\mathrm{BE}}=\widetilde O(h\beta/\epsilon)\).
The hidden-offset lower bound gives
\(N_{\mathrm{BE}}=\Omega(\alpha\beta/\epsilon)
=\Omega(h\beta/\epsilon)\), since \(\alpha=\Theta(h)\).
\end{proof}

\begin{rem}[Uniform-grid comparison]
\label{rem:uniform-grid-comparison}
The block-encoding ratio estimator of
\cite[Appendix~D]{RouzeFrancaAlhambra2024} admits the following coherent
fixed-temperature benchmarks.\footnote{The theorem stated in that reference
uses sampling and has \(\epsilon^{-2}\) precision dependence.
The comparison here counts calls to coherent preparation and short
exponential block encodings, with their implementation costs accounted for
separately.}
For a uniform grid, set \(K:=\max\{1,\lceil h\beta\rceil\}\),
\(\theta_i:=i\beta/K\) for \(0\le i\le K\), and
\(\delta:=\beta/K\). Since \(h\delta\le1\), each step satisfies
\[
O(\theta_i,\delta)
=\frac{R(\theta_i,\delta/2)^2}{R(\theta_i,\delta)}
\ge e^{-h\delta}\ge e^{-1}.
\]
Thus the grid is \(e\)-slowly varying. Each interval has
\(k_i=0\), all reconstruction weights are one, and
\(S=\sum_{i=0}^{K-1}a_{i,0}^{2/3}=K\).
The allocation used in Theorem~\ref{thm:single-temperature-main},
based on nearly unbiased estimation~\cite{Cornelissen2023},
therefore gives \(\widetilde O(K^{3/2}/\epsilon)\) expected queries
for fixed failure probability. Estimating each logarithmic term to
accuracy \(\epsilon/K\) by amplitude estimation instead gives
\(\widetilde O(K^2/\epsilon)\) queries. For \(h=O(n)\) and constant
\(\beta\), these become \(\widetilde O(n^{3/2}/\epsilon)\) and
\(\widetilde O(n^2/\epsilon)\), respectively.

\medskip

If the coherent-temperature interface also supports the uniform-grid
temperatures at comparable per-call cost, then \(L=K\), and
Theorem~\ref{thm:coherent-indexed-main} gives
\(\widetilde O((1+h\beta)/\epsilon)\) queries for fixed failure
probability. This specialization uses only reweighting branches and a
deterministic schedule. The short-schedule construction can require fewer
distinct temperatures, and its access model only promises the temperatures
needed for that supplied schedule. Consequently, the uniform-grid
specialization is available when those additional temperatures are
supported; its gate cost must also include their coherent preparation
and loading.
\end{rem}

\section{Adaptive Cooling Schedules}
\label{sec:adaptive-cooling-schedules}

The preceding reconstruction algorithms achieve Heisenberg scaling given a
suitable slowly-varying schedule. We now explain how such schedules can be
generated and how their length is bounded. We use the adaptive procedure of
Arunachalam et al.~\cite{Arunachalam2022}, which selects each successive
inverse temperature by binary search using TFD overlap estimates.

\medskip

Following their analysis, we compare the generated schedule with a
perfectly-balanced schedule. Throughout this section,
\(f(\beta):=\ln Z(\beta)\), where
\(Z(\beta)=\operatorname{Tr}(e^{-\beta H})\).

\begin{defn}[Perfectly-balanced schedule~\cite{Arunachalam2022}]
A schedule
\[
0=\gamma_0<\gamma_1<\cdots<\gamma_{L_{\mathrm{PB}}}=\gamma
\]
is perfectly balanced if
\[
f\!\left(\frac{\gamma_i+\gamma_{i+1}}{2}\right)
=\frac{f(\gamma_i)+f(\gamma_{i+1})}{2}-1
\qquad (0\le i<L_{\mathrm{PB}}-1),
\]
with equality replaced by \(\ge\) on the final interval.
Thus every nonfinal interval has squared TFD overlap \(e^{-2}\),
and the final interval has squared overlap at least \(e^{-2}\).
Here \(L_{\mathrm{PB}}\) is the total number of intervals.
\end{defn}

\medskip

The following theorem gives the length bounds needed below.

\begin{thm}[Perfectly-balanced schedule length,
adapted from~\cite{Stefankovic2009}]
\label{thm:adaptive_schedule_existence}
\label{thm:existence}
Let \(H\) be an \(n\)-qubit Hamiltonian satisfying
\(0\preceq H\preceq hI\). For a finite endpoint
\(\beta_{\max}>0\) with \(-f'(\beta_{\max})>0\), there exists a
perfectly-balanced schedule from \(0\) to \(\beta_{\max}\) with
\[
L_{\mathrm{PB}}
\le
1+\sqrt{
\bigl(f(0)-f(\beta_{\max})\bigr)
\ln\frac{-f'(0)}{-f'(\beta_{\max})}
}.
\]
The factor \(f(0)-f(\beta_{\max})\) is at most \(h\beta_{\max}\).

\medskip

For the low-temperature specialization, assume that \(H\) is nonconstant.
Let \(E_0\) be the ground energy, \(N_0\) its degeneracy, and
\(\bar H:=H-E_0I\). Define
\[
E_s:=\min\{E-E_0:E\in\operatorname{spec}(H),\ E>E_0\}.
\]
If \(\ln N_0+2<n\ln2\), there is a finite \(\gamma>0\) such that
\(\ln\operatorname{Tr}(e^{-\gamma\bar H})=\ln N_0+2\), and the
corresponding perfectly-balanced schedule satisfies
\[
L_{\mathrm{PB}}
\le
1+\sqrt{
\bigl(n\ln2-\ln N_0-2\bigr)
\ln\!\left(
\frac{e^2\|\bar H\|_\infty}{(e^2-1)E_s}
\right)
}.
\]
The terminal interval from \(\gamma\) to \(\infty\) has squared TFD
overlap \(e^{-2}\) and is therefore \(e^2\)-slowly varying.
\end{thm}

\medskip

The finite-endpoint bound follows from
\cite[Lemma~4.3]{Stefankovic2009}. Its application to perfectly-balanced
schedules and the low-temperature specialization are proved in
Appendix~\ref{app:adaptive_schedule}.
Global energy shifts leave TFD overlaps unchanged, so \(\bar H\) is
used only to express the low-temperature bound. If
\(\ln N_0+2\ge n\ln2\), then \(N_0/2^n\ge e^{-2}\), and the
one-interval schedule from zero to any finite endpoint is already
\(e^2\)-slowly varying.

\begin{rem}[Tightness in the exponentially small-gap regime]
\label{thm:adaptive_schedule_tight}
For \(n\ge3\), consider
\[
H_n=\sum_{j=1}^{n}2^{-j}\ket1\!\bra1_j,
\]
where \(\ket1\!\bra1_j\) acts on qubit \(j\) and as the identity on
the other qubits. This family has \(E_s=2^{-n}\) and
\(\|H_n\|_\infty=1-2^{-n}\).
At the endpoint \(\gamma_n\) defined by
\(\ln\operatorname{Tr}(e^{-\gamma_n H_n})=2\), every
\(e^2\)-slowly-varying schedule has \(\Omega(n)\) intervals, whereas
Theorem~\ref{thm:adaptive_schedule_existence} gives
\(L_{\mathrm{PB}}=O(n)\). Hence \(L_{\mathrm{PB}}=\Theta(n)\),
showing that the low-temperature bound can be asymptotically tight.
The proof is given in Appendix~\ref{app:adaptive_schedule},
Theorem~\ref{thm:tight}.
\end{rem}

\medskip

The adaptive schedule need not be perfectly balanced. Instead, its
nonfinal intervals have a partition-function ratio between \(e^2\) and
\(15\). The upper bound ensures slow variation, while the lower bound
controls the number of intervals through the following comparison.

\begin{cor}[Large-ratio counting bound]
\label{cor:large-ratio-counting}
Let \(0=\beta_0<\cdots<\beta_\ell=\beta_{\max}\) be a schedule
such that
\[
\frac{Z(\beta_i)Z(\beta_{i+1})}
     {Z\!\left((\beta_i+\beta_{i+1})/2\right)^2}
\ge e^2
\qquad (0\le i<\ell).
\]
Then \(\ell\le L_{\mathrm{PB}}\), where \(L_{\mathrm{PB}}\) is the
length of a perfectly-balanced schedule with the same endpoints.
In particular, the finite-endpoint bound in
Theorem~\ref{thm:adaptive_schedule_existence} also bounds \(\ell\).
If the ratio condition holds only on the nonfinal intervals, then
\(\ell\le L_{\mathrm{PB}}+1\).
\end{cor}

\medskip

This is the comparison argument of
\cite[Corollary~2.4]{Arunachalam2022}: increasing the interval ratio
cannot increase the number of intervals. Applying the same comparison to all intervals except the last,
then adding the final interval, gives
\(\ell\le L_{\mathrm{PB}}+1\).

\medskip

We now state the quantum schedule-generation procedure in our notation.
The states \(\ket{\psi_\theta}\) below are the TFD states of
Definition~\ref{def:tfd-state}.

\begin{algorithm}[h]
\caption{Quantum cooling schedule generation,
adapted from~\cite[Algorithm~3]{Arunachalam2022}}
\label{alg:quantum-cooling-schedule-generation}
\small
\begin{algorithmic}[1]
\Require Endpoint \(\beta_{\max}>0\); a known bound \(h>0\) with
\(0\preceq H\preceq hI\); failure probability \(0<p<1\); TFD
preparation unitaries and their adjoints at the queried temperatures.
\Ensure A schedule \(0=\beta_0<\cdots<\beta_\ell=\beta_{\max}\), or a failure flag.

\State Set \(\beta_0\gets0\) and \(k\gets0\).

\State Use additive error \(0.005\) for every overlap estimate, and
conditional failure probability at most \(6p/(\pi^2j^2)\) for the
\(j\)-th such call across all binary searches.

\While{\(\beta_k<\beta_{\max}\)}
    \State If \(k\ge\lceil1+\sqrt{h\beta_{\max}}/2\rceil\),
    return a failure flag.
    \State Define
    \(f_o(\theta):=|\langle\psi_{\beta_k}\mid\psi_\theta\rangle|^2\),
    evaluated by quantum overlap estimation.

    \State Compute
    \[
    \beta^\ast\gets
    \textsc{BinarySearch}\!\left(
    f_o(\cdot)\ge0.075,\,
    [\beta_k,\beta_{\max}],\,
    \frac{1}{2h}
    \right).
    \]

    \State If \(\beta^\ast\le\beta_k\), return a failure flag.
    \State Set \(\beta_{k+1}\gets\beta^\ast\) and \(k\gets k+1\).
\EndWhile

\State Set \(\ell\gets k\) and \Return \((\beta_0,\ldots,\beta_\ell)\).
\end{algorithmic}
\end{algorithm}

\medskip

Here \textsc{BinarySearch} is the routine used in
\cite[Algorithm~3]{Arunachalam2022}: it returns the endpoint if its
estimated overlap passes the threshold; otherwise it returns the
accepted lower endpoint of a bracket of the prescribed width.
The precision \(1/(2h)\) replaces \(1/(2n)\) in the original energy
range. Each search therefore uses
\(O(\ln(2+h\beta_{\max}))\) overlap-estimation calls.

\begin{thm}[Quantum cooling schedule generation, adapted from
{\cite[Algorithm~3 and Theorem~3.4]{Arunachalam2022}}]
\label{thm:quantum-schedule-arunachalam}
Let \(H\) be an \(n\)-qubit Hamiltonian satisfying
\(0\preceq H\preceq hI\), with \(h>0\), and let
\(0<p<1\) and \(0<\beta_{\max}<\infty\).
Assume access to the TFD preparation unitaries \(U_\theta\) and their
adjoints from Definition~\ref{def:partition-function-resource-model}
at every temperature \(\theta\in[0,\beta_{\max}]\) queried by
Algorithm~\ref{alg:quantum-cooling-schedule-generation}.

\medskip

With probability at least \(1-p\), the algorithm outputs a
\(15\)-slowly-varying schedule
\(0=\beta_0<\cdots<\beta_\ell=\beta_{\max}\) satisfying
\[
e^2
\le
\frac{Z(\beta_i)Z(\beta_{i+1})}
     {Z\!\left((\beta_i+\beta_{i+1})/2\right)^2}
\le15
\qquad (0\le i<\ell-1),
\]
with only the upper bound required on the final interval.
For \(H\ne0\), its length obeys
\[
\ell
=O\!\left(
1+\sqrt{
\bigl(f(0)-f(\beta_{\max})\bigr)
\ln\frac{-f'(0)}{-f'(\beta_{\max})}
}
\right).
\]
The schedule-generation procedure uses
\[
\widetilde O\!\left(
\ell\,\ln(2+h\beta_{\max})\,\ln\frac{2}{p}
\right)
\]
reflections about TFD states. The same asymptotic bound holds for
\(N_{\mathrm{TFD}}\), and \(N_{\mathrm{BE}}=0\).
For \(H=0\), a one-interval schedule suffices.

\medskip

Under the low-temperature assumptions of
Theorem~\ref{thm:adaptive_schedule_existence}, taking the supplied
endpoint to be \(\gamma\) gives
\[
\ell
=O\!\left(
1+\sqrt{
n\ln\!\left(
\frac{e^2\|\bar H\|_\infty}{(e^2-1)E_s}
\right)
}
\right).
\]
Its generation cost is obtained by setting \(\beta_{\max}=\gamma\)
in the reflection bound above. One terminal interval to \(\infty\)
can be appended without an additional binary search.
\end{thm}

\begin{proof}
We first analyze the algorithm assuming that all overlap estimates are
accurate. The correctness argument of
\cite[Theorem~3.4]{Arunachalam2022} applies to our energy range
\(0\preceq H\preceq hI\). Indeed, for every \(\beta,\Delta\ge0\),
\[
e^{-h\Delta}Z(\beta)
\le Z(\beta+\Delta)
\le Z(\beta).
\]
Thus replacing the original binary-search precision by \(1/(2h)\)
preserves the constant-factor bounds used in that argument. Every
nonfinal interval consequently has a partition-function ratio between
\(e^2\) and \(15\), while the final interval satisfies the upper bound
\(15\). Hence the resulting schedule is \(15\)-slowly varying.

\medskip

To bound its length, apply
Corollary~\ref{cor:large-ratio-counting} to the nonfinal intervals.
Their number is at most \(L_{\mathrm{PB}}\), so including the final
interval gives
\[
\ell\le L_{\mathrm{PB}}+1.
\]
The finite-endpoint and low-temperature bounds now follow from
Theorem~\ref{thm:adaptive_schedule_existence}, with the additional
interval absorbed into the additive constant.
For \(H\ne0\), the Gibbs mean energy is strictly positive at every
finite inverse temperature. Therefore
\(-f'(\beta_{\max})>0\), and the logarithm in the finite-endpoint
bound is well defined.

\medskip

We next account for estimation failures and query costs.
Each binary search has range at most \(\beta_{\max}\) and precision
\(1/(2h)\), and therefore requires
\(O(\ln(2+h\beta_{\max}))\) overlap estimates.
Over the entire schedule, the number of overlap-estimation calls is
thus \(O(\ell\ln(2+h\beta_{\max}))\).

\medskip

Because these calls are chosen adaptively, we assign their failure
probabilities by call index. For the \(j\)-th overlap-estimation call
in the full schedule-generation procedure, set
\[
\eta_j:=\frac{6p}{\pi^2j^2}.
\]
Conditioned on the preceding history, that call returns an inaccurate
estimate with probability at most \(\eta_j\).
Since \(\sum_{j=1}^{\infty}j^{-2}=\pi^2/6\), we have
\(\sum_{j\ge1}\eta_j=p\).
A union bound therefore shows that all overlap estimates
are simultaneously accurate with probability at least \(1-p\).
The correctness and length guarantees established above therefore hold
with this probability.

\medskip

Each overlap estimate has constant additive accuracy and uses
\(O(\ln(1/\eta_j))\) reflections. Summing over the calls gives
\[
\widetilde O\!\left(
\ell\,\ln(2+h\beta_{\max})\,\ln\frac{2}{p}
\right)
\]
reflections in total.

\medskip

Finally, the TFD preparation unitary implements the required reflection
through
\[
2\ket{\psi_\theta}\!\bra{\psi_\theta}-I
=
U_\theta
\bigl(2\ket0\!\bra0-I\bigr)
U_\theta^\dagger.
\]
Its controlled version is obtained by controlling the middle
zero-state reflection. Thus either version uses one application each
of \(U_\theta\) and \(U_\theta^\dagger\).
The TFD preparations needed to initialize the overlap estimates also
fit within the same asymptotic bound. Consequently, the stated
reflection bound also bounds \(N_{\mathrm{TFD}}\), while
\(N_{\mathrm{BE}}=0\).
\end{proof}

\medskip

For a finite endpoint, any schedule whose nonfinal interval ratios are
at least \(e^2\) also satisfies \(\ell\le1+\sqrt{h\beta_{\max}}/2\).
To see this, for \(0\le i<\ell-1\), set
\(\Delta_i:=\beta_{i+1}-\beta_i\) and
\(d_i:=f'(\beta_{i+1})-f'(\beta_i)\ge0\).
Convexity gives
\[
2
\le
f(\beta_i)+f(\beta_{i+1})
-2f\!\left(\frac{\beta_i+\beta_{i+1}}2\right)
\le
\frac{\Delta_i d_i}{2}.
\]
Since \(\sum_{i=0}^{\ell-2}\Delta_i\le\beta_{\max}\) and
\(\sum_{i=0}^{\ell-2}d_i\le h\), Cauchy--Schwarz yields
\[
\ell-1
\le
\frac12\sum_{i=0}^{\ell-2}\sqrt{\Delta_i d_i}
\le
\frac12\sqrt{h\beta_{\max}}.
\]
This applies to perfectly-balanced schedules and, on the successful
event above, to the generated schedule. In particular,
\(h=O(n)\) and \(\beta_{\max}=O(1)\) give \(\ell=O(\sqrt n)\).
When all overlap estimates are accurate, each step
advances the temperature and the algorithm finishes within the cap.
Thus the stopping checks preserve the success guarantee.
On every run, generation uses at most
\(\widetilde O((1+\sqrt{h\beta_{\max}})
\ln(2+h\beta_{\max})\ln(2/p))\) TFD preparation queries
and no block-encoding calls.

\medskip

The construction follows Arunachalam et al.~\cite{Arunachalam2022};
our real-valued setting requires retaining the dependence on \(E_s\),
which is bounded below by one in their integer-valued setting after
shifting the ground energy to zero.

\medskip

Assume that the TFD oracle used for reconstruction
also supports the temperatures needed to generate the schedule.
For a finite endpoint \(\beta\), assign failure probability \(p/2\)
to each stage, giving an overall failure probability at most \(p\).
Since \(1+\sqrt{h\beta}=O(\ell+h\beta)\) for \(\ell\ge1\),
generation is absorbed by the query bounds of
Theorems~\ref{thm:single-temperature-main}
and~\ref{thm:coherent-indexed-main} when \(0<\epsilon\le1\).
The reconstruction query-cost bounds hold for every returned
increasing schedule, even if it is not slowly varying.
Replacing \(\ell\) by the deterministic cap therefore also bounds
the expected cost of the combined procedure.
\section{Circuit-Level Resource Estimates}
\label{sec:end-to-end}

A smaller query bound for coherent-temperature reconstruction does not
by itself imply a smaller gate count. Each coherent TFD preparation
may require more gates than a fixed-temperature preparation, and this
additional cost can offset the reduction in oracle calls.
We first identify when the query improvement is retained in the
gate-count bounds, then give a sufficient condition for implementing
coherent TFD preparation with only polylogarithmic overhead relative
to fixed-temperature preparation.
Throughout this section, we count one- and two-qubit gates, allowing
arbitrary single-qubit rotations, and define depth as the number of
parallel gate layers. We allow arbitrary pairwise connectivity and
the clean workspace specified by each construction.

\subsection{From query bounds to gate counts}

Let \(G_{\mathrm{fixed}}\) and \(G_{\mathrm{coh}}\) be gate-count
upper bounds for one fixed-temperature and one coherent-temperature
TFD preparation, respectively. We compare chosen implementations
at a common preparation accuracy sufficient for both algorithms,
using the prescribed canonical TFD convention.
The fixed-temperature bound holds uniformly over all queried
temperatures. Each bound also applies to the required inverse and
controlled versions.

\medskip

For the same supplied slowly-varying schedule,
Theorems~\ref{thm:single-temperature-main}
and~\ref{thm:coherent-indexed-main} give the following expected
gate-count upper bounds for the TFD part of reconstruction:
\[
\begin{aligned}
\text{Fixed temperature:}\quad&
\widetilde O\!\left(
\frac{\sqrt{\ell}(\ell+h\beta)\ln(2/p)}{\epsilon}
G_{\mathrm{fixed}}\right),\\[3pt]
\text{Coherent temperature:}\quad&
\widetilde O\!\left(
\frac{(\ell+h\beta)\ln(2/p)}{\epsilon}
G_{\mathrm{coh}}\right).
\end{aligned}
\]
Thus a sufficient condition for retaining the \(\sqrt\ell\)
improvement in these bounds, up to polylogarithmic factors, is
\[
G_{\mathrm{coh}}=\widetilde O(G_{\mathrm{fixed}}).
\]
An additional factor of order \(\sqrt\ell\) per coherent preparation
can cancel the improvement in the TFD contribution. A strict
asymptotic comparison must also account for the suppressed
polylogarithmic factors.

\medskip

The full gate count also includes Hamiltonian block encoding,
branch preparation, parameter and QSVT phase loading, arithmetic,
and reflections.
Assume block-encoding calls have comparable gate costs in both
methods. Also assume that, for each method, the total auxiliary
cost is at most of the same order as the combined cost of TFD
preparation and block encoding.
Under these assumptions and the condition above, the
\(\sqrt{\ell}\) improvement also holds for the full gate-count
upper bounds, up to polylogarithmic factors.
Depth must be analyzed separately, since it also depends on
which gates can run in parallel and how many auxiliary qubits
are available.

\subsection{Efficient coherent implementation}

Suppose that preparations at \(m\) temperatures are supplied as known
circuits \(U_0,\ldots,U_{m-1}\), with gate counts \(G_0,\ldots,G_{m-1}\).
Their coherent selection acts on a binary index register and a target
register as
\[
|k\rangle|\varphi\rangle
\longmapsto
|k\rangle U_k|\varphi\rangle,
\qquad 0\le k<m.
\]
A direct implementation applies each \(U_k\) controlled on its
corresponding index value. Controlling each elementary gate incurs
only a constant-factor gate overhead.
A unitary compute--uncompute implementation of unary iteration
generates the selection flags using \(O(m)\) additional gates and
\(O(\log m)\) additional clean selection qubits
\cite[Section~III~A]{Babbush2018LinearT}.
The resulting construction therefore uses
\[
O\!\left(\sum_{k=0}^{m-1}G_k+m\right)
\]
gates. Any conversion between temperature encodings and these
indices must also be counted.

\medskip

If \(G_k\le G_{\mathrm{fixed}}\), this construction gives an
\(O(mG_{\mathrm{fixed}})\) gate-count bound.
For circuits of comparable size, the linear selection overhead
can cancel the \(\sqrt{\ell}\) improvement in the TFD gate-count
comparison above.
This construction assumes known circuit descriptions; uncontrolled
access to unknown black-box unitaries does not generally suffice
to obtain their controlled versions
\cite{Araujo2014QuantumControl}.

\medskip

A shared circuit structure can reduce this cost by computing the
rotation angles from the temperature register.
The following proposition gives a sufficient condition.

\begin{prop}[Efficient coherent parameterization]
Suppose the fixed-temperature preparations share a common sequence
of at most \(G_{\mathrm{fixed}}\) gates with fixed gate locations.
Each gate is fixed or a rotation \(e^{-i\phi P/2}\), where \(P\)
is a fixed one- or two-qubit Pauli operator and only the angle
\(\phi\) depends on temperature.
Let \(b\) be the number of bits used to encode the temperature and
each angle, with angles represented modulo \(4\pi\).
Assume each encoded angle can be computed reversibly from the
temperature using \(\operatorname{poly}(b)\) gates and clean work
qubits, uniformly over all gate locations.
Then the encoded family admits a coherent implementation with
\[
G_{\mathrm{coh}}
=O\!\left(G_{\mathrm{fixed}}\operatorname{poly}(b)\right).
\]
The additional parameter workspace uses \(\operatorname{poly}(b)\)
qubits and returns to zero. The inverse and controlled versions
satisfy the same bound.
\end{prop}

\begin{proof}
For each temperature-dependent gate, compute its angle into a clean
register. Implement the rotation as a product of \(O(b)\) fixed
rotations controlled by the angle bits, then reverse the angle
computation. Reuse the workspace for the next gate.
This costs \(\operatorname{poly}(b)\) gates per original gate and
preserves the temperature register, so the construction also acts
coherently on superpositions of temperatures.
Reversing the circuit gives its inverse; controlling its elementary
gates changes the gate count by only a constant factor.
\end{proof}

Angle rounding contributes
\(O(G_{\mathrm{fixed}}2^{-b})\) operator-norm error per branch.
Because the temperature register is preserved, the error of the
coherent implementation is the largest branch error.
If the required temperature and preparation accuracies can be
achieved with \(b\) polylogarithmic in the problem size and inverse
accuracy, the proposition gives
\(G_{\mathrm{coh}}=\widetilde O(G_{\mathrm{fixed}})\),
as required for the preceding comparison.

\medskip

The next section uses a shared temperature-dependent preparation
procedure for spin chains, with polylogarithmic-length parameter
encodings, to construct coherent TFD access at comparable gate cost.
It also includes the block-encoding and auxiliary costs to obtain
gate-count and depth bounds. The construction and precision
requirements are given in
Appendix~\ref{app:coherent-temperature-circuits}.

\section{Application: one-dimensional partition-function estimation}

Consider an open chain of \(n\) qubits with explicitly specified
interactions of bounded range, bounded term norm, and bounded number
of terms per site. The interactions need not commute or be translation
invariant. All queried inverse temperatures lie in
\([0,\beta_{\max}]\), where \(\beta_{\max}>0\) is fixed
independently of \(n\). Local coefficients are available to the
required precision. We use the circuit architecture of
Section~\ref{sec:end-to-end}.

\medskip

One-dimensional thermal states admit matrix-product operator
approximations that accurately reproduce local observables
\cite[Theorem~2]{Huang2021}. Here we instead use preparation of
the full canonical TFD. At each fixed finite inverse temperature,
this state can be prepared to trace-distance error \(\eta\)
with depth \(\operatorname{polylog}(2n/\eta)\) and size
\(n\operatorname{polylog}(2n/\eta)\)
\cite[Corollary~I.1 and Appendix~C]{bergamaschi2026fastmixingquantumspin}.

\medskip

This preparation has a common structure that supports a coherent
temperature input. Let \(\mathcal H_t\) denote the parent Hamiltonian
whose ground state is the canonical TFD at inverse temperature \(t\).
For endpoint temperature \(\theta\), the quasi-adiabatic construction
uses the path \(s\mapsto\mathcal H_{s\theta}\), \(0\le s\le1\).
Uniform gap and locality bounds on \([0,\beta_{\max}]\) allow the
block sizes, simulation cutoffs, and error tolerances to be chosen
independently of \(\theta\)
\cite[Section~VII and Appendix~C]{bergamaschi2026fastmixingquantumspin}.
The local simulations use temperature-controlled block encodings
of the parent Hamiltonian and its derivative
\cite[v2, Lemmas~C.5 and~C.6]{chen2023efficient}.
For the chosen path, these are evaluated at \(s\theta\), with
\(\partial_s\mathcal H_{s\theta}
=\theta\mathcal H'_{s\theta}\).
Reversibly computing these parameters lets the same subroutines
act coherently on the endpoint temperature, with polynomial
arithmetic overhead in the parameter bit length and no multiplicative
factor equal to the number of supported temperatures.
This gives a concrete realization of inexpensive coherent access
as discussed in Section~\ref{sec:end-to-end}.

\begin{prop}[Coherent TFD preparation for spin chains]
\label{prop:coherent-tfd-spin-chains}
Fix \(0<\eta<1\). For the chains above and inverse temperatures
\(\theta\in[0,\beta_{\max}]\) represented with
\(O(\log(2n/\eta))\) bits, there is a temperature-preserving
unitary satisfying
\[
\left\|
\widetilde U_{\mathrm{cohTFD}}\ket{\theta}\ket0
-\ket{\theta}\ket{\psi_\theta}\ket0_{\mathrm{aux}}
\right\|_2\le\eta,
\]
where \(\ket{\psi_\theta}\) is the canonical TFD of
Definition~\ref{def:tfd-state}, and \(\mathrm{aux}\) contains
all additional preparation workspace. The bound also holds for
normalized superpositions of temperature labels.
The circuit, its inverse, and its controlled version use
\(n\operatorname{polylog}(2n/\eta)\) gates and auxiliary qubits,
with depth \(\operatorname{polylog}(2n/\eta)\).
The constants depend on \(\beta_{\max}\) and the interaction bounds.
\end{prop}

The proof, including temperature rounding and workspace errors,
is in Appendix~\ref{app:coherent-temperature-circuits},
Subsection~\ref{subsec:coherent-temperature-proof}.
No gap assumption on the physical Hamiltonian is required.
Fixing the temperature input gives the fixed-temperature preparation
with the same resource bounds.

\medskip

A known scalar shift gives \(0\preceq H\preceq hI\) with a
chosen bound \(h=\Theta(n)\), without changing the TFD. A block
encoding with \(\alpha=\Theta(n)\) and the auxiliary reconstruction
operations have \(\widetilde O(n)\) gate cost and polylogarithmic
depth per primitive; see
Appendix~\ref{app:coherent-temperature-circuits},
Subsection~\ref{subsec:coherent-temperature-auxiliary}.

\begin{cor}[Circuit bounds for interacting spin chains]
\label{cor:spin-chain-circuit-comparison}
\label{4.1}
For the chains and architecture above, let \(0<\epsilon\le1\),
\(0<p<1\), and \(0<\beta\le\beta_{\max}\). Suppose a
\(C\)-slowly-varying schedule is supplied with
\(\ell=\widetilde O(\sqrt n)\), \(C=O(1)\), and
polylogarithmic-length parameter encodings at the required precision.
Both reconstruction methods estimate \(f(\beta)=\ln Z(\beta)\)
to additive error \(\epsilon\), with failure probability at most
\(p\). They admit implementations with the following expected
resource upper bounds:
\[
\begin{array}{c|c|c}
\text{Method}&\text{Depth}&\text{Gate count}\\[2pt]\hline
\text{Fixed temperature}
&\displaystyle\widetilde O\!\left(
\frac{n^{5/4}\ln(2/p)}{\epsilon}\right)
&\displaystyle\widetilde O\!\left(
\frac{n^{9/4}\ln(2/p)}{\epsilon}\right)\\[7pt]
\text{Coherent temperature}
&\displaystyle\widetilde O\!\left(
\frac{n\ln(2/p)}{\epsilon}\right)
&\displaystyle\widetilde O\!\left(
\frac{n^2\ln(2/p)}{\epsilon}\right).
\end{array}
\]
These bounds include auxiliary, controlled, and inverse operations.
They measure quantum circuit resources, with classical circuit
compilation treated as preprocessing.
The notation \(\widetilde O\) suppresses polylogarithmic
implementation factors in \(n\), \(1/\epsilon\), and \(1/p\).
\end{cor}

\begin{proof}
Use Proposition~\ref{prop:coherent-tfd-spin-chains} with a classical
or coherent temperature input. Substitute its costs and the auxiliary
costs in Appendix~\ref{app:coherent-temperature-circuits} into
Theorems~\ref{thm:single-temperature-main}
and~\ref{thm:coherent-indexed-main}. Choose the preparation accuracy
as in Appendix~\ref{sec:approx-prep}, together with the numerical
error budgets in Appendix~\ref{app:coherent-temperature-circuits}.
All additional precision costs are polylogarithmic.
\end{proof}

\medskip

For this class, implementing coherent access preserves the smaller
reconstruction upper bound. The displayed bounds differ by
\(n^{1/4}\), up to polylogarithmic factors. This compares two
implementations on the same model class and supplied schedule, not
the optimal costs of the access models or a separation from classical
methods. In particular, the depth bounds do not optimize every
possible parallel execution of the independent fixed-temperature
estimates.
\section{Discussion and Conclusion}
\label{sec:conclusion}

We have extended adaptive quantum algorithms for partition-function estimation with Heisenberg scaling
from classical Gibbs distributions to general quantum Hamiltonians.
Existing quantum speedups for this problem, from quantum mean
estimation~\cite{Montanaro2015} and adaptive
annealing~\cite{Harrow2020,Arunachalam2022} to nearly unbiased product
estimation~\cite{Cornelissen2023}, rely on coherent evaluation of
classical energies and, in several cases, on integer-valued spectra.
Neither assumption holds for a Hamiltonian accessed through a block
encoding, and resolving individual energy levels would require
additional spectral precision. Removing these assumptions required
new machinery: a recursive-doubling identity that expresses each
schedule increment through thermofield-double overlaps and short
imaginary-time reweighting steps, whose probabilities stay bounded
away from zero along a slowly-varying schedule; a variance-sensitive
allocation of estimation depths for fixed-temperature access; a
coherent inverse-binomial estimator for logarithms of success
probabilities; and an analysis of adaptive cooling schedules for
real-valued spectra. All resulting estimators retain Heisenberg
scaling in the precision.

For extensive Hamiltonians at constant inverse temperature, coherent-temperature
reconstruction attains the optimal \(\widetilde O(n/\epsilon)\)
block-encoding query complexity in the hidden-offset model
(Corollary~\ref{cor:coherent-block-query-optimality}).
The fixed-temperature bound remains a factor of \(n^{1/4}\) above
this benchmark, up to polylogarithmic factors.

Several questions remain open. The most immediate is whether the gap
between the two access models is fundamental. Coherent-temperature
access can be costly to implement in general. We constructed it
efficiently for one-dimensional spin chains, but it would be desirable
either to remove this requirement altogether or to improve the
fixed-temperature algorithm so that it approaches the
\(\widetilde O(n/\epsilon)\) scaling. Conversely, a lower bound
separating the two models would show that coherent temperature control
is a genuine resource for this problem.

A second direction concerns thermodynamic integration, which expresses
the log partition function as an integral of the mean energy,
\[
f(\beta)-f(0)=-\int_0^\beta\Tr(H\rho_t)\,\dd t .
\]
With coherent temperature access, the integration variable can be
sampled in superposition and the integrand estimated through the block
encoding. Quantum mean estimation should then give scaling similar to
ours in some regimes. The two approaches rely on different resources.
Thermodynamic integration needs no cooling schedule, but it requires
coherent access across a fine grid of temperatures spanning the whole
interval. Our reconstruction needs only the temperatures of a short
slowly-varying schedule. It would be interesting to compare the two
more closely, both in their constants and variances and in their
practical performance on concrete Hamiltonians.

\section*{Acknowledgements}
We thank Aram Harrow, Cambyse Rouz\'e, Kristan Temme, Alvaro Alhambra, Matteo Scandi and Pawel Wocjan for interesting discussions.
D.S.F. acknowledges financial support from the Novo Nordisk Foundation
(Grant No. NNF20OC0059939 Quantum for Life) and from the ERC grant GIFNEQ
101163938. Yufei Wang acknowledges support from the Danish e-Infrastructure Consortium (DeiC) through a PhD scholarship in quantum algorithms and quantum software, grant no. 5259-00016B.
This project was funded within the QuantERA II Programme that has received funding from the EU’s H2020 research and innovation programme under the GA No 101017733.

\paragraph{Use of artificial intelligence.}
OpenAI Codex with ChatGPT 5.6 and 6 Astra assisted with exploring and checking
proof strategies, and with preparing the manuscript. In particular, it played a central role in the construction of the coherent temperature algorithm. The current
organisation and exposition of the proofs are due to the authors. The
authors are responsible for the mathematical claims, the final
presentation, and any remaining errors.

\bibliographystyle{plainnat}
\bibliography{ref}

\appendix
\section{A linear lower bound from hidden energy normalization}
\label{app:linear-lower-bound}

We prove a lower bound on block-encoding queries when the absolute energy
offset is available only through the block encoding. A scalar energy shift
changes the partition function but leaves normalized thermal states
unchanged. Resolving this shift gives a lower bound proportional to the
block-encoding normalization, and hence an
\(\Omega(n\beta/\epsilon)\) bound for extensive normalization.
We first bound how well an algorithm can distinguish two nearby oracles.

\begin{lem}[Hybrid bound for query algorithms]
\label{lem:hybrid-lower-bound}
Let \(\mathcal A\) be a quantum algorithm that makes \(q\) queries, each of
which may be a unitary oracle, its adjoint, or a controlled version, interleaved
with oracle-independent operations.  If the two possible unitary oracles are
\(U\) and \(V\), then the final pure states satisfy
\[
\|\ket{\psi_U}-\ket{\psi_V}\|
\le q\|U-V\|.
\]
\end{lem}

\begin{proof}
By deferred measurement, the algorithm may be purified without changing its
query count.  Insert the standard sequence of hybrids in which the first
\(r\) oracle calls use \(U\) and the remaining calls use \(V\).  For adjoint
queries,
\(\|U^\dagger-V^\dagger\|=\|U-V\|\), and adding a control does not change
this norm.  Neighboring hybrids therefore differ by at most \(\|U-V\|\).
Summing the \(q\) neighboring distances proves the claim.
\end{proof}

\begin{lem}[Distinguishing consequence]
\label{lem:distinguish-lower-bound}
If a valid estimator forces some measurement event to have probabilities
differing by at least \(1/3\) on two inputs, then the corresponding final states
have Euclidean distance at least \(1/3\).
\end{lem}

\begin{proof}
For any measurement event, the probability difference is bounded by the trace
distance of the two final states.  For pure states, the trace distance is at
most the Euclidean distance.
\end{proof}

\begin{thm}[Nontrivial spectrum plus hidden offset]
\label{thm:linear-lower-bound}
Let \(D\ge2\) and \(h,\alpha,\beta>0\), and fix a known, non-scalar
Hamiltonian \(A\) satisfying
\[
0\preceq A\preceq \frac{\min\{h,\alpha\}}{4}I_D .
\]
For \(0\le c\le\min\{h,\alpha\}/4\), define
\[
H_c=A+cI_D .
\]
Every Hamiltonian in this family satisfies
\(0\preceq H_c\preceq hI_D\).
Suppose that the only input interface depending on \(c\) is the exact
\(\alpha\)-block encoding
\[
U_c
=
\begin{pmatrix}
M_c & \sqrt{I-M_c^2}\\
\sqrt{I-M_c^2} & -M_c
\end{pmatrix},
\qquad
M_c:=\frac{H_c}{\alpha}.
\]
The algorithm may know \(A\) and may receive, at no query cost,
fixed-temperature or parameter-controlled TFD preparation unitaries, their
adjoints, and their controlled versions at arbitrary inverse temperatures.
It may also receive any other auxiliary operations. The purification
convention, the unitary extensions of all TFD preparations, and every other
freely supplied operation are required to be independent of \(c\).

\medskip

Any algorithm that outputs an additive-\(\epsilon\) estimate of
\[
f_c=\ln Z_{H_c}(\beta)
\]
with probability at least \(2/3\), uniformly over this family, must make
\[
q=\Omega\!\left(\frac{\alpha\beta}{\epsilon}\right)
\]
queries to \(U_c\) or \(U_c^\dagger\), with controlled queries counted in
the same way, provided
\[
0<\epsilon\le \frac{\beta\min\{h,\alpha\}}{16}.
\]
The same asymptotic lower bound holds for the worst-case expected
number of block-encoding queries when the query count is random.
\end{thm}

\begin{proof}
The spectrum of \(H_c\) is the spectrum of \(A\) shifted by \(c\), and
\[
Z_{H_c}(\beta)
=
e^{-\beta c}Z_A(\beta),
\qquad
f_c
=
\ln Z_A(\beta)-\beta c .
\]
However, for every inverse temperature \(\theta\ge0\), the normalized Gibbs
state satisfies
\[
\rho_{H_c}(\theta)
=
\frac{e^{-\theta(A+cI_D)}}{Z_{H_c}(\theta)}
=
\frac{e^{-\theta A}}{Z_A(\theta)}
=
\rho_A(\theta).
\]
Hence fixed-temperature and parameter-controlled TFD preparations can be
chosen with a common purification convention and unitary extension that are
independent of \(c\). Even free coherent TFD access at all temperatures
therefore does not reveal the unknown offset.

\medskip

Choose
\[
c_0=0,
\qquad
c_1=\frac{4\epsilon}{\beta}.
\]
The assumption \(\epsilon\le\beta\min\{h,\alpha\}/16\) ensures that
\(c_1\le\min\{h,\alpha\}/4\), so both inputs belong to the family and
satisfy \(0\preceq H_{c_i}\preceq hI_D\).
The two target values differ by
\[
|f_{c_1}-f_{c_0}|
=
\beta(c_1-c_0)
=
4\epsilon .
\]
Let \(m=(f_{c_0}+f_{c_1})/2\). On input \(c_0\), every output within
\(\epsilon\) of \(f_{c_0}\) lies above \(m\), whereas on input \(c_1\), every
output within \(\epsilon\) of \(f_{c_1}\) lies below \(m\). Therefore the event
that the output exceeds \(m\) has probabilities differing by at least \(1/3\)
on the two inputs. By Lemma~\ref{lem:distinguish-lower-bound}, the
corresponding final states have Euclidean distance at least \(1/3\).

\medskip

It remains to bound the distance between the two block-encoding oracles.
Since
\[
0\preceq H_c
\preceq \frac{\min\{h,\alpha\}}{2}I_D
\preceq \frac{\alpha}{2}I_D,
\]
the spectrum of \(M_c\) is contained in \([0,1/2]\). Define
\[
u(x)=
\begin{pmatrix}
x & \sqrt{1-x^2}\\
\sqrt{1-x^2} & -x
\end{pmatrix},
\qquad
0\le x\le\frac12.
\]
Its derivative is
\[
u'(x)
=
\begin{pmatrix}
1 & -\dfrac{x}{\sqrt{1-x^2}}\\[1ex]
-\dfrac{x}{\sqrt{1-x^2}} & -1
\end{pmatrix},
\]
and hence
\[
\|u'(x)\|
=
\frac{1}{\sqrt{1-x^2}}
\le \frac{2}{\sqrt3}
<2.
\]
Because all operators \(M_c\) are functions of the fixed Hamiltonian \(A\),
they are simultaneously diagonalizable. Applying the scalar mean-value bound
in their common eigenbasis gives
\[
\|U_{c_1}-U_{c_0}\|
\le
\frac{2|c_1-c_0|}{\alpha}.
\]
The same bound holds for adjoint and controlled queries.
By Lemmas~\ref{lem:hybrid-lower-bound} and
\ref{lem:distinguish-lower-bound},
\[
\frac13
\le
q\|U_{c_1}-U_{c_0}\|
\le
q\,\frac{2(c_1-c_0)}{\alpha}
=
q\,\frac{8\epsilon}{\alpha\beta}.
\]
Thus
\[
q\ge \frac{\alpha\beta}{24\epsilon}.
\]

\medskip

To extend the bound to a random query count, let \(T\) be the number
of block-encoding queries made by the algorithm, and set
\(Q:=\max\{\mathbb E_{c_0}[T],\mathbb E_{c_1}[T]\}\).
If \(Q=\infty\), the claim is immediate. Otherwise, truncate the
algorithm before its \((q+1)\)-st query, where
\(q:=\lfloor12Q\rfloor\), and return an arbitrary estimate if it is
truncated. For either input, Markov's inequality gives
\[
\Pr[T>q]
\le
\frac{\mathbb E[T]}{q+1}
\le
\frac{Q}{\lfloor12Q\rfloor+1}
<\frac1{12}.
\]
The truncated algorithm therefore succeeds with probability at least
\(7/12\) on each input. The event that its output exceeds \(m\)
then has probabilities differing by at least \(1/6\).
As in Lemma~\ref{lem:distinguish-lower-bound}, this probability
difference lower-bounds the Euclidean distance between purified
final states. The hybrid bound consequently gives
\[
\frac16
\le
q\,\frac{8\epsilon}{\alpha\beta}
\le
\frac{96Q\epsilon}{\alpha\beta},
\qquad
Q\ge\frac{\alpha\beta}{576\epsilon}.
\]
Since \(Q\) is bounded above by the worst-case expected query count
over the family, this proves the expected-query lower bound.
\end{proof}

\begin{cor}[Linear lower bound for extensive block encodings]
\label{cor:linear-lower-bound}
Let \(h=\Theta(n)\), \(\alpha=\Theta(n)\), and \(\beta>0\).
Consider the hidden-offset family and input model of
Theorem~\ref{thm:linear-lower-bound} with \(D=2^n\).
Every Hamiltonian in this family satisfies \(0\preceq H_c\preceq hI\).
Any algorithm that estimates \(\ln Z_{H_c}(\beta)\) to additive error
\(\epsilon\), with success probability at least \(2/3\) for every
member of the family, requires
\[
\Omega\!\left(\frac{n\beta}{\epsilon}\right)
\]
block-encoding queries in worst-case expectation, provided
\[
0<\epsilon\le \frac{\beta\min\{h,\alpha\}}{16}.
\]
The bound holds even with free fixed-temperature and
parameter-controlled TFD preparations, subject to the
offset-independence conditions in that theorem.
In particular, for \(\beta=\Theta(1)\) and the same accuracy condition,
the lower bound is
\[
\Omega\!\left(\frac{n}{\epsilon}\right).
\]
\end{cor}

\begin{proof}
Apply Theorem~\ref{thm:linear-lower-bound} with \(D=2^n\).
Under the stated accuracy condition, it gives
\[
q=\Omega\!\left(\frac{\alpha\beta}{\epsilon}\right)
=\Omega\!\left(\frac{n\beta}{\epsilon}\right),
\]
where the second equality uses \(\alpha=\Theta(n)\).
The constant-temperature statement follows by taking
\(\beta=\Theta(1)\).
\end{proof}
\begin{rem}[Interpretation of the offset lower bound]
The lower bound uses an unknown absolute offset.  It therefore applies in the
generic block-encoding model whenever the absolute ground energy is not
separately supplied, even if a non-tight lower bound on that energy is known.
It does not apply to a stronger input model that promises and enforces an exact
normalization such as \(\lambda_{\min}(H)=0\).

\medskip

The bound concerns \(N_{\mathrm{BE}}\); it places no lower bound on
\(N_{\mathrm{TFD}}\), \(N_{\mathrm{cohTFD}}\), total gate count, or
circuit depth.

\medskip

The common \(h\beta\) scale has different roles in the two arguments.  In the
upper bounds it controls the repeated halving of schedule intervals needed to keep
\(\delta_i\|H\|=O(1)\).  In the lower bound it arises because an offset
changes \(f(\beta)\) by \(\beta c\) while perturbing an
\(\alpha\)-block encoding by only \(O(c/\alpha)\).
\end{rem}

\section{Proofs for fixed-temperature reconstruction}
\label{app:single-temperature}

This appendix proves the local logarithmic bounds and the variable-depth
reconstruction used in Section~\ref{sec:single-temperature}.  Every local
branch is chosen classically.  No coherent superposition over temperatures or
doubling levels is used.
\subsection{QSVT implementation of the imaginary-time step}
\label{app:qsvt-imaginary-time}

The reweighting terms use a block encoding of the scaled imaginary-time
operator \(e^{-1}e^{-\delta H/2}\). The fixed scale ensures that the
approximating polynomial is bounded on the full interval \([-1,1]\), as
required by QSVT. Its effect on the reweighting logarithm is the known
additive correction \(2\).

\medskip

Let \(S\) denote the system register on which \(H\) acts, and suppose
that \(U_H\) is an exact \(\alpha\)-block encoding of \(H\).
Let \(A\) contain both the original block-encoding ancillas and the two
additional ancilla qubits used by the polynomial transformation, and let
\(a\) be the total number of qubits in \(A\). Padding \(U_H\) with the
identity on the additional qubits, we have
\[
(\bra{0_a}_A\otimes I_S)U_H(\ket{0_a}_A\otimes I_S)=H/\alpha,
\]
where \(\ket{0_a}_A\) is the all-zero state of \(A\).

\begin{lem}[Block encoding of the imaginary-time step]
\label{lem:qsvt-imaginary-time}
Let \(\delta\ge0\) and \(0<\nu<1\). There exists a unitary
\(D_\delta\) acting on \(A\otimes S\) such that
\begin{equation}
\label{eq:qsvt-delta-block}
(\bra{0_a}_A\otimes I_S)D_\delta(\ket{0_a}_A\otimes I_S)
=\widetilde A_\delta,
\qquad
\|\widetilde A_\delta-e^{-1}e^{-\delta H/2}\|\le\nu.
\end{equation}
The unitary \(D_\delta\) uses
\[
q=O\!\left((1+\alpha\delta)\log(2/\nu)\right)
\]
queries to \(U_H\), \(U_H^\dagger\), and their controlled versions.
In particular, if \(\alpha=\Theta(h)\) and \(h\delta\le1\), then
\(q=O(\log(2/\nu))\). Hence \(q=\widetilde O(1)\) when \(\nu\) is
inverse polynomial in the reconstruction parameters.
\end{lem}

\begin{proof}
Set \(t:=\alpha\delta/2\). If \(t=0\), a single-qubit rotation
implements the block \(e^{-1}I\) without oracle queries.
For \(t>0\), set
\[
s:=\min\!\left\{1,\frac{1}{4t}\right\},
\qquad
\varepsilon_*:=\min\!\left\{\frac{\nu}{2},\frac{1}{80}\right\}.
\]
The function \(g(x):=e^{-1-tx}\) has Taylor coefficients
\(c_k=e^{-1-t}(-t)^k/k!\) about \(x=1\), satisfying
\[
\sum_{k=0}^{\infty}(1+s)^k|c_k|
=e^{-1+ts}\le e^{-3/4}<\frac{19}{40}.
\]
Apply \cite[Corollary~66]{Gilyen2019QSVT} with
\(x_0=r=1\), extension parameter \(s\), coefficient bound \(19/40\),
and accuracy \(\varepsilon_*\).
Taking the real part of the resulting polynomial gives
\(P\in\mathbb R[x]\) such that
\[
\sup_{x\in[0,1]}|P(x)-e^{-1-tx}|\le\varepsilon_*,
\qquad
\sup_{x\in[-1,1]}|P(x)|
\le\frac{19}{40}+\varepsilon_*<\frac12,
\]
with degree
\[
d=O\!\left(s^{-1}\log(1/\varepsilon_*)\right)
=O\!\left((1+\alpha\delta)\log(2/\nu)\right).
\]
By \cite[Theorem~56]{Gilyen2019QSVT}, the exact block encoding
of \(H/\alpha\) can therefore be transformed into a block encoding
of \(P(H/\alpha)\) with error at most \(\nu/2\), using \(O(d)\)
queries to \(U_H\), \(U_H^\dagger\), and their controlled versions.
The two additional ancilla qubits required by this construction
are included in \(A\).
Since the spectrum of \(H/\alpha\) lies in \([0,1]\), the resulting
all-zero block satisfies
\[
\left\|\widetilde A_\delta-e^{-1}e^{-\delta H/2}\right\|
\le\frac{\nu}{2}+\varepsilon_*\le\nu.
\]
This proves the approximation and query bounds.
\end{proof}

\subsection{Nondestructive amplitude estimation and local logarithmic bounds}
\label{app:local-log-estimation}

We first record the amplitude-estimation result used by
Algorithm~\ref{alg:local-log-estimate}.

\begin{thm}[Nondestructive unbiased amplitude estimation;
Cornelissen--Hamoudi]
\label{thm:ndub-amplitude}
\cite[Theorem~2.4]{Cornelissen2023}
Let \(\ket{\psi}\) be a quantum state and let \(\Pi\) be a projector with
\(p=\|\Pi\ket{\psi}\|^2\).  Given \(t>4\) and
\(\epsilon\in(0,1)\), the algorithm
\(\operatorname{NdUb\text{-}Amplitude}(\ket{\psi},\Pi,t,\epsilon)\)
of Cornelissen and Hamoudi~\cite[Algorithm~3]{Cornelissen2023} outputs
\(\widetilde p\in[-2\pi,2\pi]\) such that
\[
\left|\mathbb E[\widetilde p]-p\right|\le\epsilon,
\qquad
\operatorname{Var}(\widetilde p)
\le
\frac{91p}{t^2}+\epsilon.
\]
It uses one copy of \(\ket{\psi}\), which is restored at the end with
probability at least \(1-\epsilon\), and an expected
\(O(t\log\log t\log(t/\epsilon))\) applications of the reflections
\(I-2\ket{\psi}\!\bra{\psi}\) and \(I-2\Pi\).
\end{thm}

\medskip

In Algorithm~\ref{alg:local-log-estimate}, we set \(t=M\) and use
\(\eta\) as the internal accuracy parameter.  Writing
\(c_v:=91\) and denoting the resulting estimate by
\(\widehat p_{M,\eta}\), Theorem~\ref{thm:ndub-amplitude} gives
\begin{equation}
\label{eq:single-temperature-ae-moments-main}
\left|\mathbb E[\widehat p_{M,\eta}]-p\right|\le\eta,
\qquad
\operatorname{Var}(\widehat p_{M,\eta})
\le
\frac{c_vp}{M^2}+\eta
\le
\frac{c_v}{M^2}+\eta.
\end{equation}
We now use these bounds to prove the bias, variance, and query estimates for
the logarithmic estimator in
Lemma~\ref{lem:almost-unbiased-local}.

\begin{proof}[Proof of Lemma~\ref{lem:almost-unbiased-local}]
Define
\[
\Lambda_\gamma:=\ln\frac{2}{\gamma},
\qquad
q_+(x):=
\ln\!\left(\min\!\left\{2,\max\!\left\{\frac{\gamma}{2},x\right\}\right\}\right),
\qquad
q_-(x):=-q_+(x).
\]
The algorithm returns \(q_-(\widehat p_{M,\eta})\) for an overlap term and
\(q_+(\widehat p_{M,\eta})+2\) for a reweighting term. The unshifted maps
\(q_+\) and \(q_-\) are \(2/\gamma\)-Lipschitz and have absolute value
at most \(\Lambda_\gamma\). Adding the constant \(2\) does not change
the Lipschitz bound or the variance.

\medskip

For either sign, the Lipschitz bound and
Eq.~\eqref{eq:single-temperature-ae-moments-main} give
\[
\operatorname{Var}(\widehat Y_{M,\eta})
\le
\frac{4}{\gamma^2}
\left(\frac{c_v}{M^2}+\eta\right)
\le
K_\gamma\left(M^{-2}+\eta\right).
\]

\medskip

It remains to bound the bias.  Write
\(\Delta:=\widehat p_{M,\eta}-p\).  Since \(0<\eta<1\),
\begin{equation}
\label{eq:single-temperature-local-second-moment}
\mathbb E[\Delta^2]
=
\operatorname{Var}(\widehat p_{M,\eta})
+\bigl(\mathbb E[\Delta]\bigr)^2
\le
\frac{c_v}{M^2}+2\eta.
\end{equation}
Let \(\mathcal E:=\{|\Delta|\le\gamma/2\}\).  Since
\(p\in[\gamma,1]\), clipping is inactive on \(\mathcal E\).  Let
\(g(x)=\ln x+2\) for a reweighting term and \(g(x)=-\ln x\) for an overlap
term.  Taylor's theorem about \(p\), using
\(|g'(p)|\le1/\gamma\) and
\(|g''(x)|\le4/\gamma^2\) on the relevant interval, gives
\[
\left|
\mathbb E\!\left[
\bigl(\widehat Y_{M,\eta}-Y\bigr)\mathbf 1_{\mathcal E}
\right]
\right|
\le
\frac{1}{\gamma}
\left|
\mathbb E[\Delta\mathbf 1_{\mathcal E}]
\right|
+
\frac{2}{\gamma^2}\mathbb E[\Delta^2].
\]
Moreover,
\[
\left|
\mathbb E[\Delta\mathbf 1_{\mathcal E}]
\right|
\le
|\mathbb E[\Delta]|
+
\left|
\mathbb E[\Delta\mathbf 1_{\mathcal E^c}]
\right|
\le
\eta+\frac{2}{\gamma}\mathbb E[\Delta^2],
\]
where on \(\mathcal E^c\) we used
\(|\Delta|\le2\Delta^2/\gamma\).  Hence the contribution from
\(\mathcal E\) is at most
\[
\frac{\eta}{\gamma}
+
\frac{4}{\gamma^2}\mathbb E[\Delta^2].
\]

\medskip

By Markov's inequality,
\(\Pr(\mathcal E^c)\le
4\mathbb E[\Delta^2]/\gamma^2\).
The additive correction cancels from \(\widehat Y_{M,\eta}-Y\).
Each unshifted logarithm is bounded in absolute value by \(\Lambda_\gamma\),
so \(|\widehat Y_{M,\eta}-Y|\le2\Lambda_\gamma\) for either branch.
The contribution from \(\mathcal E^c\) is therefore at most
\(8\Lambda_\gamma\mathbb E[\Delta^2]/\gamma^2\).  Combining the two events
and using
Eq.~\eqref{eq:single-temperature-local-second-moment} gives
\[
\left|
\mathbb E[\widehat Y_{M,\eta}]-Y
\right|
\le
\frac{\eta}{\gamma}
+
\frac{4+8\Lambda_\gamma}{\gamma^2}
\left(\frac{c_v}{M^2}+2\eta\right)
\le
K_\gamma\left(M^{-2}+\eta\right),
\]
with \(K_\gamma\) as defined in
Eq.~\eqref{eq:single-temperature-K-gamma-main}.  This proves the bias and
variance bounds.  If \(\eta\le M^{-2}\), both are therefore at most
\(2K_\gamma/M^2\).

\medskip

We finally count oracle calls.  For an overlap term, the prepared state is
obtained from
\(U_\theta^\dagger U_{\theta+\delta}\).  One reflection about this state
uses the preparation and its inverse, and hence four fixed-temperature TFD
calls and no block-encoding query.  For a reweighting term, one reflection
about the prepared state uses \(U_\theta\), \(U_\theta^\dagger\),
\(D_\delta\), and \(D_\delta^\dagger\).  Thus it uses two
fixed-temperature TFD calls and \(O(q)\) block-encoding queries when
\(D_\delta\) has QSVT degree \(q\).

\medskip

By Theorem~\ref{thm:ndub-amplitude}, the expected number of such reflections
is
\(O(M\log\log M\log(M/\eta))=\widetilde O(M)\).
Therefore an overlap term uses
\(N_{\mathrm{TFD}}=\widetilde O(M)\) and \(N_{\mathrm{BE}}=0\), while a
reweighting term uses
\(N_{\mathrm{TFD}}=\widetilde O(M)\) and
\(N_{\mathrm{BE}}=\widetilde O(Mq)\).
No coherent-temperature TFD call is used.  Under
\(\alpha=\Theta(h)\), \(h\delta\le1\), and inverse-polynomial QSVT
accuracy, standard QSVT approximation gives
\(q=\widetilde O(1)\)~\cite{Gilyen2019QSVT}.
\end{proof}

\subsection{Stability of the reweighting dilation}
\label{app:single-temperature-dilation-stability}

For a reweighting term, let \(A_i:=e^{-\delta_iH/2}\), and let
\(\widetilde A_i\) denote the all-zero block of the QSVT unitary
\(D_{\delta_i}\) from Lemma~\ref{lem:qsvt-imaginary-time}.
The following lemma controls the implemented success probability and its
logarithm after the known scaling correction.

\begin{lem}[Probability and logarithm stability]
\label{lem:single-temperature-dilation-stability}
Suppose that \(\|\widetilde A_i-e^{-1}A_i\|\le\nu\). Let
\[
R(\beta_i,\delta_i)
:=\|(A_i\otimes I_R)\ket{\psi_{\beta_i}}\|^2,
\qquad
p_{i,0}
:=\|(\widetilde A_i\otimes I_R)\ket{\psi_{\beta_i}}\|^2,
\]
where \(R\) in \(I_R\) denotes the TFD reference register.
Then
\[
|p_{i,0}-e^{-2}R(\beta_i,\delta_i)|\le2\nu.
\]
If, in addition,
\[
e^{-2}R(\beta_i,\delta_i)\ge\gamma_0,
\qquad
2\nu\le\frac{\gamma_0}{2},
\]
then
\[
p_{i,0}\ge\frac{\gamma_0}{2},
\qquad
|\ln p_{i,0}+2-\ln R(\beta_i,\delta_i)|
\le\frac{4\nu}{\gamma_0}.
\]
\end{lem}

\begin{proof}
Since \(H\succeq0\), the operator \(e^{-1}A_i\) is a contraction.
The operator \(\widetilde A_i\) is also a contraction because it is a
block of the unitary \(D_{\delta_i}\). Hence
\[
\begin{aligned}
|p_{i,0}-e^{-2}R(\beta_i,\delta_i)|
&\le\|\widetilde A_i^\dagger\widetilde A_i-e^{-2}A_i^\dagger A_i\|\\
&\le\|\widetilde A_i^\dagger(\widetilde A_i-e^{-1}A_i)\|
+\|(\widetilde A_i^\dagger-e^{-1}A_i^\dagger)e^{-1}A_i\|\\
&\le2\nu.
\end{aligned}
\]
Under the additional assumptions,
\[
p_{i,0}\ge e^{-2}R(\beta_i,\delta_i)-2\nu\ge\frac{\gamma_0}{2}.
\]
Both \(p_{i,0}\) and \(e^{-2}R(\beta_i,\delta_i)\) therefore lie in
\([\gamma_0/2,1]\). The mean-value theorem gives
\[
\begin{aligned}
|\ln p_{i,0}+2-\ln R(\beta_i,\delta_i)|
&=|\ln p_{i,0}-\ln(e^{-2}R(\beta_i,\delta_i))|\\
&\le\frac{2}{\gamma_0}|p_{i,0}-e^{-2}R(\beta_i,\delta_i)|
\le\frac{4\nu}{\gamma_0}.
\end{aligned}
\]
\end{proof}

\medskip

Let \(\widetilde f(\beta)\) denote the deterministic reconstruction
obtained from Eq.~\eqref{eq:full-dyadic-reconstruction} by replacing each
\(\ln R(\beta_i,\delta_i)\) with \(\ln p_{i,0}+2\), leaving the overlap
terms unchanged. Therefore
\[
\widetilde f(\beta)-f(\beta)
=
\sum_{i=0}^{\ell-1}2^{k_i}
\bigl(\ln p_{i,0}+2-\ln R(\beta_i,\delta_i)\bigr).
\]
The uniform bound \(e^{-2}R(\beta_i,\delta_i)\ge\gamma_0\) follows from
Eq.~\eqref{eq:uniform-conditioning}. With the choice
\[
\nu=\frac{\gamma_0\epsilon}{64\sum_{i=0}^{\ell-1}2^{k_i}}
\]
in Algorithm~\ref{alg:finite-temperature-reconstruction}, we have
\(2\nu\le\gamma_0/2\). Thus the stability lemma gives
\begin{equation}
\label{eq:single-temperature-total-dilation-bias}
|\widetilde f(\beta)-f(\beta)|
\le\frac{4\nu}{\gamma_0}\sum_{i=0}^{\ell-1}2^{k_i}
=\frac{\epsilon}{16}.
\end{equation}
The QSVT approximation therefore contributes at most \(\epsilon/16\)
to the reconstruction error.

\medskip

Finally, since \(\alpha=\Theta(h)\) and \(h\delta_i\le1\),
Lemma~\ref{lem:qsvt-imaginary-time} implements an approximation to
\(e^{-1}e^{-\delta_iH/2}\) with QSVT degree \(q=\widetilde O(1)\),
where the suppressed dependence on \(\nu^{-1}\) is polylogarithmic.

\subsection{Choice of the local estimation depths}
\label{app:local-depth-allocation}

We derive the depths \(M_{i,j}\) used in
Algorithm~\ref{alg:finite-temperature-reconstruction}.
Fix one repetition \(r\) before taking the final median, and let
\(\widehat Y_{i,j}\) denote its local logarithmic estimates.
Let \(a_{i,0}=2^{k_i}\), \(a_{i,j}=2^{k_i-j}\) for \(j\ge1\), and define
\begin{equation}
\label{eq:single-temperature-depth-S}
S:=\sum_{i=0}^{\ell -1}
\sum_{j=0}^{k_i}a_{i,j}^{2/3}.
\end{equation}
We choose \(\eta_{i,j}=M_{i,j}^{-2}\).  By
Lemma~\ref{lem:almost-unbiased-local},
\(\operatorname{Var}(\widehat Y_{i,j})
\le 2K_\gamma/M_{i,j}^2\).
Since the local runs are independent, the weighted reconstruction satisfies
\[
\operatorname{Var}(\widehat f^{(r)}(\beta))
\le
2K_\gamma
\sum_{i,j}\frac{a_{i,j}^2}{M_{i,j}^2}.
\]
It is therefore sufficient to impose
\(
\sum_{i,j}a_{i,j}^2/M_{i,j}^2
\le \epsilon_{\rm stat}^2/(16K_\gamma)
\).

\medskip

Up to logarithmic factors, a local estimate of depth \(M_{i,j}\) costs
\(O(M_{i,j})\) queries.  Ignoring integer rounding, we thus minimize
\(\sum_{i,j}M_{i,j}\) subject to the variance constraint above.  Introducing
a Lagrange multiplier \(\lambda\) gives
\[
\mathcal L
=
\sum_{i,j}M_{i,j}
+
\lambda\left(
\sum_{i,j}\frac{a_{i,j}^2}{M_{i,j}^2}
-
\frac{\epsilon_{\rm stat}^2}{16K_\gamma}
\right),
\qquad
\frac{\partial\mathcal L}{\partial M_{i,j}}
=
1-\frac{2\lambda a_{i,j}^2}{M_{i,j}^3}.
\]
Hence the stationary point satisfies
\(M_{i,j}\propto a_{i,j}^{2/3}\).  Writing
\(M_{i,j}=c\,a_{i,j}^{2/3}\) and substituting into the variance constraint
gives \(S/c^2=\epsilon_{\rm stat}^2/(16K_\gamma)\), and therefore
\begin{equation}
\label{eq:single-temperature-optimal-depth}
M_{i,j}^{*}
=
\frac{
4\sqrt{K_\gamma}\,S^{1/2}a_{i,j}^{2/3}
}{
\epsilon_{\rm stat}
}.
\end{equation}

\medskip

The algorithm uses
\[
M_{i,j}
=
\max\!\left\{
5,\left\lceil M_{i,j}^{*}\right\rceil
\right\},
\qquad
\eta_{i,j}=M_{i,j}^{-2}.
\]
Increasing \(M_{i,j}\) can only decrease the variance, so rounding preserves
the variance bound.  Since \(a_{i,j}\ge1\), the number of local terms is at
most \(S\), and hence
\[
\sum_{i,j}M_{i,j}
=
O\!\left(
S+\frac{\sqrt{K_\gamma}\,S^{3/2}}{\epsilon_{\rm stat}}
\right)
=
O\!\left(\frac{S^{3/2}}{\epsilon}\right),
\]
where the last equality uses
\(\epsilon_{\rm stat}=\epsilon/2\) and
\(K_\gamma=O(1)\) for \(C=O(1)\).

\medskip

The same choice also controls the accumulated bias.  Indeed,
Lemma~\ref{lem:almost-unbiased-local} gives
\[
\left|
\mathbb E[\widehat f^{(r)}(\beta)]-\widetilde f(\beta)
\right|
\le
\frac{\epsilon_{\rm stat}^2}{8S}
\sum_{i,j}a_{i,j}^{-1/3}
\le
\frac{\epsilon_{\rm stat}^2}{8},
\]
because \(a_{i,j}^{-1/3}\le a_{i,j}^{2/3}\).
Thus Eq.~\eqref{eq:single-temperature-optimal-depth} gives the minimum
continuous total depth under the required variance constraint and provides
the depth allocation used in
Algorithm~\ref{alg:finite-temperature-reconstruction}.

\subsection{Finite-temperature normalization}
\label{app:finite-temperature-details}

The reconstruction starts from the exact value \(f(0)=n\ln2\).  For every
finite \(\beta\), \(0\preceq H\preceq hI\) implies
\[
2^n e^{-\beta h}\le Z(\beta)\le2^n,
\qquad
n\ln2-\beta h\le f(\beta)\le n\ln2.
\]
If the normalization additionally satisfies \(\lambda_{\min}(H)=0\), then
\(Z(\beta)\ge1\) and hence \(f(\beta)\ge0\).  No zero-temperature limit is
used in either the algorithm or its proof.
\section{Implementation and proofs for coherent temperature access}
\label{app:sampled-dyadic}

This appendix gives the circuit-level implementation and the complete error
and resource accounting for Section~\ref{sec:sampled-dyadic}. The algorithm
itself remains in the main text. We use the query counts defined in
Definitions~\ref{def:partition-function-resource-model}
and~\ref{def:coherent-temperature-access};
branch-state preparation, parameter loading,
and reversible arithmetic are additional gate costs.

\subsection{Preparing and loading the branch}
\label{subsec:coherent-branch-preparation}

For each interval \(i\),
\[
\sum_{j=0}^{k_i}a_{i,j}
=
\sum_{j=0}^{k_i}2^{k_i-j}
=
2^{k_i+1}-1.
\]
Hence the branch-preparation unitary \(U_\pi\) can be implemented coherently
in two stages.  Starting from the all-zero branch register, first prepare
\[
\sum_{i=0}^{\ell-1}
\sqrt{\frac{2^{k_i+1}-1}{L}}\ket{i}.
\]
Conditioned coherently on \(i\), prepare the second branch register as
\[
\ket{i}\ket0
\longmapsto
\ket{i}
\sum_{j=0}^{k_i}
\sqrt{
\frac{2^{k_i-j}}{2^{k_i+1}-1}
}
\ket{j}.
\]
The resulting state is
\[
\sum_{(i,j)\in\mathcal R}
\sqrt{\frac{2^{k_i-j}}{L}}\ket{i,j}
=
\sum_{(i,j)\in\mathcal R}
\sqrt{\pi_{i,j}}\ket{i,j},
\]
which is exactly the branch state required by
Definition~\ref{def:coherent-temperature-access}.  No measurement or
classical sampling of \(i\) or \(j\) is performed.

\medskip

The coherent preparation of this branch state is part of the access model
in Definition~\ref{def:coherent-temperature-access}.  Throughout this
section, \(U_\pi\) is assumed to be implemented exactly.  Its gate cost is
not included in \(N_{\mathrm{TFD}}\), \(N_{\mathrm{cohTFD}}\), or
\(N_{\mathrm{BE}}\).

\medskip

After preparing the branch register, apply the reversible loading unitary
\[
U_{\mathrm{load}}:
\ket{i,j}\ket0
\longmapsto
\ket{i,j}
\ket{\beta_i,\delta_i,2^j\delta_i,
\beta_i+2^j\delta_i}.
\]
Thus each branch remains in coherent superposition while the parameters
required by the corresponding overlap or reweighting construction are made
available in separate registers.  The inverse \(U_{\mathrm{load}}^\dagger\)
is used when the coherent sample is uncomputed.

\medskip

The schedule values are treated as exact classical bit strings supplied to
the algorithm.  Approximate parameter loading would require an additional
error budget and is not included in the stated query bounds.

\subsection{Coherent implementation of the two branch types}
\label{subsec:coherent-branch-implementation}

Let \(I,J\) be the branch registers and let \(P\) contain the parameters
loaded for each branch. Let \(Q=S\otimes R\) be the TFD registers, where
\(S\) is the physical register and \(R\) is the reference register. Let \(A\) be an \(a\)-qubit QSVT ancilla register and \(F\) a
one-qubit flag register.
After preparing the branch distribution and loading the corresponding
parameters, the state has the form
\[
\ket{\Psi_{\mathrm{in}}}
=
\sum_{(i,j)\in\mathcal R}
\sqrt{\pi_{i,j}}\,
\ket{i,j}_{IJ}
\ket{\mathrm{par}_{i,j}}_P
\ket0_Q\ket{0^a}_A\ket0_F .
\]
The registers \(I,J\) are not measured. The parameters and branch-dependent
operations below are selected coherently by \((i,j)\).

\medskip

By Lemma~\ref{lem:local-quantity-projection} and
Eq.~\eqref{eq:implemented-local-probabilities}, the overlap branches
\(j\ge1\) realize the exact probabilities
\(p_{i,j}=O(\beta_i,2^j\delta_i)\), while the reweighting branches
\(j=0\) realize the QSVT-implemented probabilities
\(p_{i,0}=\widetilde R(\beta_i,\delta_i)\).
We now describe the two branch actions.

\medskip

\noindent\textbf{Overlap branches \(j\ge1\).}
Fix a branch \((i,j)\) with \(1\le j\le k_i\). Using the parameters
stored in \(\ket{\mathrm{par}_{i,j}}_P\), apply
\(U_{\mathrm{cohTFD}}\) at inverse temperature
\(\beta_i+2^j\delta_i\) to prepare
\(\ket{\psi_{\beta_i+2^j\delta_i}}_Q\), and then apply
\(U_{\mathrm{cohTFD}}^\dagger\) at inverse temperature \(\beta_i\).
The amplitude of the all-zero state on \(Q\) is
\[
\left\langle
\psi_{\beta_i}
\,\middle|\,
\psi_{\beta_i+2^j\delta_i}
\right\rangle .
\]
Hence, by Lemma~\ref{lem:local-quantity-projection}, the corresponding
success probability is
\[
p_{i,j}=O(\beta_i,2^j\delta_i).
\]
A reversible all-zero test on \(Q\) writes this outcome into \(F\), so that
\(\Pr(F=1\mid i,j)=p_{i,j}\). This probability is exact and involves no
QSVT approximation.

\medskip

One forward overlap action uses one call to
\(U_{\mathrm{cohTFD}}\) and one call to
\(U_{\mathrm{cohTFD}}^\dagger\), and no Hamiltonian block-encoding query.
The inverse action has the same query cost.

\medskip

\noindent\textbf{Reweighting branches \(j=0\).}
Fix a branch \((i,0)\). With the inverse-temperature input set to
\(\beta_i\), apply \(U_{\mathrm{cohTFD}}\) to prepare
\(\ket{\psi_{\beta_i}}_Q\). The reweighting term requires
\(e^{-1}e^{-\delta_iH/2}\) on the physical register \(S\). We implement this
scaled operator using the QSVT unitary \(D_{\delta_i}\) from
Lemma~\ref{lem:qsvt-imaginary-time}, acting on \(A\otimes S\) and
leaving \(R\) unchanged.

\medskip

If \(\widetilde A_{\delta_i}\) denotes the all-zero block of
\(D_{\delta_i}\), then
\[
(\bra{0^a}_A\otimes I_S)
D_{\delta_i}
(\ket{0^a}_A\otimes I_S)
=
\widetilde A_{\delta_i},
\qquad
\left\|
\widetilde A_{\delta_i}
-
e^{-1}e^{-\delta_iH/2}
\right\|
\le\nu .
\]
Therefore
\[
(D_{\delta_i}\otimes I_R)
\bigl(\ket{0^a}_A\ket{\psi_{\beta_i}}_Q\bigr)
=
\ket{0^a}_A
(\widetilde A_{\delta_i}\otimes I_R)
\ket{\psi_{\beta_i}}_Q
+
\ket{\perp_{i,0}}_{AQ},
\]
where
\(
(\bra{0^a}_A\otimes I_Q)\ket{\perp_{i,0}}_{AQ}=0
\).

\medskip

A reversible all-zero test on \(A\) flips the flag qubit \(F\). Hence the
implemented branch probability is
\[
p_{i,0}
=
\Pr(F=1\mid i,0)
=
\left\|
(\widetilde A_{\delta_i}\otimes I_R)
\ket{\psi_{\beta_i}}
\right\|^2.
\]
The corresponding exact reweighting quantity is
\[
R(\beta_i,\delta_i)
=
\Tr\!\left(\rho_{\beta_i}e^{-\delta_iH}\right).
\]
Lemma~\ref{lem:single-temperature-dilation-stability} bounds
\(|p_{i,0}-e^{-2}R(\beta_i,\delta_i)|\le2\nu\) and, under the uniform
conditioning assumption,
\(|\ln p_{i,0}+2-\ln R(\beta_i,\delta_i)|\le4\nu/\gamma_0\).

\medskip

At the accuracy used in
Algorithm~\ref{alg:sampled-dyadic-reconstruction},
Lemma~\ref{lem:qsvt-imaginary-time} gives QSVT degree
\(\widetilde{\mathcal O}(1)\) uniformly in \(i\), since
\(\alpha=\Theta(h)\) and \(h\delta_i\le1\).
Thus the standalone reweighting action uses one call to
\(U_{\mathrm{cohTFD}}\) and
\(\widetilde{\mathcal O}(1)\) Hamiltonian block-encoding queries.
Its inverse has the same query cost.

\medskip

\noindent\textbf{Combined branch unitary.}
Let \(W_A\) and \(W_Q\) denote the reversible all-zero tests on \(A\) and
\(Q\), respectively:
\[
W_A
:=
\ket{0^a}\!\bra{0^a}_A\otimes X_F
+
\left(I_A-\ket{0^a}\!\bra{0^a}_A\right)\otimes I_F ,
\qquad
W_Q
:=
\ket0\!\bra0_Q\otimes X_F
+
\left(I_Q-\ket0\!\bra0_Q\right)\otimes I_F .
\]
These tests flip \(F\) exactly when the corresponding tested register is
all zero, without otherwise changing that register.

\medskip

Define the branch actions
\[
U^{\mathrm{rw}}_i
:=
U_{\mathrm{cohTFD}}^\dagger
W_A(D_{\delta_i}\otimes I_R)U_{\mathrm{cohTFD}},
\qquad
U^{\mathrm{ov}}_{i,j}
:=
W_Q U_{\mathrm{cohTFD}}^\dagger U_{\mathrm{cohTFD}},
\]
where both TFD calls in \(U^{\mathrm{rw}}_i\) use temperature \(\beta_i\), while in
\(U^{\mathrm{ov}}_{i,j}\) the forward call uses
\(\beta_i+2^j\delta_i\) and the adjoint call uses \(\beta_i\).
The final adjoint in \(U^{\mathrm{rw}}_i\) acts only on \(Q\), so it
preserves the flag probability computed above.

\medskip

On the valid branch subspace, define
\begin{equation}
\label{eq:coherent-branch-definition}
U_{\mathrm{branch}}
:=
\sum_{i=0}^{\ell-1}
\ket{i,0}\!\bra{i,0}_{IJ}\otimes U^{\mathrm{rw}}_i
+
\sum_{i=0}^{\ell-1}
\sum_{j=1}^{k_i}
\ket{i,j}\!\bra{i,j}_{IJ}\otimes U^{\mathrm{ov}}_{i,j},
\end{equation}
and extend it as the identity on the orthogonal complement of this
subspace. Since the branch projectors are mutually orthogonal and each
branch action is unitary, this defines a unitary.

\medskip

To implement this direct sum, first load \(\beta_i\) on reweighting
branches and \(\beta_i+2^j\delta_i\) on overlap branches into a temporary
temperature register, and apply \(U_{\mathrm{cohTFD}}\) unconditionally.
Apply \(D_{\delta_i}\) and \(W_A\) only on the reweighting branches.
Reversibly replace the temperature by \(\beta_i\), apply
\(U_{\mathrm{cohTFD}}^\dagger\) unconditionally, and apply \(W_Q\) only
on the overlap branches. Unload the temperature register; the branch and
parameter registers remain unchanged, and the temporary register returns
to zero. For invalid branch indices, use temperature zero in both calls
and omit the QSVT operation and both tests, so the two calls cancel
exactly. No externally controlled coherent-TFD call is required.

\medskip

For every \((i,j)\in\mathcal R\), its action on the initialized workspace is
\begin{equation}
\label{eq:coherent-branch-unitary}
\begin{aligned}
&U_{\mathrm{branch}}
\ket{i,j}_{IJ}\ket{\mathrm{par}_{i,j}}_P
\ket0_Q\ket{0^a}_A\ket0_F
\\
&\qquad=
\ket{i,j}_{IJ}\ket{\mathrm{par}_{i,j}}_P
\left(
\sqrt{p_{i,j}}\,
\ket{\phi^1_{i,j}}_{QA}\ket1_F
+
\sqrt{1-p_{i,j}}\,
\ket{\phi^0_{i,j}}_{QA}\ket0_F
\right),
\end{aligned}
\end{equation}
where \(\ket{\phi^1_{i,j}}\) and \(\ket{\phi^0_{i,j}}\) are normalized
states of \(Q\otimes A\). Thus
\(\Pr(F=1\mid i,j)=p_{i,j}\), while the branch and parameter registers are
unchanged. By linearity, the same unitary acts coherently on the full branch
superposition.

\medskip

The two coherent-TFD calls are shared by all branches. Thus one application
of \(U_{\mathrm{branch}}\) uses one forward and one adjoint coherent-TFD
call, no fixed-temperature TFD call, and
\(\widetilde{\mathcal O}(1)\) block-encoding queries. For the latter bound,
pad the branch-dependent QSVT circuits to a common query schedule and use
the supplied controlled block-encoding access to enable the required
queries on each reweighting branch.
The adjoint \(U_{\mathrm{branch}}^\dagger\) has the same query cost.
\subsection{Reversible harmonic-number sample}
\label{subsec:reversible-harmonic-sample}

\begin{proof}[Proof of
Lemma~\ref{lem:coherent-logarithmic-sampling-summary}]
Fix a branch with success probability \(\mu\in[\gamma,1]\).
For \(1\le m\le M\),
\(\Pr(G_M\ge m)=(1-\mu)^m\). Hence the tail-sum identity gives
\[
\mathbb E[H_{G_M}]
=
\sum_{m=1}^{M}\frac{(1-\mu)^m}{m}.
\]
Since
\(-\ln\mu=\sum_{m=1}^{\infty}(1-\mu)^m/m\), the truncation bias is
nonnegative and satisfies
\begin{align*}
-\ln\mu-\mathbb E[H_{G_M}]
&=
\sum_{m=M+1}^{\infty}\frac{(1-\mu)^m}{m}\\
&\le
\frac{(1-\mu)^{M+1}}{\mu(M+1)}
\le
\frac{e^{-\gamma(M+1)}}{\gamma(M+1)}.
\end{align*}
Moreover, \(0\le H_{G_M}\le H_M\), so
\(H_{G_M}^2\le H_MH_{G_M}\). Taking expectations and using
\(\mathbb E[H_{G_M}]\le-\ln\mu\le\ln(1/\gamma)\) proves the
second-moment bound.

\medskip

For the coherent implementation, use \(M\) fresh flag and workspace
registers and apply \(U_{\mathrm{branch}}\) once to each pair, without
measurement. Conditioned on the fixed branch, let \(b_m\) be the success
flag of the \(m\)-th trial. Reversibly compute
\[
r_m:=\prod_{t=1}^{m}(1-b_t),
\qquad 1\le m\le M,
\]
so that
\[
H_{G_M}
=
\sum_{m=1}^{M}\frac{r_m}{m}.
\]
Reversible Boolean circuits compute the \(r_m\), and controlled fixed-point
additions approximate this sum within \(\zeta\). The branch register \(J\)
selects the logarithmic sample \(2-\widetilde H_{G_M}\) for \(j=0\)
and \(\widetilde H_{G_M}\) for \(j\ge1\). All these computations are
reversible, so the complete coherent preparation is unitary and its inverse
is obtained by reversing the circuit. The full register-level construction is given below, in the proof of
Lemma~\ref{lem:full-coherent-sample-preparation}.
This proves the claimed coherent implementation.
\end{proof}

\medskip

Let \(\widehat Y_{i,0}:=2-\widetilde H_{G_M}\) and
\(\widehat Y_{i,j}:=\widetilde H_{G_M}\) for \(j\ge1\) be the local
logarithmic samples, and define
\[
\widetilde Y_{i,j}
:=
\begin{cases}
\ln p_{i,0}+2, & j=0,\\
-\ln p_{i,j}, & j\ge1.
\end{cases}
\]
With the parameters in
Algorithm~\ref{alg:sampled-dyadic-reconstruction}, every implemented branch
probability satisfies \(p_{i,j}\ge\gamma=\gamma_0/2\). Therefore
\begin{equation}
\label{eq:coherent-local-bias-second-moment}
\left|
\mathbb E[\widehat Y_{i,j}]
-
\widetilde Y_{i,j}
\right|
\le
\frac{\epsilon}{8L},
\qquad
\mathbb E[\widehat Y_{i,j}^{\,2}]
\le
4H_M\ln(1/\gamma)+4\zeta^2+8
=\widetilde O(1).
\end{equation}
The first bound combines the truncation error
\(\epsilon/(16L)\) and the arithmetic error
\(\zeta=\epsilon/(16L)\); the constant correction cancels when subtracting
\(\widetilde Y_{i,j}\). For the second bound, both branch types satisfy
\[
\widehat Y_{i,j}^{\,2}\le2\widetilde H_{G_M}^{\,2}+8,
\]
so Eq.~\eqref{eq:coherent-fixed-point-log-bounds} gives the stated uniform
bound. For \(j\ge1\), \(\widetilde Y_{i,j}=Y_{i,j}\), whereas for \(j=0\)
the difference between \(\widetilde Y_{i,0}\) and \(Y_{i,0}\) is controlled
by Lemma~\ref{lem:single-temperature-dilation-stability}.

\medskip

Preparing one coherent logarithmic sample, excluding branch-state
preparation and arithmetic gates, uses \(M\) applications of
\(U_{\mathrm{branch}}\), and its inverse uses \(M\) applications of
\(U_{\mathrm{branch}}^\dagger\). Since each branch trial uses at most two
coherent-TFD calls and \(\widetilde O(1)\) block-encoding queries, one
reflection about the prepared sample state uses at most \(4M\)
coherent-TFD calls, \(\widetilde O(M)\) block-encoding queries, and no
fixed-temperature TFD call.

\begin{proof}[Proof of Lemma~\ref{lem:full-coherent-sample-preparation}]
\label{proof:full-coherent-sample-preparation}
Use the branch set \(\mathcal R\), distribution \(\pi\), scale \(L\),
truncation length \(M\), and accuracy \(\zeta\) from the statement.
The registers \(I,J\) store the branch index, and \(P\) stores
the corresponding parameter string \(\mathrm{par}_{i,j}\).
The branch-preparation and parameter-loading unitaries satisfy
\[
U_\pi\ket0_{IJ}
=
\sum_{(i,j)\in\mathcal R}
\sqrt{\pi_{i,j}}\ket{i,j}_{IJ},
\]
and
\[
U_{\mathrm{load}}
\ket{i,j}_{IJ}\ket0_P
=
\ket{i,j}_{IJ}\ket{\mathrm{par}_{i,j}}_P.
\]

\medskip

For \(m=1,\ldots,M\), let \(F_m\) be the success-flag register
for the \(m\)-th trial, and let
\(U_{\mathrm{branch}}^{(m)}\) act on the shared \(I,J,P\)
registers and a fresh flag and workspace:
\[
\begin{aligned}
&U_{\mathrm{branch}}^{(m)}
\ket{i,j}_{IJ}\ket{\mathrm{par}_{i,j}}_P\ket0
\\
&\quad=
\ket{i,j}_{IJ}\ket{\mathrm{par}_{i,j}}_P
\left(
\sqrt{p_{i,j}}\ket1_{F_m}\ket{\phi^1_{i,j}}
+
\sqrt{1-p_{i,j}}\ket0_{F_m}\ket{\phi^0_{i,j}}
\right).
\end{aligned}
\]
Here \(p_{i,j}\) is the branch success probability, and
\(\ket{\phi^b_{i,j}}\) is the normalized conditional workspace
state for flag value \(b\in\{0,1\}\).
For a zero-probability outcome, this state may be chosen arbitrarily.
The workspace registers of different trials are disjoint.
Since the branch and parameter registers are unchanged, the
flags are independent conditional on \((i,j)\).

\medskip

Write \(F=F_1\cdots F_M\) for the collection of flag registers.
Introduce registers \(R_{\mathrm{pref}}\), \(H\), and \(X\) to store
the prefix-failure indicators, the fixed-point harmonic value,
and the scaled logarithmic sample, respectively.
All registers are initially zero.
The arithmetic maps below retain their inputs and are implemented
reversibly, with any temporary arithmetic workspace returned to zero.

\medskip

First, define \(U_r\) by
\[
U_r:
\ket{b_1,\ldots,b_M}_F\ket{0^M}_{R_{\mathrm{pref}}}
\longmapsto
\ket{b_1,\ldots,b_M}_F
\ket{r_1,\ldots,r_M}_{R_{\mathrm{pref}}},
\qquad
r_m:=\prod_{t=1}^{m}(1-b_t).
\]
Thus \(r_m=1\) exactly when the first \(m\) trials all fail.

\medskip

Choose fixed-point constants \(\widetilde c_m\), represented
at a common precision, such that
\[
\left|\widetilde c_m-\frac1m\right|
\le\frac{\zeta}{M},
\qquad m=1,\ldots,M.
\]
Use a sufficiently large accumulator to perform the following
sum without further rounding or overflow:
\[
U_H:
\ket{r_1,\ldots,r_M}_{R_{\mathrm{pref}}}\ket0_H
\longmapsto
\ket{r_1,\ldots,r_M}_{R_{\mathrm{pref}}}\ket{\widetilde H}_H,
\qquad
\widetilde H:=\sum_{m=1}^{M}r_m\widetilde c_m.
\]
Since \(r_m\in\{0,1\}\),
\[
\left|
\widetilde H-\sum_{m=1}^{M}\frac{r_m}{m}
\right|
\le
\sum_{m=1}^{M}
r_m\left|\widetilde c_m-\frac1m\right|
\le\zeta.
\]

\medskip

Using the sample-value map from the statement,
\[
x_j(y)
=
\begin{cases}
L(2-y), & j=0,\\
Ly, & j\ge1,
\end{cases}
\]
define \(U_X\) by
\[
U_X:
\ket{i,j}_{IJ}\ket{\widetilde H}_H\ket0_X
\longmapsto
\ket{i,j}_{IJ}\ket{\widetilde H}_H
\ket{x_j(\widetilde H)}_X.
\]
The correction \(2\) and the integer scale \(L\) are represented exactly;
with a sufficiently large value register, this computation introduces no
additional rounding error or oracle queries.
The full coherent sampler is
\[
U_{\mathrm{sample}}
:=
U_XU_HU_r
\left(
U_{\mathrm{branch}}^{(M)}
\cdots
U_{\mathrm{branch}}^{(1)}
\right)
U_{\mathrm{load}}U_\pi.
\]

\medskip

For \(\mathbf b=(b_1,\ldots,b_M)\in\{0,1\}^M\), write
\[
\begin{aligned}
q_{i,j}(\mathbf b)
&:=
\prod_{m=1}^{M}
p_{i,j}^{\,b_m}(1-p_{i,j})^{1-b_m},\\
\ket{\Phi_{i,j,\mathbf b}}
&:=
\bigotimes_{m=1}^{M}\ket{\phi^{b_m}_{i,j}},\\
r(\mathbf b)
&:=
\bigl(r_1(\mathbf b),\ldots,r_M(\mathbf b)\bigr),\\
\widetilde H(\mathbf b)
&:=
\sum_{m=1}^{M}r_m(\mathbf b)\widetilde c_m.
\end{aligned}
\]
Here \(\ket{\Phi_{i,j,\mathbf b}}\) is the joint conditional
state of all trial workspaces.
Expanding the circuit gives
\[
\begin{aligned}
U_{\mathrm{sample}}\ket0
={}&
\sum_{(i,j)\in\mathcal R}
\sum_{\mathbf b\in\{0,1\}^M}
\sqrt{\pi_{i,j}q_{i,j}(\mathbf b)}
\\
&\quad\times
\ket{i,j}_{IJ}
\ket{\mathrm{par}_{i,j}}_P
\ket{\mathbf b}_F
\ket{\Phi_{i,j,\mathbf b}}
\\
&\quad\otimes
\ket{r(\mathbf b)}_{R_{\mathrm{pref}}}
\ket{\widetilde H(\mathbf b)}_H
\ket{x_j(\widetilde H(\mathbf b))}_X.
\end{aligned}
\]
Collecting all registers other than \(X\) into
\(\ket{\mathrm{work}_{i,j,\mathbf b}}\) gives the state in
the statement. These workspace states retain the branch index
and all flags, so distinct \((i,j,\mathbf b)\) give orthogonal
states. Measuring only the value register therefore produces
\(X=L\widehat Y_{I,J}\), with \((I,J)\sim\pi\) and conditional
flag distribution \(q_{i,j}\).
No measurement is performed during the construction.

\medskip

The circuit uses one application each of \(U_\pi\) and
\(U_{\mathrm{load}}\), \(M\) applications of
\(U_{\mathrm{branch}}\), and reversible arithmetic only.
Its adjoint is obtained by applying the adjoints of all operations
in reverse order, and uses \(M\) applications of
\(U_{\mathrm{branch}}^\dagger\).
The arithmetic computing \(r_m\), \(\widetilde H\), and \(X\)
requires no additional TFD-preparation or block-encoding queries.
\end{proof}
\section{Adaptive cooling schedules}
\label{app:adaptive_schedule}

This appendix proves the schedule-length bounds used in
Section~\ref{sec:adaptive-cooling-schedules}. We first recall the convexity
properties of the log-partition function and then prove the bound on
\(L_{\mathrm{PB}}\), the length of a perfectly-balanced schedule, for
finite endpoints and in the low-temperature specialization. We also give
a simple family showing that the low-temperature bound can be
asymptotically tight when the smallest excitation energy is exponentially
small. The arguments are based on the convexity analysis of
\cite{Stefankovic2009}.

\begin{lem}[\cite{Stefankovic2009}]
\label{lem:convexity-facts}
Let \(H\succeq0\) have eigenvalues
\(E_0,\ldots,E_{N-1}\ge0\), and let
\(f(\beta)=\ln Z(\beta)\), where
\[
Z(\beta)=\sum_{i=0}^{N-1}e^{-\beta E_i}
\]
is the partition function, with eigenvalues counted with multiplicity.
Let \(Y_\beta\) denote the random energy under the Gibbs distribution:
\[
\Pr(Y_\beta=E)
=
\frac{\sum_{i:E_i=E}e^{-\beta E_i}}{Z(\beta)},
\qquad E\in\operatorname{spec}(H).
\]
Then:
\begin{enumerate}
    \item \(f(\beta)\) is nonincreasing, with
    \[
    f'(\beta)=-\mathbb{E}[Y_\beta]\le0.
    \]

    \item \(f(\beta)\) is convex, with
    \[
    f''(\beta)=\operatorname{Var}(Y_\beta)\ge0.
    \]

    \item At \(\beta=0\),
    \[
    f'(0)=-\mathbb{E}[Y_0]\ge-\|H\|_\infty.
    \]
\end{enumerate}
\end{lem}

\medskip

We now prove the perfectly-balanced schedule-length bound used in the
main text. The argument is based on the convexity construction of
\cite{Stefankovic2009}, but we use it directly for perfectly-balanced
schedules and do not perform the additional Chebyshev refinement.
The related bound in Harrow and Wei~\cite{Harrow2020} appears to contain
minor sign mistakes as written, so we give an alternative proof here for completeness.

\begin{proof}[Proof of Theorem~\ref{thm:adaptive_schedule_existence}]
The finite-endpoint statement follows from
\cite[Lemma~4.3]{Stefankovic2009}. Applied to
\(f(\beta)=\ln Z(\beta)\), that lemma gives a sequence
\[
0=\gamma_0<\gamma_1<\cdots<\gamma_{L_{\mathrm{PB}}}=\gamma
\]
satisfying
\[
f\!\left(\frac{\gamma_i+\gamma_{i+1}}{2}\right)
\ge
\frac{f(\gamma_i)+f(\gamma_{i+1})}{2}-1
\]
for every interval and
\[
L_{\mathrm{PB}}
\le
1+\sqrt{
\bigl(f(0)-f(\gamma)\bigr)
\ln\frac{-f'(0)}{-f'(\gamma)}
}.
\]
In the maximal-endpoint construction used in the proof of that lemma,
every nonfinal interval saturates the inequality. Hence every nonfinal
interval is perfectly balanced in our terminology, while the final
interval may satisfy the inequality strictly if it terminates at the
prescribed endpoint \(\gamma\). The stated \(O(\cdot)\) bound follows
immediately when \(f(0)-f(\gamma)=O(n)\).

\medskip

We now prove the low-temperature specialization. For
\(\bar H=H-E_0I\), write
\[
\bar Z(\beta):=\Tr(e^{-\beta\bar H}),
\qquad
f_{\bar H}(\beta):=\ln\bar Z(\beta).
\]
The shifted Hamiltonian has ground energy zero, ground-state degeneracy
\(N_0\), and smallest positive excitation \(E_s\). Therefore
\[
f_{\bar H}(0)=n\ln2,
\qquad
\lim_{\beta\to\infty}f_{\bar H}(\beta)=\ln N_0.
\]
Since \(\ln N_0+2<n\ln2\), continuity implies that there exists a finite
\(\gamma>0\) such that
\[
f_{\bar H}(\gamma)=\ln N_0+2.
\]
Thus
\[
f_{\bar H}(0)-f_{\bar H}(\gamma)
=
n\ln2-\ln N_0-2,
\qquad
\bar Z(\gamma)=N_0e^2.
\]

\medskip

The excited states contribute
\[
\bar Z(\gamma)-N_0=N_0(e^2-1).
\]
Since every positive eigenvalue of \(\bar H\) is at least \(E_s\),
\[
-f_{\bar H}'(\gamma)
=
\frac{
\sum_i(E_i-E_0)e^{-\gamma(E_i-E_0)}
}{
\bar Z(\gamma)
}
\ge
\frac{E_s(e^2-1)}{e^2}.
\]
At \(\beta=0\),
\[
-f_{\bar H}'(0)\le\|\bar H\|_\infty.
\]
Substituting these estimates into the finite-endpoint bound gives
\[
L_{\mathrm{PB}}
\le
1+\sqrt{
\bigl(n\ln2-\ln N_0-2\bigr)
\ln\!\left(
\frac{e^2\|\bar H\|_\infty}{(e^2-1)E_s}
\right)
}.
\]
Since \(n\ln2-\ln N_0-2\le n\ln2\), the stated asymptotic bound follows.

\medskip

Finally,
\[
\bar Z(\infty)=N_0,
\qquad
\bar Z(\gamma)=N_0e^2.
\]
Hence the squared TFD overlap between inverse temperatures
\(\gamma\) and \(\infty\) is
\[
O_{\bar H}(\gamma,\infty)
=
\frac{\bar Z(\infty)^2}
     {\bar Z(\gamma)\bar Z(\infty)}
=
e^{-2}.
\]
Thus the terminal interval is \(e^2\)-slowly varying.
\end{proof}

\medskip

We now give a simple family showing that the low-temperature
schedule-length bound can be asymptotically tight when the smallest
excitation energy is exponentially small.

\begin{thm}[A tightness example]
\label{thm:tight}
For every \(n\ge3\), consider the \(n\)-qubit Hamiltonian
\[
H_n:=\sum_{j=1}^{n}2^{-j}\ket1\!\bra1_j,
\]
where \(\ket1\!\bra1_j\) denotes the projector
\(\ket1\!\bra1\) acting on qubit \(j\) and the identity on all other
qubits. Thus, on a computational-basis state
\(\ket{x_1\cdots x_n}\),
\[
H_n\ket{x_1\cdots x_n}
=
\left(\sum_{j=1}^{n}x_j2^{-j}\right)
\ket{x_1\cdots x_n}.
\]
Let \(f_n(\beta):=\ln Z_n(\beta)\), and let \(\gamma_n>0\) be the unique
inverse temperature satisfying
\[
f_n(\gamma_n)=2.
\]
Then every \(e^2\)-slowly-varying schedule from \(0\) to \(\gamma_n\)
contains \(\Omega(n)\) intervals. Conversely, there exists a
perfectly-balanced schedule from \(0\) to \(\gamma_n\) with \(O(n)\)
intervals. Hence the minimum number of intervals is \(\Theta(n)\).

\medskip

Moreover,
\[
E_s=2^{-n},
\qquad
\|H_n\|_\infty=1-2^{-n},
\]
so that
\[
\sqrt{
n\ln\frac{\|H_n\|_\infty}{E_s}
}
=
\Theta(n).
\]
Thus the low-temperature schedule-length bound is asymptotically tight
in this regime.
\end{thm}

\begin{proof}
Set \(N:=2^n\). By binary expansion, the spectrum of \(H_n\) is
\[
\left\{
0,\frac1N,\frac2N,\ldots,\frac{N-1}{N}
\right\},
\]
with each eigenvalue appearing once. Therefore
\(E_s=1/N=2^{-n}\) and
\(\|H_n\|_\infty=(N-1)/N\). The partition function is
\[
Z_n(\beta)
=
\sum_{k=0}^{N-1}e^{-\beta k/N}
=
\frac{1-e^{-\beta}}{1-e^{-\beta/N}},
\qquad \beta>0,
\]
with \(Z_n(0)=N\). Since \(N>e^2\) for \(n\ge3\), while
\(Z_n(\beta)\) decreases continuously from \(N\) to \(1\), there is a
unique finite \(\gamma_n>0\) such that \(Z_n(\gamma_n)=e^2\).

\medskip

We first prove the lower bound. For every \(t\ge0\),
\[
\frac{Z_n(t/2)}{Z_n(t)}
=
\frac{1+e^{-t/(2N)}}{1+e^{-t/2}}
\le2,
\]
where the identity for \(t=0\) is understood by continuity. Consider any
interval \(0\le a<b\). Since \(Z_n\) is decreasing and
\((a+b)/2\ge b/2\),
\[
Z_n\!\left(\frac{a+b}{2}\right)
\le
Z_n(b/2)
\le
2Z_n(b).
\]
Hence its squared TFD overlap satisfies
\[
O(a,b-a)
=
\frac{Z_n((a+b)/2)^2}{Z_n(a)Z_n(b)}
\le
4\,\frac{Z_n(b)}{Z_n(a)}.
\]
If the interval is \(e^2\)-slowly varying, then
\(O(a,b-a)\ge e^{-2}\). Therefore
\[
f_n(a)-f_n(b)
=
\ln\frac{Z_n(a)}{Z_n(b)}
\le
2+\ln4.
\]

\medskip

Now let
\[
0=\beta_0<\beta_1<\cdots<\beta_\ell=\gamma_n
\]
be any \(e^2\)-slowly-varying schedule. Summing the preceding bound over
its intervals gives
\[
n\ln2-2
=
f_n(0)-f_n(\gamma_n)
\le
\ell(2+\ln4).
\]
Thus
\[
\ell
\ge
\frac{n\ln2-2}{2+\ln4}
=
\Omega(n).
\]

\medskip

Every interval of a perfectly-balanced schedule has squared TFD
overlap at least \(e^{-2}\), with equality on every nonfinal interval.
Thus every perfectly-balanced schedule is \(e^2\)-slowly varying,
and the lower bound above applies to it. Hence
\[
L_{\mathrm{PB}}
\ge
\frac{n\ln2-2}{2+\ln4}
=
\Omega(n).
\]

\medskip

A matching upper bound follows directly from
Theorem~\ref{thm:existence}. For \(H_n\), the ground state is unique,
\(E_s=2^{-n}=1/N\), and
\(\|H_n\|_\infty=1-2^{-n}=(N-1)/N\). Moreover,
\(f_n(\gamma_n)=2\), so the low-temperature specialization of
Theorem~\ref{thm:existence} gives
\[
L_{\mathrm{PB}}
\le
1+\sqrt{
(n\ln2-2)
\ln\!\left(
\frac{(N-1)e^2}{e^2-1}
\right)
}
=
O(n),
\]
where \(N=2^n\). Combining the two bounds gives
\[
L_{\mathrm{PB}}=\Theta(n).
\]

\medskip

Finally,
\[
\ln\frac{\|H_n\|_\infty}{E_s}
=
\ln(N-1)
=
\Theta(n),
\]
and therefore
\[
\sqrt{
n\ln\frac{\|H_n\|_\infty}{E_s}
}
=
\Theta(n).
\]
Thus this family saturates the low-temperature schedule-length bound up
to constant factors.
\end{proof}

\section{Approximate preparation}
\label{sec:approx-prep}

In our algorithm, the logarithmic terms \(A(\beta_i,2^j\delta_i)\) and
\(B(\beta_i,\delta_i)\) from Section~\ref{sec:technical-overview} are evaluated using
the canonical thermofield-double purifications of Gibbs states.
In practice, these purifications can only be prepared approximately.
It is therefore necessary to quantify how state-preparation errors affect
the local estimators used by the algorithm. In this appendix, we analyze
their effect on the logarithmic terms and on the final reconstruction.

\medskip

Throughout this section, \(\|\cdot\|_2\) denotes the Euclidean norm on
state vectors, and \(\|\cdot\|_1\) denotes the trace norm on operators.
As in the main text, assume \(0\preceq H\preceq hI\), with \(h>0\), and
\(\delta_i\geq0\). For each inverse temperature \(\beta\), let
\(|\psi_\beta\rangle\) denote the canonical thermofield-double
purification used in the reconstruction. The purification convention
is fixed across temperatures.
Let \(|\tilde\psi_\beta\rangle\) denote the normalized state actually
prepared by the algorithm. We assume that, for every inverse temperature
queried by the algorithm,
\[
\bigl\|\,|\tilde\psi_\beta\rangle-|\psi_\beta\rangle\,\bigr\|_2
\leq\eta.
\]
This assumption concerns approximation to the prescribed canonical
purification; closeness of the reduced Gibbs states alone is not sufficient.
Each approximate preparation is implemented by a fixed unitary
\(\widetilde U_\beta\) satisfying
\(\widetilde U_\beta|0\rangle=|\tilde\psi_\beta\rangle\), and its adjoint
is used consistently in the estimation circuits. Thus the preparation
error changes the local probabilities being estimated, while the
estimation procedure still uses a fixed coherent preparation.

\medskip

When comparing reduced states, we write
\[
\rho_\beta
=
\Tr_{\mathrm{anc}}\!\bigl(
|\psi_\beta\rangle\langle\psi_\beta|
\bigr),
\qquad
\tilde\rho_\beta
=
\Tr_{\mathrm{anc}}\!\bigl(
|\tilde\psi_\beta\rangle\langle\tilde\psi_\beta|
\bigr).
\]
We first bound the perturbation of the overlap-based quantity
\(A(\beta_i,2^j\delta_i)\), assuming that the ideal squared overlap is
bounded away from zero. We then prove an analogous bound for
\(B(\beta_i,\delta_i)\), using contractivity of partial trace and the
trace-norm bound for expectation values.

\begin{thm}[Stability under approximate purification]
Assume the preparation model above. Suppose that there exist
\(\gamma_A,\gamma_B>0\) such that, for every queried overlap pair
\((i,j)\) and every queried reweighting term \(i\),
\[
\left|
\left\langle
\psi_{\beta_i}\middle|\psi_{\beta_i+2^j\delta_i}
\right\rangle
\right|^2
\geq\gamma_A,
\qquad
\Tr(\rho_{\beta_i}e^{-\delta_iH})\geq\gamma_B.
\]
If
\[
\eta\leq\frac18\min\{\gamma_A,\gamma_B\},
\]
then the quantities obtained by replacing the ideal purifications by
the prepared purifications satisfy
\[
\bigl|
\widetilde A(\beta_i,2^j\delta_i)
-
A(\beta_i,2^j\delta_i)
\bigr|
\leq\frac{8\eta}{\gamma_A},
\qquad
\bigl|
\widetilde B(\beta_i,\delta_i)
-
B(\beta_i,\delta_i)
\bigr|
\leq\frac{4\eta}{\gamma_B}.
\]
\end{thm}

\begin{proof}
We first consider
\[
A(\beta_i,2^j\delta_i)
=
-\ln
\bigl|
\langle\psi_{\beta_i}\mid\psi_{\beta_i+2^j\delta_i}\rangle
\bigr|^2,
\]
and its approximate counterpart
\[
\widetilde A(\beta_i,2^j\delta_i)
=
-\ln
\bigl|
\langle\tilde\psi_{\beta_i}\mid
\tilde\psi_{\beta_i+2^j\delta_i}\rangle
\bigr|^2.
\]
By normalization and the triangle inequality,
\[
\bigl|
\langle\tilde\psi_{\beta_i}\mid
\tilde\psi_{\beta_i+2^j\delta_i}\rangle
-
\langle\psi_{\beta_i}\mid
\psi_{\beta_i+2^j\delta_i}\rangle
\bigr|
\leq2\eta.
\]
Since both overlap amplitudes have modulus at most one, their squared
moduli differ by at most \(4\eta\). The ideal squared overlap is at
least \(\gamma_A\), so the approximate squared overlap is at least
\(\gamma_A-4\eta\geq\gamma_A/2\).
The logarithm is \(2/\gamma_A\)-Lipschitz on
\([\gamma_A/2,1]\). Consequently,
\[
\bigl|
\widetilde A(\beta_i,2^j\delta_i)
-
A(\beta_i,2^j\delta_i)
\bigr|
\leq\frac{8\eta}{\gamma_A}.
\]

\medskip

We next consider
\[
B(\beta_i,\delta_i)
=
\ln\Tr(\rho_{\beta_i}e^{-\delta_iH}),
\qquad
\widetilde B(\beta_i,\delta_i)
=
\ln\Tr(\tilde\rho_{\beta_i}e^{-\delta_iH}).
\]
For normalized pure states, the state-vector error implies
\[
\left\|
|\tilde\psi_{\beta_i}\rangle
\langle\tilde\psi_{\beta_i}|
-
|\psi_{\beta_i}\rangle
\langle\psi_{\beta_i}|
\right\|_1
\leq
2\bigl\|
|\tilde\psi_{\beta_i}\rangle-|\psi_{\beta_i}\rangle
\bigr\|_2
\leq2\eta.
\]
Contractivity of partial trace therefore gives
\[
\|\tilde\rho_{\beta_i}-\rho_{\beta_i}\|_1\leq2\eta.
\]
Since \(H\succeq0\) and \(\delta_i\geq0\),
\(\|e^{-\delta_iH}\|_\infty\leq1\), and hence
\[
\bigl|
\Tr(\tilde\rho_{\beta_i}e^{-\delta_iH})
-
\Tr(\rho_{\beta_i}e^{-\delta_iH})
\bigr|
\leq
\|\tilde\rho_{\beta_i}-\rho_{\beta_i}\|_1
\|e^{-\delta_iH}\|_\infty
\leq2\eta.
\]
The approximate expectation is thus at least
\(\gamma_B-2\eta\geq\gamma_B/2\).
Applying the \(2/\gamma_B\)-Lipschitz bound for the logarithm yields
\[
\bigl|
\widetilde B(\beta_i,\delta_i)
-
B(\beta_i,\delta_i)
\bigr|
\leq\frac{4\eta}{\gamma_B}.
\]
\end{proof}

\medskip

It follows that a uniform additive perturbation of
\(O(\tilde\varepsilon)\) in each logarithmic term is ensured
by
\[
\eta
=
O\!\left(
\tilde\varepsilon\min\{\gamma_A,\gamma_B\}
\right),
\]
with a sufficiently small implied constant.

\medskip

The statistical analysis in Section~\ref{sec:single-temperature} uses
variable-depth variance-sensitive accumulation. Preparation error,
however, produces a deterministic bias in the implemented logarithmic
terms, which we control using the sum of the absolute
reconstruction weights. To make this accumulation explicit, let
\(f_{\mathrm{prep}}(\beta)\) denote the deterministic reconstruction
obtained by replacing \(A\) and \(B\) by \(\widetilde A\) and
\(\widetilde B\), respectively. For the schedule
\(0=\beta_0<\cdots<\beta_{l_{\mathrm{SV}}}=\beta\), the reconstruction
in Section~\ref{sec:technical-overview} and the theorem above give
\[
\begin{aligned}
|f_{\mathrm{prep}}(\beta)-f(\beta)|
&\leq
\sum_{i=0}^{l_{\mathrm{SV}}-1}
\left(
2^{k_i}\frac{4\eta}{\gamma_B}
+
\sum_{j=1}^{k_i}
2^{k_i-j}\frac{8\eta}{\gamma_A}
\right)\\
&\leq
\frac{8\eta}{\min\{\gamma_A,\gamma_B\}}
\sum_{i=0}^{l_{\mathrm{SV}}-1}
\bigl(2^{k_i+1}-1\bigr).
\end{aligned}
\]
Here the depth choice
\[
k_i
=
\max\!\left\{
0,\left\lceil
\log_2\!\bigl(h(\beta_{i+1}-\beta_i)\bigr)
\right\rceil
\right\}
\]
implies \(2^{k_i}\leq1+2h(\beta_{i+1}-\beta_i)\). Therefore,
\[
\sum_{i=0}^{l_{\mathrm{SV}}-1}
\bigl(2^{k_i+1}-1\bigr)
\leq l_{\mathrm{SV}}+4h\beta.
\]
Thus it suffices to take
\[
\tilde\varepsilon
=
\Theta\!\left(
\frac{\varepsilon}{l_{\mathrm{SV}}+h\beta}
\right),
\qquad
\eta
=
O\!\left(
\frac{\varepsilon\min\{\gamma_A,\gamma_B\}}
{l_{\mathrm{SV}}+h\beta}
\right),
\]
with constants chosen so that preparation bias uses only a fixed
fraction of the total error budget.

\medskip

The approximate local probabilities remain bounded below by constant
fractions of their ideal lower bounds, so the statistical analysis
applies with correspondingly adjusted conditioning constants.
Any operator-approximation error in the reweighting circuit is
budgeted separately and combined with preparation error by the
triangle inequality. With these error budgets, the final additive
accuracy for \(f(\beta)=\ln Z(\beta)\), and hence the corresponding
multiplicative accuracy for \(Z(\beta)\), is preserved.
\section{Ground-state degeneracy from finite-temperature estimates}
\label{app:ground-state-degeneracy}
\label{sec:ground-state-degeneracy}

This appendix gives the detailed reduction from finite-temperature
partition-function estimation to ground-state degeneracy estimation.
The reduction never requires evaluating the partition function at
\(\beta=\infty\). Instead, a sufficiently large finite inverse temperature
suppresses the excited-state contribution. The required inverse temperature
is controlled by a lower bound on the spectral gap above the ground space,
and the reduction additionally requires either the exact ground energy or
an estimate of it.

\begin{thm}[Ground-state degeneracy from finite-temperature estimates]
\label{thm:ground-state-degeneracy}
Let \(H\) be an \(n\)-qubit Hamiltonian satisfying
\[
0\preceq H\preceq hI,
\]
with ground-state energy \(E_0\) and ground-state degeneracy \(N_0\).
Write
\[
Z(\beta):=\operatorname{Tr}(e^{-\beta H}),
\qquad
f(\beta):=\ln Z(\beta).
\]
Assume that a lower bound \(\Delta_{\mathrm{lb}}>0\) on the gap above
the ground space is known, so that
\[
E_j-E_0\ge\Delta_{\mathrm{lb}}
\]
for every eigenvalue \(E_j>E_0\).
Fix \(\epsilon\in(0,1]\), and choose a finite inverse temperature
satisfying
\[
\beta
\ge
\frac{n\ln2+\ln(6/\epsilon)}{\Delta_{\mathrm{lb}}}.
\]
Suppose that estimates \(\widehat f(\beta)\) and \(\widehat E_0\)
satisfy
\[
|\widehat f(\beta)-f(\beta)|
\le \frac{\epsilon}{6},
\qquad
|\widehat E_0-E_0|
\le \frac{\epsilon}{6\beta}.
\]
Define
\[
\widehat N_0
:=
\exp\!\left(\widehat f(\beta)+\beta\widehat E_0\right).
\]
Then
\[
1-\epsilon
\le
\frac{\widehat N_0}{N_0}
\le
1+\epsilon.
\]
\end{thm}

\begin{proof}
Writing the eigenvalues with multiplicity, we have
\[
Z(\beta)
=
N_0e^{-\beta E_0}
+
\sum_{j:E_j>E_0}e^{-\beta E_j}.
\]
After factoring out the ground-state contribution,
\[
e^{\beta E_0}Z(\beta)
=
N_0+T_\beta,
\qquad
T_\beta
:=
\sum_{j:E_j>E_0}e^{-\beta(E_j-E_0)}.
\]
By the assumed gap lower bound,
\[
T_\beta
\le
2^n e^{-\beta\Delta_{\mathrm{lb}}}.
\]
The chosen value of \(\beta\) therefore implies
\[
T_\beta\le\frac{\epsilon}{6}.
\]
Since \(N_0\ge1\),
\[
\frac{T_\beta}{N_0}
\le
\frac{\epsilon}{6}.
\]

\medskip

Write the estimation errors as
\[
\widehat f(\beta)=f(\beta)+\eta_f,
\qquad
\widehat E_0=E_0+\eta_E.
\]
Then
\[
\widehat N_0
=
\exp\!\left(
f(\beta)+\beta E_0+\eta_f+\beta\eta_E
\right)
=
e^{\eta_f+\beta\eta_E}(N_0+T_\beta),
\]
and hence
\[
\frac{\widehat N_0}{N_0}
=
e^{\eta_f+\beta\eta_E}
\left(1+\frac{T_\beta}{N_0}\right).
\]
The assumed estimation guarantees give
\[
|\eta_f+\beta\eta_E|
\le
\frac{\epsilon}{3}.
\]
Together with \(T_\beta/N_0\le\epsilon/6\), this implies
\[
\frac{\widehat N_0}{N_0}
\le
e^{\epsilon/3}
\left(1+\frac{\epsilon}{6}\right).
\]
For \(0<\epsilon\le1\), the inequality
\(e^x\le1+2x\) for \(0\le x\le1/3\) gives
\[
\begin{aligned}
\frac{\widehat N_0}{N_0}
&\le
\left(1+\frac{2\epsilon}{3}\right)
\left(1+\frac{\epsilon}{6}\right)\\
&=
1+\frac{5\epsilon}{6}+\frac{\epsilon^2}{9}
\le
1+\epsilon.
\end{aligned}
\]
Similarly,
\[
\frac{\widehat N_0}{N_0}
\ge
e^{-\epsilon/3}
\ge
1-\epsilon.
\]
This proves the multiplicative error guarantee.
\end{proof}

\medskip

The theorem specifies the temperature and estimation accuracies needed
to recover the ground-state degeneracy. We now apply the two
partition-function reconstruction theorems at these parameters to obtain
the query bounds. Any additional cost of estimating the ground energy
is accounted for separately.

\begin{cor}[Cost of ground-state degeneracy estimation]
\label{cor:ground-state-degeneracy-cost}
Under the assumptions of
Theorem~\ref{thm:ground-state-degeneracy}, fix \(0<p<1\) and let
\[
\beta
=
\frac{n\ln2+\ln(6/\epsilon)}{\Delta_{\mathrm{lb}}}.
\]
Suppose that a supplied \(C\)-slowly-varying schedule of length
\(\ell_{\mathrm{SV}}(\beta)\) reaches this inverse temperature,
where \(C=O(1)\).
For each reconstruction method below, assume the block-encoding
normalization and state-preparation access required by the
corresponding partition-function estimation theorem.

\medskip

Using the fixed-temperature reconstruction of
Theorem~\ref{thm:single-temperature-main}
with failure probability \(p/2\), the expected reconstruction-stage
query counts for the required partition-function estimate satisfy
\[
N_{\mathrm{TFD}},N_{\mathrm{BE}}
=
\widetilde{\mathcal O}\!\left(
\frac{
\sqrt{\ell_{\mathrm{SV}}(\beta)}
\bigl(\ell_{\mathrm{SV}}(\beta)+h\beta\bigr)
\ln(4/p)}{\epsilon}
\right).
\]
Alternatively, using the coherent-temperature reconstruction of
Theorem~\ref{thm:coherent-indexed-main}
with failure probability \(p/2\), the expected reconstruction-stage
query counts for the same partition-function estimate satisfy
\[
N_{\mathrm{cohTFD}},N_{\mathrm{BE}}
=
\widetilde{\mathcal O}\!\left(
\frac{
\bigl(\ell_{\mathrm{SV}}(\beta)+h\beta\bigr)
\ln(4/p)}{\epsilon}
\right).
\]

\medskip

If \(E_0\) is not known exactly, an additional estimate
\(\widehat E_0\) is required with
\[
|\widehat E_0-E_0|
\le
\frac{\epsilon}{6\beta}
=
\frac{\epsilon\Delta_{\mathrm{lb}}}
{6\bigl(n\ln2+\ln(6/\epsilon)\bigr)}.
\]
If the partition-function and ground-energy estimates each fail
with probability at most \(p/2\), then the resulting
\(\widehat N_0\) has relative error at most \(\epsilon\) with
probability at least \(1-p\).
The displayed query bounds exclude schedule-generation costs
and any additional cost of estimating \(E_0\).
\end{cor}

\begin{proof}
Run Theorem~\ref{thm:single-temperature-main}, or alternatively
Theorem~\ref{thm:coherent-indexed-main}, with additive accuracy
\(\epsilon/6\) and failure probability \(p/2\).
Constant rescaling of the accuracy is absorbed by
\(\widetilde{\mathcal O}\).
The required ground-energy accuracy is exactly the second
estimation guarantee in
Theorem~\ref{thm:ground-state-degeneracy}.
A union bound shows that both estimation guarantees hold with
probability at least \(1-p\), and that theorem then gives the
claimed relative error.
\end{proof}

\medskip

The gap lower bound determines the inverse temperature needed to suppress
excited states. This choice of \(\beta\) also sets the required precision
of the ground-energy estimate: a larger inverse temperature requires
a more accurate estimate of the ground energy.

\section{Comparison with prior work}

Table~\ref{tab:prior-work} compares guarantees and resource units
under the stated access models. The system-size bounds for this work
use \(h=O(n)\), \(\beta=\Theta(1)\), and
\(\ell=\widetilde O(\sqrt n)\).

\begin{table}[!htbp]
\centering
\small
\setlength{\tabcolsep}{4pt}
\begin{tabular}{@{}L{0.18\textwidth}L{0.30\textwidth}L{0.30\textwidth}L{0.16\textwidth}@{}}
\toprule
Work & Setting and access & Main guarantee & Resource \\
\midrule
\multicolumn{4}{@{}l}{\textit{Classical partition-function estimators}} \\
\addlinespace[2pt]
\v{S}tefankovi\v{c} et al.~\cite{Stefankovic2009}
& Classical energies and Gibbs samples; supplied Chebyshev schedule
& \(O(\ell^2/\epsilon^2)\) samples for relative partition-function error \(\epsilon\)
& Gibbs-sampling calls \\
\addlinespace
Huber; Kolmogorov~\cite{Huber2015,Kolmogorov2017}
& Classical energies and Gibbs samples
& Adaptive schedules and paired-product estimators; \(\epsilon^{-2}\) estimation dependence
& Gibbs-sampling calls \\
\midrule
\multicolumn{4}{@{}l}{\textit{Classical methods for quantum thermal systems}} \\
\addlinespace[2pt]
Harrow et al.~\cite{Harrow_2020}
& Local quantum Hamiltonians; suitable complex zero-free region around \([0,\beta]\)
& Relative partition-function error \(\epsilon\) in \(n^{O(\log(n/\epsilon))}\) time at fixed \(\beta\)
& Classical runtime \\
\addlinespace
Bakshi et al.~\cite{BakshiLiuMoitraTang2024}
& Bounded-degree local qubit Hamiltonians at sufficiently high temperature
& Sample product states whose mixture has Gibbs-state trace-norm error \(\eta\), in \(\operatorname{poly}(n,\log(1/\eta))\) time
& Classical runtime for mixed-state preparation \\
\addlinespace
Huang~\cite{Huang2021}
& One-dimensional bounded nearest-neighbor interactions
& Matrix-product operator approximation for local-observable error \(\eta\); bond dimension independent of \(n\) at fixed \(\beta,\eta\)
& Representation bond dimension \\
\midrule
\multicolumn{4}{@{}l}{\textit{Quantum partition-function estimators with classical-energy access}} \\
\addlinespace[2pt]
Montanaro~\cite{Montanaro2015}
& Coherent energy evaluation, Gibbs states, and their reflections
& \(\widetilde O(\ell^2/\epsilon)\); bounded-variance quantum mean estimation of ratios
& Gibbs-state reflections \\
\addlinespace
Harrow--Wei~\cite{Harrow2020}
& Coherent energy evaluation and Gibbs-state reflections
& \(\widetilde O(\ell^2/\epsilon)\) estimation; adaptive annealing with state reuse
& Gibbs-state reflections \\
\addlinespace
Arunachalam et al.~\cite{Arunachalam2022}
& Coherent energy evaluation and Gibbs-state reflections
& \(\widetilde O(\ell^2/\epsilon)\) estimation; overlap-based scheduling and paired products
& Gibbs-state reflections \\
\addlinespace
Cornelissen--Hamoudi~\cite{Cornelissen2023}
& Coherent values, Gibbs states, and their reflections; each factor has at most \(d\) values
& \(\widetilde O(\ell^{3/2}/\epsilon+\ell\log d)\); nearly unbiased product estimation
& Expected Gibbs-state reflections \\
\midrule
\multicolumn{4}{@{}l}{\textit{This work: real spectra with TFD and block-encoding access}} \\
\addlinespace[2pt]
Fixed temperature
& TFD preparation at classically chosen temperatures; Hamiltonian block encoding
& \(\widetilde O\!\left(\frac{\sqrt\ell(\ell+h\beta)}{\epsilon}\right)\); \(\widetilde O(n^{5/4}/\epsilon)\) in the regime above
& Expected TFD and block-encoding queries, each \\
\addlinespace
Coherent temperature
& Coherent-temperature interface with branch and parameter loading; block encoding
& \(\widetilde O\!\left(\frac{\ell+h\beta}{\epsilon}\right)\); \(\widetilde O(n/\epsilon)\) in the regime above
& Expected coherent TFD and block-encoding queries, each \\
\bottomrule
\end{tabular}
\caption{Representative results under the stated access models and resource units.
Costs involving \(\ell\) assume a supplied suitable schedule of length \(\ell\),
constant conditioning, and constant failure probability; schedule-generation
costs are separate.}
\label{tab:prior-work}
\end{table}

\clearpage

For partition-function estimation, relative error and additive error in
\(\ln Z\) are equivalent up to constant rescaling for
\(0<\epsilon\leq1\).

\medskip

Remark~\ref{rem:uniform-grid-comparison} derives the uniform-grid benchmarks
for our access models and states the additional temperature-access
condition for their coherent-temperature specialization.

\section{Additional notes on the common reconstruction}

This appendix records the energy-shift identities and interpretations
of the overlap and reweighting terms used in
Section~\ref{sec:technical-overview}.

\subsection{Energy normalization}

For the normalization \(H=H_{\mathrm{orig}}-LI\) introduced in
Section~\ref{sec:technical-overview}, the TFD overlaps are unchanged,
while the reweighting terms transform as
\[
R_H(\beta,s)
=
e^{sL}R_{H_{\mathrm{orig}}}(\beta,s),
\qquad
B_H(\beta,s)
=
B_{H_{\mathrm{orig}}}(\beta,s)+sL.
\]
The stronger information needed to extract the ground-state
degeneracy is discussed separately in
Section~\ref{subsec:ground-degeneracy}.

\subsection{Thermodynamic interpretation}

The logarithmic terms \(A(\beta,s)\) and \(B(\beta,s)\) are
second and first finite differences of \(f\), respectively, as
specified in Definition~\ref{def:local-thermal-quantities}.
Their small-step limits give the usual thermodynamic derivatives.

\begin{rem}[Thermodynamic derivatives]
\label{rem:thermodynamic-derivatives}
Since the system is finite dimensional, \(f\) is smooth at every
finite inverse temperature.
The same two quantities recover the usual thermodynamic
derivatives:
\[
    \lim_{s\rightarrow0}
    \frac{B(\beta,s)}{s}
    =
    f'(\beta)
    =
    -\Tr(H\rho_\beta),
\]
and
\[
    \lim_{s\rightarrow0}
    \frac{A(\beta,s)}{(s/2)^2}
    =
    f''(\beta)
    =
    \Var_{\rho_\beta}(H).
\]
Hence \(B\) contains first-order information about \(f\), while
the TFD overlap contains the corresponding second-order
information.
\end{rem}

\medskip

The recursive-doubling identity uses this finite-difference structure
to recover a full schedule increment from a short reweighting step.
The following interpretation explains the role of its overlap terms.

\begin{rem}[Interpretation of the recursive doubling]
\label{rem:recursive-doubling-interpretation}
The reweighting term in
Eq.~\eqref{eq:interval-dyadic-reconstruction} may be written as
\[
    2^{k_i}B(\beta_i,\delta_i)
    =
    \Delta_i
    \frac{B(\beta_i,\delta_i)}{\delta_i}.
\]
This multiplies the change in \(f\) over the first short interval
by \(2^{k_i}\). It would give the full change
\(f(\beta_{i+1})-f(\beta_i)\) if \(f\) were linear.

In general, doubling the interval length does not double the
change in \(f\). When the length increases from
\(2^{j-1}\delta_i\) to \(2^j\delta_i\), adding
\(A(\beta_i,2^j\delta_i)\) to twice the previous change gives
the new change exactly.
Repeating this step gives the weighted \(A\)-terms in the
reconstruction formula and recovers
\(f(\beta_{i+1})-f(\beta_i)\) exactly.
\end{rem}

\section{Circuit constructions for spin chains}
\label{app:coherent-temperature-circuits}

This appendix gives the circuit constructions for the one-dimensional
application. We prove Proposition~\ref{prop:coherent-tfd-spin-chains}
and account for the auxiliary operations and implementation errors
used in Corollary~\ref{cor:spin-chain-circuit-comparison}.
We use the circuit architecture of Section~\ref{sec:end-to-end}.
The preparation accuracy concerns the canonical TFD and all additional
workspace, not only the reduced Gibbs state.

\subsection{Proof of Proposition~\ref{prop:coherent-tfd-spin-chains}}
\label{subsec:coherent-temperature-proof}

\begin{proof}
Group a fixed number of adjacent sites if necessary to use the
nearest-neighbor formulation of the cited preparation construction.
Let \(\mathcal H_t\) be the Hermitian parent Hamiltonian of its
Gaussian-filter Gibbs sampler with single-site Pauli jumps and
local-rate normalization. Choose its sign so that its unique ground
state is \(\ket{\psi_t}\), with energy zero. The gap bound and its
uniformity on bounded temperature intervals follow from
\cite[Theorem~I.1 and Section~VII]{bergamaschi2026fastmixingquantumspin}.
Thus one positive gap bound applies throughout \([0,\beta_{\max}]\), independently
of the chain length. This is a gap of the parent Hamiltonian, not an
assumption about the physical Hamiltonian.

\medskip

\noindent\textbf{Uniform preparation path.}
For endpoint \(\theta\), use the path
\(s\mapsto\mathcal H_{s\theta}\), \(0\le s\le1\).
At \(s=0\), the canonical TFD is a product of Bell pairs.
The quasi-adiabatic construction in
\cite[Theorem~C.1 and Appendix~C]{bergamaschi2026fastmixingquantumspin}
transports this ground state by unitary evolution. Its local
decomposition has a constant number of layers of simulations on
blocks of polylogarithmically many sites. Choose the spatial
cutoffs, integration ranges, discretizations, and error tolerances
using the common gap and locality bounds for \([0,\beta_{\max}]\). These
choices do not depend on the temperature label.

\medskip

\noindent\textbf{Temperature-controlled local operations.}
The temperature-dependent ingredients are block encodings of the
parent Hamiltonian and its derivative. Their controlled versions are
provided by \cite[v2, Lemmas~C.5 and~C.6]{chen2023efficient}.
Reversibly compute the product \(s\theta\) to the required
precision and use it as the temperature input. The path derivative is
\[
\frac{\partial}{\partial s}\mathcal H_{s\theta}
=\theta\mathcal H'_{s\theta},
\]
and the second derivative is
\(\theta^2\mathcal H''_{s\theta}\).
The scalar factors can be encoded using bounds valid for every
\(\theta\le \beta_{\max}\). Multiplication and the additional controls use
polynomially many gates in the precision-bit count; they do not
require a lookup over the set of queried temperatures.

\medskip

We verify the normalization and locality requirements for these
subroutines. On a block of \(r\) sites, there are \(O(r)\)
single-site Pauli jumps. Dividing each jump by the square root of
their number gives
\(\|\sum_a A_a^\dagger A_a\|\le1\), as required by the cited
block-encoding lemmas. Restoring local-rate normalization multiplies
the block-encoding normalization by \(O(r)\), not by \(n\).
For bounded finite-range interactions, the first and second
commutators of a single-site jump with the physical Hamiltonian
contain only boundedly many bounded local terms. The derivative
estimates in \cite[v2, Proposition~C.2]{chen2023efficient} therefore
give polynomial bounds in \(r\) for the required local operators
and their derivatives, uniformly in \(\theta\le \beta_{\max}\).
Their Pauli coefficients and controls are obtained from the supplied
local interactions.

\medskip

The local operators being simulated are the restrictions of the
quasi-adiabatic generator, not independent Gibbs preparations on
small blocks. Use the explicitly computable decomposition in
\cite[Appendix~C.3]{bergamaschi2026fastmixingquantumspin}.
It expresses these restrictions as polynomially many sums and
differences of integrals of local derivatives conjugated by local
parent-Hamiltonian evolutions. The preceding block encodings apply
to those terms and their spatial truncations. The integral
block-encoding and time-dependent simulation construction in
\cite[v2, proof of Proposition~C.1]{chen2023efficient} then applies
with \(\theta\) retained as a control. Its filter uses the gap of
the full preparation path; no gap of a truncated block is assumed.
If a block simulation has target error \(\xi\), its gate count is
\[
\operatorname{poly}\bigl(r,\log(2n/\eta),\log(1/\xi)\bigr).
\]
This conclusion uses the explicit local block encodings, rather than
arbitrary matrix synthesis on \(r\) sites.

\medskip

\noindent\textbf{Parallel execution and workspace.}
Take \(r=\operatorname{polylog}(2n/\eta)\), as required by the
local decomposition. There are \(O(n)\) block simulations in a
constant number of layers. Assign the decomposition error at most
\(\eta/4\), and give each block a full-state simulation error
\(O(\eta/n)\), with constants chosen so their sum is at most
\(\eta/4\). The per-block cost is then polylogarithmic in
\(2n/\eta\). An implementation that initially provides only a
unitary block approximation can be run at squared accuracy: if its
ancilla block approximates a unitary within \(\tau\), the full
output on a zero-ancilla input differs by at most \(\sqrt{2\tau}\).
This converts the block guarantee to a workspace-inclusive error
with only logarithmically increased precision.

\medskip

Distribute the temperature bits to the block workspaces using CNOT
trees. This takes \(O(n\log(2n/\eta))\) gates and
\(O(\log(2n))\) depth. Disjoint blocks in the same layer run in
parallel with separate simulation workspace. Their parameter
registers serve only as controls; reverse the parameter computations
and CNOT trees after use. This gives a size and auxiliary-qubit
count of \(n\operatorname{polylog}(2n/\eta)\), and a depth of
\(\operatorname{polylog}(2n/\eta)\). Locality is used separately
on each temperature branch before adding these controls; we do not
assume that a Hamiltonian coupled to one shared temperature register
is spatially local.

\medskip

\noindent\textbf{Phase and coherent error.}
The exact canonical path satisfies
\(\langle\psi_t|\partial_t\psi_t\rangle=0\).
Since its parent eigenvalue is zero for every \(t\),
\(\langle\psi_t|\mathcal H'_t|\psi_t\rangle=0\).
The quasi-adiabatic generator is an integral of conjugated path
derivatives, so it has zero expectation in the corresponding ground
state. Evolution from the canonical Bell-pair state therefore has
no temperature-dependent phase relative to \(\ket{\psi_\theta}\).
This fixes the phase before any numerical simulation is introduced.
Allocate the remaining error budget to internal arithmetic and
simulation approximations, using uniform tolerances over \([0,\beta_{\max}]\).

\medskip

Let \(e_\theta\) be the full output error on label \(\theta\).
The constructed circuit preserves that label. Hence, for
\(\sum_\theta|c_\theta|^2=1\),
\[
\left\|\sum_\theta c_\theta\ket{\theta}e_\theta\right\|_2^2
=\sum_\theta|c_\theta|^2\|e_\theta\|_2^2
\le\eta^2.
\]
Thus the error does not acquire a factor equal to the number of
temperatures. Reversing the explicit circuit gives its inverse.
For an external control, distribute its bit to parallel workspaces
and control the elementary gates. This changes the costs by at most
polylogarithmic factors and restores the control workspace to zero.
\end{proof}

\medskip

\noindent\textbf{Temperature rounding.}
The proposition refers to the temperature represented by the input
bits. For an unrounded requested temperature, the additional error
is controlled by
\[
\left\|\frac{\partial}{\partial t}\ket{\psi_t}\right\|_2
=\frac12\sqrt{\operatorname{Var}_{\rho_t}(H)}\le\frac h4.
\]
The last inequality follows because the spectrum is contained in an
interval of width at most \(h\). Rounding within \(O(\eta/h)\)
and using a constant fraction of \(\eta\) for preparation gives
an \(\eta\)-accurate output for the requested temperature.
Since \(h=O(n)\) and \(\beta_{\max}\) is fixed, this needs
\(O(\log(2n/\eta))\) temperature bits. Guard bits for internal
products and simulation parameters have the same asymptotic length.

\subsection{Auxiliary reconstruction costs}
\label{subsec:coherent-temperature-auxiliary}

We now account for the quantum operations outside TFD preparation.
The supplied schedule and the compiled numerical data have
polylogarithmically many bits per entry at the precisions specified
below. Classical compilation of these data is not included in the
quantum circuit costs.

\medskip

\noindent\textbf{Hamiltonian block encoding.}
After a known scalar shift, the Hamiltonian has a Pauli expansion
\(H=\sum_a c_aP_a\) with \(O(n)\) terms and
\(\sum_a|c_a|=O(n)\). When this sum is nonzero, prepare its
coefficient state with amplitudes
\(\sqrt{|c_a|/\sum_b|c_b|}\) in a register with exactly one
occupied position. A balanced tree of two-qubit rotations splits
the amplitude according to the total weight in each subtree. It
uses \(O(n)\) gates and logarithmic depth. Each occupied position
controls the corresponding signed Pauli term. The overlap graph of
the nonidentity terms has bounded degree, so a constant number of
groups suffices to apply these controlled terms in parallel. An
identity term only requires its coefficient control. Reversing the
coefficient preparation completes the linear-combination block
encoding. If \(H=0\), the block encoding is trivial. Padding the
normalization by an ancilla rotation gives a chosen
\(\alpha=\Theta(n)\), and a known \(h=\Theta(n)\) can be used
as a spectral upper bound. These operations have
\(\widetilde O(n)\) gates and polylogarithmic depth.

\medskip

\noindent\textbf{Branch preparation and table loading.}
The branch set \(\mathcal R\) of
Eq.~\eqref{eq:coherent-branch-data} has
\(\widetilde O(\ell)\) entries. Prepare its distribution by the
same balanced-tree construction on an occupied-position register.
Compute the binary pair \((i,j)\) by parity trees on the relevant
position bits. Since each valid pair identifies one position,
reversibly decode the pair and clear the position register.
Distribute control bits as needed and uncompute temporary copies.
This implements \(U_\pi\) with
\(\widetilde O(\ell)\) gates and workspace and polylogarithmic
depth, including the index conversion.

\medskip

For \(U_{\mathrm{load}}\), fan out the index bits, compare them
with every table address in parallel, and form each output bit by a
parity tree over the matching entries. Reverse all comparisons and
copies. The table has \(\widetilde O(\ell)\) entries of
polylogarithmic length, giving the same resource bounds. This
implements the operations specified in
Appendix~\ref{app:sampled-dyadic},
Subsection~\ref{subsec:coherent-branch-preparation}, without assuming
a special formula for the schedule. Repeated table values cause
no problem because the branch labels are retained.

\medskip

\noindent\textbf{QSVT, estimation, and reflections.}
The elementary reweighting steps satisfy \(\alpha\delta_i=O(1)\).
Their QSVT circuits therefore have polylogarithmic length at the
required accuracy. Store the compiled phase angles and active-step
flags in the same table. Pad to a common maximum length, making
unused gates identities by their flags. Binary-angle-controlled
rotations require only linearly many elementary controlled rotations
in the angle bit count. The actual block-encoding queries remain
charged to the reconstruction query count.

\medskip

The harmonic sample uses polylogarithmically many flags and
precision bits. Prefix products and fixed-point additions therefore
have polylogarithmic cost, as do the controlled arithmetic operations
in quantum mean estimation. In the fixed-temperature estimator, the
phase-estimation and signal-processing routines of
\cite{Cornelissen2023} add at most polylogarithmic overhead per
reflection: powers of the reflected unitary are implemented by the
corresponding number of calls, not charged as single gates.
Clipping and logarithm evaluation use polynomial-time arithmetic in
the bit length. Weighted summation and the final median cost at most
polylogarithmic overhead per estimate. Their costs are included here,
even when performed classically between quantum calls.

\medskip

A zero-state reflection on the full
\(\widetilde O(n)\)-qubit workspace is implemented by a reversible
AND tree, a phase flip, and uncomputation. It has
\(\widetilde O(n)\) gates and polylogarithmic depth.
Together with Proposition~\ref{prop:coherent-tfd-spin-chains}, these
constructions bound each elementary oracle or reflection call by
\(\widetilde O(n)\) gates and polylogarithmic depth. A full
coherent sample uses polylogarithmically many such calls. For
\(\ell=\widetilde O(\sqrt n)\), the remaining work above fits
within the resource bounds of
Corollary~\ref{cor:spin-chain-circuit-comparison}.

\medskip

\noindent\textbf{Preparation and numerical error budgets.}
Run the statistical reconstruction with a constant fraction of the
final error tolerance. Choose the uniform TFD accuracy as in
Appendix~\ref{sec:approx-prep}, decreasing its constant if necessary
to preserve the success-probability floor. Since
\(L=\sum_{i,j}a_{i,j}=O(\ell+h\beta)\), an accuracy
\(\eta=O(\epsilon/L)\) suffices when \(C=O(1)\).
Use the same approximate preparation and its own inverse every time.
Preparation error then changes the fixed branch probabilities; it
is not a new independent error on every call.

\medskip

Approximate loading means exact reversible loading of fixed rounded
data. Temperature rounding is included in the preceding state-error
budget. Rounding a nonnegative imaginary-time step is controlled by
\[
\|e^{-\delta H/2}-e^{-\delta' H/2}\|
\le\frac h2|\delta-\delta'|.
\]
Choose the step precision and QSVT phase accuracy so that the total
block error meets the dilation tolerance in the reconstruction
algorithm. For a QSVT sequence of polylogarithmic length, summing its
gate errors requires only polylogarithmic angle precision. Fixed-point
sample evaluation meets the tolerance \(\zeta\) in the coherent
algorithm; rounding the accumulated output uses another fixed
fraction of \(\epsilon\).

\medskip

Branch-state approximation has a separate, bounded effect. If the
implemented distribution differs from \(\pi\) in total variation
by at most \(\xi\), then, using
\(|X|\le L(2+H_M+\zeta)\), its contribution to the mean error
is at most
\[
2L(2+H_M+\zeta)\xi.
\]
Taking \(\xi=O(\epsilon/[L(2+H_M+\zeta)])\) therefore suffices.
Uniform branch conditioning and the second-moment bounds remain
valid under this change. All numerical tolerances above require
only polylogarithmic bit lengths. Allocating their constants and
the statistical failure budget gives the claimed additive error and
success probability.

\medskip

\noindent\textbf{Optional compilation to a fixed finite gate set.}
The displayed bounds use the elementary-gate model of
Section~\ref{sec:end-to-end}. They can also be compiled to a standard
finite gate set with polylogarithmic overhead~\cite{NielsenChuang}.
To address the random running time, let \(G_0,D_0\) be uniform expected gate and depth
bounds for the ideal implementation, run with failure probability
at most \(p/2\). Stop and report failure before exceeding gate or
depth budgets \(32G_0/p\) and \(32D_0/p\), respectively.
By Markov's inequality, these cutoffs add failure probability at
most \(p/16\). The counters can be updated using the known gate
and depth costs of each executed layer, with polylogarithmic overhead.

\medskip

The number of gates is now deterministically bounded. Approximate
each rotation sufficiently well that a hybrid bound makes the total
variation distance of the complete stopped computation, including
its execution history, at most \(p/128\). It is enough to use
operator-norm accuracy proportional to \(p\) divided by the gate
cutoff. The expected gate cost changes by at most the gate cutoff
times \(p/128\), hence by at most \(G_0/4\); the analogous
change in expected depth is at most \(D_0/4\). Synthesis of each
rotation and the counters add only polylogarithmic factors, and
the added failure probabilities sum to less than \(p/2\).
This yields the same displayed bounds up to the suppressed
polylogarithmic factors, without assuming a deterministic running
time for the original estimator.
\end{document}